\documentclass[showpacs,amsmath,amssymb,twocolumn,superscriptaddress,notitlepage,preprintnumbers,pra]{revtex4-1}
\usepackage[dvips]{graphicx}
\usepackage{amsmath,amssymb,amsthm,mathrsfs,amsfonts,dsfont}
\usepackage{subfigure, epsfig}
\usepackage{braket}
\usepackage{bm}
\usepackage{enumerate}

\usepackage{qcircuit}
\usepackage{graphicx}
\usepackage{algorithm}
\usepackage{algpseudocode}
\usepackage{algpseudocode}
\usepackage{hyperref}
\hypersetup{colorlinks=true, linkcolor=blue, citecolor=black, urlcolor=blue }
\usepackage{comment}

\newcommand{\tr}{\operatorname{Tr}}
\newcommand{\ST}{S_{\mathrm{T}}}
\newcommand{\SC}{S_{\mathrm{CNOT}}}
\newcommand{\PT}{P_{\mathrm{T}}}
\newcommand{\PC}{P_{\mathrm{CNOT}}}
\newcommand{\selectH}{\operatorname{select}(H)}
\newcommand{\prep}{\textsc{prepare}}
\newcommand{\CNOT}{\operatorname{CNOT}}
\newcommand{\Rz}{{R_z}}
\newcommand{\wt}{\operatorname{wt}}
\newcommand{\wtm}{\operatorname{wt_m}}
\newcommand{\SWAP}{\textsc{swap}}
\newcommand{\RE}{\operatorname{Re}}
\newcommand{\IM}{\operatorname{Im}}
\newcommand{\poly}{\operatorname{poly}}

\usepackage[table]{xcolor}
\usepackage{colortbl}
\usepackage{amsmath}
\usepackage{booktabs}
\usepackage{multirow}
\usepackage{booktabs}
\definecolor{header}{RGB}{200,220,255}
\definecolor{row1}{RGB}{230,220,240}
\definecolor{row2}{RGB}{240,230,245}

\newcommand{\GateC}{\operatorname{C}}

 \DeclareMathOperator*{\argmax}{arg\,max}

\newcommand{\mc}[1]{\mathcal{#1}}
 \DeclareMathOperator*{\E}{\mathbb{E}}

\usepackage{varioref}
\labelformat{equation}{(#1)}

\labelformat{section}{#1}

\labelformat{figure}{#1}

\labelformat{proposition}{#1}

\labelformat{lemma}{#1}

\labelformat{theorem}{#1}

\labelformat{observation}{#1}

\labelformat{definition}{#1}

\labelformat{problem}{#1}

\labelformat{algorithm}{#1}

\labelformat{corollary}{#1}

\newtheorem{theorem}{Theorem}
\newtheorem{observation}{Observation}
\newtheorem{proposition}{Proposition}
\newtheorem{lemma}{Lemma}

\newtheorem{corollary}{Corollary}

\newtheorem{problem}{Problem}

 \newcommand{\sun}[1]{\textcolor{black}{#1}}
 \newcommand{\red}[1]{\textcolor{black}{#1}}
\newcommand{\revise}{\textcolor{black}} 

\begin{document}





\title{
High-precision and low-depth quantum algorithm design for eigenstate problems
}



\date{\today}
\author{Jinzhao Sun}
\email{jinzhao.sun.phys@gmail.com}
\affiliation{Blackett Laboratory, Imperial College London, London SW7 2AZ, UK}
\affiliation{Computer Laboratory, University of Cambridge, Cambridge CB3 0FD, UK}
\affiliation{School of Physical and Chemical Sciences, Queen Mary University of London, London E1 4NS, UK }

\author{Pei Zeng}
\email{qubitpei@gmail.com}
\affiliation{ Pritzker School of Molecular Engineering, The University of Chicago, Illinois 60637, USA }

\author{Tom Gur}
\email{tom.gur@cl.cam.ac.uk}
\affiliation{Computer Laboratory, University of Cambridge, Cambridge CB3 0FD, UK}

\author{M. S. Kim}
\email{m.kim@imperial.ac.uk}
\affiliation{Blackett Laboratory, Imperial College London, London SW7 2AZ, UK}

\begin{abstract}

Estimating the eigenstate properties of quantum systems is a long-standing, challenging problem for both classical and quantum computing. Existing universal quantum algorithms typically rely on ideal and efficient query models (e.g. time evolution operator or block encoding of the Hamiltonian), which, however, become suboptimal for actual implementation at the quantum circuit level. Here, we present a full-stack design of quantum algorithms for estimating the eigenenergy and eigenstate properties, which can achieve high precision and good scaling with system size. \revise{The gate complexity per circuit} for estimating generic Hamiltonians' eigenstate properties is $ \tilde{\mathcal{O}}(\log \varepsilon^{-1})$, which has a logarithmic dependence on the inverse precision $\varepsilon$. For lattice Hamiltonians, the circuit depth of our design achieves near-optimal system-size scaling, even with local qubit connectivity. Our full-stack algorithm has low overhead in circuit compilation, which thus results in a small actual gate count  (\textsc{cnot} and non-Clifford gates) for lattice and molecular problems compared to advanced eigenstate algorithms. The algorithm is implemented on IBM quantum devices using up to 2,000 two-qubit gates and 20,000 single-qubit gates, and achieves high-precision eigenenergy estimation for Heisenberg-type Hamiltonians, demonstrating its noise robustness.
\end{abstract}

\maketitle

\let\oldaddcontentsline\addcontentsline
\renewcommand{\addcontentsline}[3]{}

\section{Introduction}

Estimating the properties of the ground and excited states of quantum many-body systems is a long-standing problem of fundamental interest, which has applications in condensed matter physics, quantum chemistry and material science~\cite{bharti2021noisy,bauer2020quantum,dalzell2023quantum}. 
Despite its quantum hardness both in  theoretical complexity~\cite{kempe2006complexity} and empirical numerical evidence~\cite{lee2023evaluating}, finding eigenstates of many-body Hamiltonians remains a central goal, driving ongoing exploration of quantum algorithms, from quantum phase estimation (QPE)~\cite{kitaev1995quantum,higgins2007entanglement,knill2007optimal,rall2021faster,meister2022resource,martyn2021grand,ding2023even,hejazi2024better} to spectral filter algorithms~\cite{keen2021quantum,chakraborty2024implementing,yang2021accelerated,wan2021randomized,keen2021quantum,wang2023faster,ding2024quantum,wang2023quantum,lin2021heisenberg,zeng2021universal,lu2021algorithms,zhang2022computing,huo2021shallow,wang2023quantum,ding2024quantum,he2022quantum,wang2023faster,wang2024qubit,lin2020near,lin2020optimal,an2023linear,kiss2025early}, dissipation-based algorithms~\cite{cubitt2023dissipative,ding2023single,chen2023local} and others~\cite{huggins2022unbiasing,feng2024escaping,kovalsky2023self,motta2020determining,hejazi2023adiabatic}.
In particular, with rapid development of quantum hardware and error correction, there is increasing interest in designing quantum algorithms considering the feature of early fault-tolerant quantum computing (FTQC) devices~\cite{lin2021heisenberg,katabarwa2023early}, where minimising controlled operations and circuit depth is essential. This constraint also applies to noisy intermediate-scale quantum (NISQ) devices.
In this context, it is desirable to design quantum algorithms that satisfy the constraints of fewer qubits, low circuit depth, and restricted qubit connectivity, which are often interrelated when compiling nonlocal controlled gates into local ones~\cite{devulapalli2022quantum}.



Considering the above hardware constraints~\cite{lin2021heisenberg}, spectral filter based methods are good candidates for effectively finding the ground state~\cite{wang2023faster,lin2021heisenberg,zeng2021universal,lu2021algorithms,zhang2022computing,huo2021shallow,he2022quantum,wang2023quantum,ding2024quantum,ding2023even,wang2023faster}, which can achieve good asymptotic query complexity under the assumptions of nonvanishing initial overlap and energy gap.
These include algorithms for ground-state energy estimation~\cite {lin2021heisenberg} and property estimation~\cite{zeng2021universal,zhang2022computing}, as well as subsequent developments~\cite{wang2023faster,ding2023even,ding2024quantum,wang2023quantum,chakraborty2024implementing}.
Nonetheless, these algorithms are typically based on perfect and efficient queries to either the real-time evolution operator $e^{-iHt}$ which is the case for most early FTQC algorithms~\cite{hejazi2024better,wang2023faster,zeng2021universal,zhang2022computing,lin2021heisenberg,zeng2021universal,lu2021algorithms,zhang2022computing,huo2021shallow,he2022quantum,wang2023quantum,ding2024quantum,ding2023even,wang2023faster}, or the block encoding of the Hamiltonian $H$~\cite{gilyen2019quantum,lin2020near}. Their favourable scaling properties no longer exist when they are compiled into elementary gates in the quantum circuit. 
For example, implementing real-time evolution via Trotterisation as a subroutine will eliminate the advantage of logarithmic precision scaling in eigenstate property estimation. On the other hand, realisation via block-encoding will lose the good system-size scaling; moreover, it requires many ancillary qubits and nonlocal controlled gates, which violates the spirit of early FTQC.
Moreover, the system-size dependence of the algorithms is rarely discussed in existing works, as it highly depends on the circuit-level implementation as well as the qubit connectivity of quantum devices. These algorithms thus become suboptimal at the quantum circuit level.
An important question is, when considering the gate complexity and qubit connectivity in NISQ and early FTQC applications, how to design high-precision and low-depth algorithms.


In this work, we present a full-stack quantum algorithm for eigenstate property and eigenenergy estimation, based on randomised composite linear-combination-of-unitaries (LCU) formulae. 
The maximum gate complexity for each circuit at a single run is shown to be $\mc O ( \poly\log (\varepsilon^{-1}) )$, outperforming the QPE-based method and matching the result by quantum signal processing (QSP)~\cite{lin2020near,gilyen2019quantum}. This precision scaling may not be achieved by coherent methods relying on a coherent implementation of real-time evolution, such as a variant of QSP~\cite{lin2020near} by quantum eigenvalue transformation of unitaries (QETU)~\cite{dong2022ground}.
The second advantage of our method is its low circuit depth for various physical problems with conserved symmetries. For lattice Hamiltonians, our method achieves near-optimal system-size \revise{gate complexity per circuit} even with nearest-neighbour qubit connectivity. Specifically, the depth complexity $d = \mathcal{O} (n^{o(1)})$ is nearly independent of $n$, besides the implicit dependence from the energy gap, which outperforms the other strategies. Here the small-$o$ notation $ o(1)$ indicates approaching to zero asymptotically. To accomplish this, we design new composite LCU circuits that maintain the system's symmetries.  Thanks to the integrated feature of the full-stack design, we are able to achieve near-optimal precision and system-size scaling, even when restricted to nearest-neighbour architecture. 
A comparison with representative advanced eigenstate algorithms is displayed in~\autoref{table:ObservableComp}.

An important feature of our full-stack algorithm is the low overhead in circuit compilation, resulting in a small actual gate count in each coherent circuit run.  We present the resource requirements, including the \textsc{cnot} and non-Clifford gates, for representative physical models in condensed matter and chemistry. Prior resource estimates~\cite{kim2022fault,reiher2017elucidating,goings2022reliably,wecker2014gate,campbell2019random,kim2022fault,von2021quantum} are mostly based on QPE, which is \red{further concatenated with either} Trotterisation~\cite{kivlichan2020improved,campbell2021early,hejazi2024better} or qubitized quantum walks~(QW)~\cite{babbush2018encoding,lee2021even,berry2025rapid}. For lattice models and quantum chemistry problems,  our method requires fewer gates compared to other advanced methods. In particular, the \textsc{cnot} gate cost for a 20-site Heisenberg model is on the order of $ 10^4$ while the T gate cost is about $ 10^6$. This makes our approach particularly suitable for NISQ and early FTQC applications.
Our work also provides a useful toolbox for compiling state-of-the-art quantum algorithms into elementary gates. 
The algorithmic efficiency and inherent noise resilience of our approach enabled a high-precision implementation on IBM quantum devices with up to 2,000 two-qubit gates and 20,000 single-qubit gates. We have tested several Heisenberg-type Hamiltonians with coupling strength $J$ and achieved ground-state energy estimation errors of about $0.01J$.

\begin{table}[t!]
\centering
\setlength{\tabcolsep}{0.15pt} 
\renewcommand{\arraystretch}{1.2} 
\begin{tabular}{c c c c } 
\hline
\toprule
\rowcolor{header}
Methods              & Gate (Generic model)   & Depth (Lattice)    \\
\rowcolor{row2}
\hline
\hline
This work            &  $ {\mathcal{O}}(L    
\boldsymbol{\mathrm{log} \varepsilon^{-1} }  )$ & $\tilde {\mathcal{O}}( \boldsymbol{ n^{o(1)} }     \varepsilon^{-o(1)})$      \\
\rowcolor{row1}
    (zeroth-order)          &  $\mathcal{O}(  \lambda^2 \log^2 \varepsilon^{-1})$  & $\mathcal{O}(n^2  \log^2 \varepsilon^{-1})$      \\
    \hline
\rowcolor{row2}
QPE + Trotter    & $ {\mathcal{O}}( L    \varepsilon^{-(1 + o(1))} )$  & 
 $ {\mathcal{O}}( n^{1+o(1)}   \varepsilon^{-(1 + o(1))} )$  \\
 \rowcolor{row1}
QPE + QW~\cite{babbush2018encoding,lee2021even}   & $ \tilde{\mathcal{O}}( L   \varepsilon^{-1})$ & $ {\mathcal{O}}( n^2  \varepsilon^{-1})$      \\
\rowcolor{row2}
QSP~\cite{lin2020near} & $ {\mathcal{O}}( L  \log\varepsilon^{-1})$ & $ {\mathcal{O}}( n^2   \log\varepsilon^{-1})$  \\
\rowcolor{row1}
QETU~\cite{dong2022ground} & $\tilde {\mathcal{O}}( L   \varepsilon^{-o(1) })$ & $ {\mathcal{O}}( n \varepsilon^{-o(1)} )$   \\
\hline
\end{tabular}
\caption{\textbf{Comparison of advanced methods for eigenstate property estimation (\autoref{problem:observ}) on parameters of the systems ($L$ and $n$) and target precision $\varepsilon$}. 
The second column displays the \revise{maximum gate complexity in a single circuit instance} for generic Hamiltonians with $L$ terms. The third column displays the depth complexity for 1D Heisenberg Hamiltonians with qubit connectivity restricted to a nearest-neighbour architecture. The big-$O$ and small-$o$ notations are used. For example,  $\tilde{O}$ denotes the complexity up to polylogarithmic factors as used in \cite{lin2020near}.  
Different methods have similar dependence on  $\Delta$, which is $\Delta^{-(1+o(1))}$, and hence not shown in this table. 
Table III and Table IV in Supplementary Sec. I present more detailed asymptotic scaling analyses for the eigenstate properties and energy estimation, respectively, which also include the dependence on other parameters $\lambda$ and $\Delta$.
Note that our zeroth-order design with $k=0$ is similar to \cite{wan2021randomized} though our ancilla-free scheme preserves the advantage in circuit depth and the sampling procedure is simpler. When $\lambda$ scales smaller than $O(\sqrt{L})$, the zeroth-order case may be advantageous in Hamiltonian-parameter scaling.  
}
\label{table:ObservableComp}
\end{table}

\begin{figure*}[t!]
\centering
\includegraphics[width=1\textwidth]{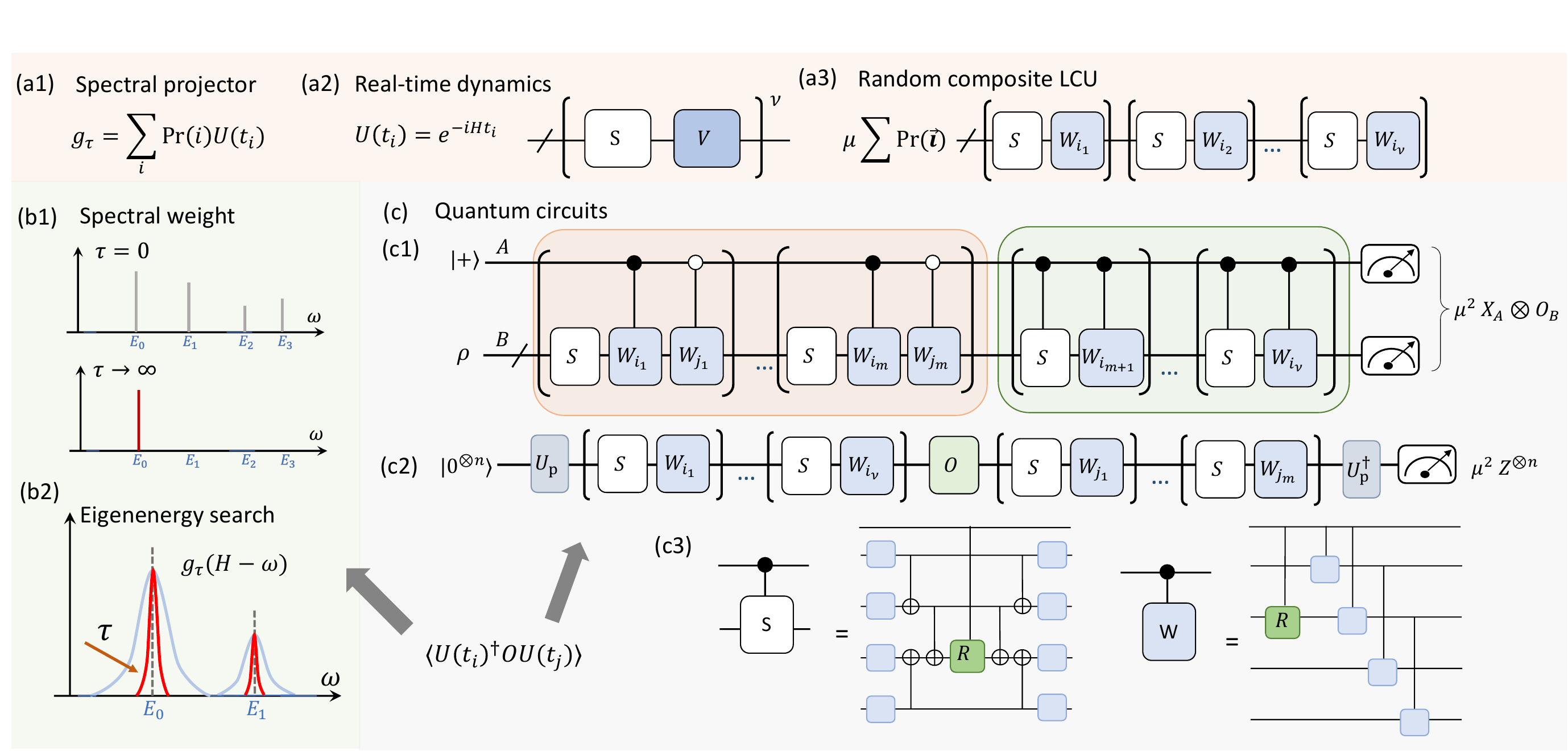}
\caption{ \textbf{Workflow of the eigenstate property and eigenenergy estimation algorithm, from formula construction to circuit implementation.} (\textbf{A}) The hierarchical structure and components of randomised composite LCU formulae. (a1) The spectral filter is decomposed into a combination of unitaries $U(t_i) = e^{-iHt_i}$, as described by \autoref{eq:randomLCU_main}. (a2) The real-time evolution is divided into $\nu$ segments. Each segment with time duration $t_i/\nu$ is realised by the Trotter formula $S$ and the Trotter remainder $V$ which is further decomposed into an LCU form in \autoref{eq:TrotterError_LCU_V}, resulting in an overall composite LCU form. 
(a3) The overall explicit composite LCU form of $g_{\tau}$. The summands of the circuit are shown here; each circuit instance is sampled and described by \autoref{eq:U_ti_LCU} with sampled $\vec{{i}}$. Generally, $W_{i_q} $  is composed of a Pauli rotation  $R = \exp(-i \theta_{i_q} P_{i_q})$ and a tensor product of single-qubit Pauli operators, where $P_{i_q}$ is drawn from the specified probability distribution. The realisation of controlled-$W$ operation is shown in (c3). In the ancilla-free measurement scheme illustrated in (c2), $W_{i_q}$ is chosen to be a symmetry-conserved operator.
(\textbf{B})
Illustration of spectral filtering for effective eigenstate preparation and eigenenergy search.
(b1) The spectral weight of the effective state after applying the spectral filter $g_{\tau}(H-E_j)$. The $y$ axis represents the spectral weight $\braket{u_i| g_{\tau}(H-E_0)|\psi_0}$, where $|u_i\rangle$ is the eigenstate with energy $E_i$. By increasing the imaginary time $\tau$, only the spectral weight on the ground state will be preserved while the contributions from the other eigenstates are suppressed exponentially with $\tau$.
(b2) The eigenenergy search using $D_{\tau}(\omega) = \braket{\psi_0 | g_{\tau}^2 (H-\omega)|\psi_0}$ where $D_{\tau}(\omega) $ is defined in \autoref{eq:OND_main}. By increasing the imaginary time $\tau$, the peaks become sharper (from the blue line to the red line) and the true eigenenergies can be resolved by finding the peaks.
The key quantity in (b1) and (b2) is $\braket{U^{\dagger} (t_j) O U(t_i)}$ in \autoref{eq:observable_dynamics_main}. The sampled instance in $\braket{U^{\dagger} (t_j) O U(t_i)}$ can be measured using the circuit in (\textbf{C}).
(\textbf{C}) Quantum circuit implementation.
(c1) Circuit for general Hamiltonians. 
Two extreme cases: Case I ($t_j = 0$) reduces to the one in the green box; Case II ($t_i = t_j$)  reduces to the one in the orange box.
Cases with $t_j < t_i$ lie between and can be implemented with (c1).
(c2) Ancilla-free measurement scheme for symmetry-conserved cases. $U_p$ is the unitary operator for state preparation $\ket{\psi_0} =  U_p \ket{0^{\otimes n}} $ or $\frac{1}{\sqrt{2}}(\ket{\psi_{\textrm{Ref}}} + \ket{\psi_0}) =  U_p \ket{0^{\otimes n}} $  as discussed in Methods.
(c3) Circuit compilation for controlled $S $ and $W$, which can both be efficiently compiled into single-qubit Pauli rotation gates (green box), \textsc{cnot} gates and Pauli gates (blue box).  
}
\label{fig:cartoon}
\end{figure*}

\section{Results}



\noindent
\textbf{Randomised composite linear-combination-of-unitaries formulae}

\noindent
Here, we introduce the framework of eigenenergy and eigenstate property estimation with the randomised composite linear-combination-of-unitaries (LCU) formulae. 
Let us start by formulating the problems and introducing the notations and assumptions used throughout this work.
Consider an $n$-qubit gapped quantum system whose Hamiltonian~$H$ has a Pauli decomposition $H = \sum_{l=1}^L \alpha_l P_l := \lambda \sum_{l=1}^L \tilde{\alpha}_l P_l$. Here $P_l$ is a Pauli operator, $\lambda := \sum_l |\alpha_l| $, and $\tilde{\alpha} := \alpha_l/\lambda$. The eigenstate $\ket{u_j}$ and the corresponding eigenenergy $E_j$ of the Hamiltonian satisfy,
$
    H \ket{u_j} = E_j \ket{u_j}.
$
The tasks concerned in this work are (1) to estimate the eigenenergy $E_j$, and (2) to estimate an eigenstate property, characterised by an observable expectation on the target eigenstate  $\braket{u_j | O | u_j}$. 
The assumptions as commonly used in QSP and other spectral filter methods are the following: (1) a good initial state $\ket{\psi_0}$ which has a nonvanishing overlap with the target eigenstate, $\eta:= |\braket{\psi_0|u_j}|^2 = \Omega(1/\poly(n))$; (2) a nonvanishing energy gap $\Delta_j:= \min (E_{j+1}-E_j, E_j - E_{j-1} )$. 
The two eigenstate problems considered in this work are stated below.

\begin{problem}[Eigenstate property estimation]
\label{problem:observ}

Suppose the observable has a Pauli decomposition as $O = \sum_{l=1}^{L_o} o_l P_l$ with Pauli operators $P_l$ and positive coefficients $o_l$, and $\|O\|_1 := \sum_{l=1}^{L_o} |o_l|$. 
Given an initial state $\ket{\psi_0}$, the aim is to find an estimator $\hat{v}$, such that it is close to  $\braket{u_j | O | u_j}$ with  probability  at least $1-\vartheta$, i.e., $\Pr( |\hat{v} - \braket{u_j | O | u_j}| \leq \varepsilon) \geq 1-\vartheta$. 

\end{problem}
 
\begin{problem}[Eigenenergy estimation]
\label{problem:eigenenergy}
The aim is to find an eigenenergy estimator $\hat{E_j}$, such that it is $\kappa$-close to  $E_j$ with probability $1-\vartheta$, i.e., $\Pr( |\hat{E_j} - E_j| \leq \kappa) > 1-\vartheta$. 

\end{problem}

The design of the eigenstate property and eigenenergy estimation algorithm is summarised in \autoref{alg1_main} and illustrated in \autoref{fig:cartoon}.
To access the physical properties of eigenstates, a natural idea is to apply a spectral filter to the initial state that projects out the contributions from the other unwanted eigenstates, as illustrated in \autoref{fig:cartoon}(b1). 
While the spectral filter $g$ is non-unitary by construction, we can effectively realise it using LCU techniques, either by coherent \cite{Ge19,keen2021quantum} or random-sampling approaches~\cite{zeng2021universal,lu2021algorithms,zhang2022computing,huo2021shallow,he2022quantum,wang2023quantum,yang2021accelerated}. The overall idea is that at a higher level, the spectral filter is decomposed into a linear combination of unitaries $U(t_i) = e^{-i H t_i}$, as shown in \autoref{fig:cartoon}(a1).
We further realise each $U(t_i)$ by another random LCU formula. Specifically, the evolution is divided into $\nu$ segments, with each segment comprising both a Trotter formula term $S$ and a Trotter remainder term $V$.
Overall, it forms a composite LCU formula as illustrated in \autoref{fig:cartoon}(a3), which involves the summation and product of individual LCU components. 
The hierarchy of different LCU components is illustrated in \autoref{fig:section_connection} in Methods, where the error propagation will be analysed in order to prove the main theorems.
The quantum circuit realisation for randomised composite LCU is shown in \autoref{fig:cartoon}(c), including both the one-ancilla and ancilla-free schemes. The Hamiltonian-specific circuit compilation will be discussed in Methods \autoref{sec:gate_complexity_analysis}. 

The spectral filter $g$ is a non-unitary operator defined on the $n$-qubit system, which is usually a function of the target Hamiltonian $H$.
Choices for spectral filters include the imaginary-time evolution operator $g_{\tau}(H) = e^{-\tau H}$ or the Gaussian operator $g_{\tau}(H) = e^{-\tau^2 H^2}$. 
For example, by applying $g_{\tau}(H) = e^{-\tau |H-E_0|}$ to an initial state $\ket{\psi_0}$, the unormalised state becomes $\ket{\psi(\tau)} = c(\tau) \sum_i c_i e^{-|E_i - E_0| \tau} \ket{u_i}$ where $c(\tau)$ is the normalisation factor. Provided the assumptions in \autoref{problem:observ}, i.e., a nonvanishing energy gap and a nonvanishing $c_0 = \eta^{-1/2}$, the spectral weight of the unwanted excited eigenstates is exponentially suppressed with increasing $\tau$. In the large $\tau$ limit, the state becomes the ground state $\lim_{\tau  \rightarrow \infty} |\psi(\tau) \rangle = c_0 \ket{u_0}$. The procedure for obtaining the excited states is similar,  as illustrated in \autoref{fig:cartoon}(b1).
Note that the spectral filter method has been well-established in the existing literature, in particular \cite{lin2020near,zeng2021universal,zhang2022computing,wang2023quantum}. In this work, we introduce a full-stack approach to eigenstate property and energy estimation by randomised composite LCU decomposition of a Gaussian operator into elementary gates. For simplicity, we refer to this full-stack randomised LCU approach as our method, abbreviated as RLCU.

Now, let us discuss the construction of RLCU. A $(\mu,\varepsilon)$-randomised LCU formula, following the convention in the LCU formula~\cite{childs2012Hamiltonian,zeng2022simple}, of a general operator $g$ is defined to be
\begin{equation} 
\label{eq:randomLCU_main}
\tilde g = \mu_g \sum_{i} \Pr(i) U_i,
\end{equation}
such that the spectral norm distance $\|g-\tilde{g}\|\leq \varepsilon$, \sun{as illustrated in \autoref{fig:cartoon}(a1)}. Here, $\mu>0$ is the normalisation factor, $\Pr(i)$ is a probability distribution associated with an instance specified by $i$, and $\{U_i\}_i$ is a group of unitaries.
\autoref{eq:randomLCU_main} can be extended to a continuous form
\begin{equation}
    \tilde g = \mu_g \int_{-x_c}^{+x_c} dx p(x) U(x)
    \label{eq:randomLCU_continous_main}
\end{equation}
where we require that $p(x)~ (x \in [-x_c, x_c])$ is a well-defined probability distribution. 
A natural choice for $U_i$ is the real-time evolution $U_{i} = U(t_i) := e^{-iHt_i} $ with time length $t_i$ because quantum systems in nature evolve governed by the Hamiltonian. 
Then,  one can decompose a non-unitary, Gaussian spectral filter on the basis of real-time evolution by setting the distribution $p(t) = \frac{1}{\sqrt{2 \pi}} e^{-t^2/4}$. Although the Gaussian filter is represented in an integral form, it has a well-defined probability distribution and can be well-characterised by \autoref{eq:randomLCU_main}, which will be discussed in Supplementary~Sec. II. 


While most early FTQC algorithms assume perfect implementation of real-time evolution,  $U(t_i)$ in general is not directly implementable at the quantum circuit level.
An established way to implement $U(t_i)$ without involving other oracles is through Trotterisation~\cite{childs2018toward}: when the Hamiltonian can be decomposed into Pauli operators, each $U(t_i)$ can be implemented using elementary quantum gates without additional qubit overhead.
However, the issue with the Trotter method is that the remainder of a $k$th Trotter formula (i.e., Trotter error) is non-negligible, which is polynomial in the order of $k$~\cite{childs2021theory}. 
As a result, any ground-state property estimation protocol based on Trotterisation inevitably scales polynomially with the inverse of the target accuracy and  loses the advantage of achieving high precision. 
Below we show how to preserve the logarithmic dependence on the precision.


We consider dividing the time evolution operator $U(t_i)$ into $\nu$ segments, which can be written as $U( t_i ) = \left( S(\delta t_i) V(\delta t_i) \right)^{\nu}$ with $\delta t_i = t_i / \nu$. Here $S(\delta t_i)$ is a $2k$th-order Trotter formula and $V(\delta t_i)$ is the corresponding Trotter error within duration $\delta t_i$. Hereafter, the order $2k$ is omitted when there is no ambiguity. To preserve the high precision property,  we choose to implement the Trotter remainder $V$ as well, as opposed to implementing $S$ only in conventional Trotter methods.
By doing so, we can implement the spectral filter with high precision.
\revise{
Note that our algorithmic design supports the integration of advanced Hamiltonian simulation algorithms (e.g.~\cite{rendon2024improved,watson2025exponentially}).  
The benefits of using Trotter error compensation and comparison with other methods are discussed in Supplementary Sec. I.}

The spectral filter can now be formally rewritten as
\begin{equation}
    \tilde g = \mu_g \sum_i \Pr(i) \left(V(\delta t_i) S(\delta t_i) \right) ^\nu
    \label{eq:spectral_projector_decomp_VS_main}
\end{equation}
which consists of other LCU formulae in it.
Here, $\nu$ is chosen such that the approximation error is sufficiently small.  
We could see that \autoref{eq:spectral_projector_decomp_VS_main} is a modified version of the original LCU formula given by \autoref{eq:randomLCU_main}, which  involves hierarchical formulations, as illustrated in \autoref{fig:cartoon}(a). Specifically, it adopts a composite structure that integrates both the products of individual LCU components and the summations of LCU components. 
More concretely, suppose $ g_1$ takes a  $(\mu_1,\varepsilon_1)$-LCU formula of $g$, $g_1 = \mu_1 \sum_i \Pr(i) U(t_i)$ in \autoref{eq:randomLCU_main}. Each summand $U(t_i)$ is divided into $\nu$ segments, with each segment $U(\delta t_i)$ taking a   $(\mu_2,\varepsilon_2)$-LCU formula, specifically, $
\tilde U(\delta t_i) = \mu_2 \sum_{i_{\ell}} \Pr(i_{\ell}) P_{i_{\ell}}
$.	 
To analyse the properties of the composite form of LCU formulae, we introduce the following proposition.

\begin{proposition}[Composite LCU]
\label{prop:composite_LCU}
The formula 
$ 
g_2 = \mu_1\sum_i \Pr(i) (\tilde{U}(  t_i/\nu))^{\nu}
$
is a $(\mu,\varepsilon)$-LCU formula of $g$, with $\mu  := \mu_1  \mu_2^{\nu}$, and $\varepsilon \leq \varepsilon_1 +   \mu_1 
 \mu  \varepsilon_2$.
\end{proposition}


Leveraging \autoref{prop:composite_LCU}, we could calculate the error in the composite LCU form provided the error of each individual LCU approximation. 
This work will analyse the propagation of errors and the change of normalisation factors within this composite LCU framework, outlined in \autoref{fig:section_connection} in Methods. The proof is shown in Supplementary~Sec. II.

We shall briefly outline the implementation procedure as illustrated in \autoref{fig:cartoon}(a2,a3). First, the Trotter remainder $V(\delta t_i)$ is decomposed into easy-to-implement unitary operators \begin{equation}
\label{eq:TrotterError_LCU_V}
	V = \mu \sum_j \Pr(j) W_j
\end{equation}
where $\Pr(j) $ is the probability associated with the unitary operator $W_j$,   $\mu$ is the associated normalisation factor and $\delta t_i$ is omitted as this equation holds in general. Importantly, the decomposition in \autoref{eq:TrotterError_LCU_V} can be explicitly derived by $V = U S^{\dagger}$ using the Taylor expansion.
By doing so, we can compensate for the simulation error by sampling $W_j$ according to $\Pr(j)$ in \autoref{eq:TrotterError_LCU_V}.
This effectively realises a high-precision spectral filter, with the individual terms illustrated in \autoref{fig:cartoon}(a3).
In general cases, the elementary operator $W_j$ in \autoref{eq:TrotterError_LCU_V} can be chosen as Pauli operators or the exponentiation of Pauli operators, as commonly used in Trotter and LCU methods~\cite{childs2018toward,zeng2022simple}. The maximum gate count in the whole block of $W_{i_q}$ is saturated as $\wtm(H) + n$, where $\wtm{H}$ is the largest weight of the Hamiltonian terms (see methods). 

For systems with certain symmetry, the ancilla may not be needed. To enable ancilla-free measurement, we choose $W_j$ to be a symmetry-conserved operator, because choosing it as a Pauli operator will break the symmetry in the composite implementation. Specifically, to conserve the symmetry of either particle number or total spins, $W_j$ is chosen to be either the \textsc{swap} operator, Pauli-$Z$ operator, or their exponentiation,  detailed shortly after.

 
\vspace{8pt}
\noindent
\textbf{Quantum circuit realisation and depth analysis}


\noindent
The key quantity involved in the algorithm outlined in \autoref{fig:cartoon}   is $ 
	\braket{U^{\dagger} (t_j) O U(t_i)},
$ where $U({t_i}) = e^{-iHt_i}$, and $O = I$ for eigenenergy estimation and $O $ being an Pauli operators for observable estimation.
As   introduced above, the unitary $U(t_i)$ is implemented by a composite LCU formula with an explicit form as
\begin{equation}
    U(t_i) = \mu \sum_{\vec i} \Pr (\vec i)     \prod_{q =1}^{\nu} W_{i_q} S 
     \label{eq:U_ti_LCU}
\end{equation}
where $\vec i = (i_1, i_2, ..., i_{\nu})$ and $\nu$ being the segment number.
The observable expectation is thus given by      
\begin{equation}
\label{eq:observable_dynamics_main}
 \braket{U^{\dagger} (t_j) O U(t_i)} = \mu^2    \sum_{\vec i, \vec j} \Pr(\vec i)  \Pr(\vec j) \braket{\vec {\mathbf{j}} | O | \vec {\mathbf{i}}}
\end{equation}
where we denote $\ket{\vec{\mathbf{ i}}} :=   \prod_{q =1}^{\nu} W_{i_q} S  \ket{\psi_0}$.

For the general case, we can use the Hadamard test circuit to measure $\braket{\vec {\mathbf{j}} | O | \vec {\mathbf{i}}}$, with the circuit depicted in the green box in \autoref{fig:cartoon}(c). Note that there is no control over the shared Trotter term, which save the quantum resources.
For Hamiltonians with certain symmetries, there is no need to use ancilla and controlled unitaries. 
\autoref{fig:cartoon}(c2) shows the quantum circuit implementation for measuring the sampled instance involved in $\braket{\vec {\mathbf{j}} | O | \vec {\mathbf{i}}}$. 
In the ancilla-free measurement scheme, in \autoref{fig:cartoon}(c2), two types of initial states (prepared by $U_p$) are involved. The unitary operator $U_p$ prepares either $\ket{\psi_0} =  U_p \ket{0^{\otimes n}} $ or $\frac{1}{\sqrt{2}}(\ket{\psi_{\textrm{Ref}}} + \ket{\psi_0}) =  U_p \ket{0^{\otimes n}} $, where $\ket{\psi_{\textrm{Ref}}}$ belongs to a different symmetry sector from $\ket{\psi_0}$  and is orthogonal to the initial state (see Methods for details).

To realise the ancilla-free measurement, it is necessary to pair the terms in $S$ and $V$ such that each resulting term preserves the symmetry and does not break it individually.
The issue arises when the Trotter remainder $V$ is expanded in the Pauli basis: applying each individual Pauli term generally breaks the symmetry.
To address this issue, we decompose the Hamiltonian into the basis of \textsc{swap}  and tensor products of Pauli-$Z$ operators, rather than Pauli bases. 
Let us give an example of the 1D Fermi-Hubbard model, which, after the Jordan-Wigner Transformation, takes the form of
$
    H = J_1 \sum_{i} (X_i X_{i+1} + Y_i Y_{i+1}) + J_2 \sum_{i} Z_i Z_{i+1} + h_z \sum_i Z_i 
 $. This can be reformulated as
$
    H = 4 J_1 \sum_i \SWAP_{i,i+1} + (J_2 -J_1) \sum_{i} Z_i Z_{i+1} + h_z \sum_i Z_i
 $ where the identity term in the Hamiltonian is trivial, and will always be removed.
Each individual term in the Hamiltonian is unitary and commutes with the particle number operator.
The only difference lies in the compensation terms, which involve the $\textsc{swap}_{i,i+1}$ operators, Pauli-$Z$ operators, and their corresponding exponentiated forms.
When restricted to a linear nearest-neighbour (NN) architecture, the circuit depth within a segment is $d = \mathcal{O}(1)$, and is shown to be advantageous over other methods.

\begin{figure*}[ht!]
\centering
\includegraphics[width=1\textwidth]{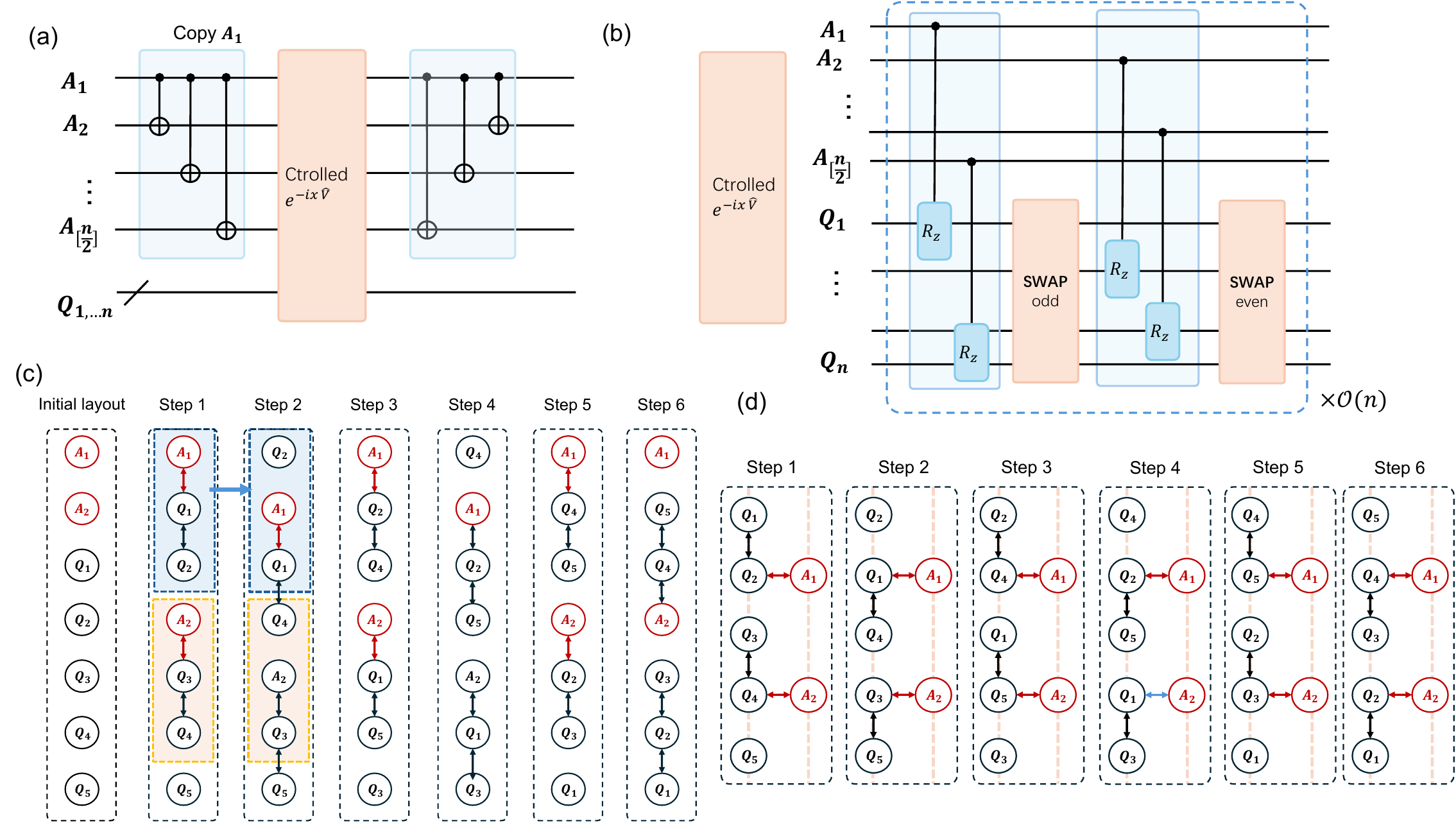}
\caption{ \textbf{
The implementation of the controlled exponentiation of potential terms Ctrl-$e^{-i x \hat{V}}$ on a nearest-neighbour architecture (1D or 2D) using $ [ \frac{n}{2} ]$ ancillas with depth $\mathcal{O}(n)$.}
(\textbf{A}) Copy the classical information on the qubit $A_1$ to $A_2$,...$A_{[n/2]}$, and then undo the copy operation. This circuit is equivalent to using $A_1$ as the single-controlled qubit to control all the other physical qubits $Q_1$ to $Q_n$.
$Q_i$:  $i$th physical qubit (encoding the $n$th spin-orbital). $A_i$: $i$th ancilla.
The copy operation allows all controlled rotations (blue-shaded box in (\textbf{B})) to be implemented using nearest-neighbour gates.
(\textbf{B}) The circuit block for controlled $e^{-i x \hat{V}}$, using nearest-neighbour operations $e^{- i Z_i Z_{i+1}}$ followed by $[n/2]$ \textsc{swap}  operations detailed in \textbf{C}. The circuit block in (\textbf{B}) is repeated $\mathcal{O}(n)$ times. An example is shown in (\textbf{C}). \textsc{cnot} operations are omitted in the sub-figure.
(\textbf{C}) 1D linear architecture. 
The ancillas $A_i$ ($i = 1,2,.., [n/2]$) and physical qubits $Q_i$ ($i = 1,2,.., n$) can be placed in the way shown in Step 1 in $\mathcal{O}(n)$ depth.
The {red} arrow connecting $A_i$ and $Q_{j}$  is used to represent to perform the controlled-$\textrm{R}_z$ rotation, sandwiched by \textsc{cnot} operations on adjacent qubits $Q_j$ and $Q_k$ (connected and illustrated by the black arrow), which realises $e^{- i Z_j Z_{k}}$. 
The black arrow connecting $Q_j$ and $Q_k$ is used to represent performing the corresponding \textsc{cnot} operations in realising $e^{- i Z_j Z_{k}}$, then followed by a \textsc{swap}  operation.  The transformation from Step 1 (the shaded blue and orange boxes) to Step 2 can be realised by 2 \textsc{swap} gates (cyclic \textsc{swap} operation). The rest of the transformation is realised in the same way.
(\textbf{D}) 2D planar architecture.
The qubit connectivity is represented by the orange dashed line.
}
\label{fig:control_V_NN}
\end{figure*}

For electronic problems, the Hamiltonian can also be reformulated with \textsc{swap}  and Pauli-$Z$ operators. The kinetic term under the Jordan-Wigner transformation, for example, $\hat{T}_{ij} = h_{ij} ( \hat{a}_i^{\dagger}   \hat{a}_j 
 + \hat{a}_j^{\dagger}  \hat{a}_i )  = \frac{h_{ij}}{2} (X_i X_j \otimes_{k = i+1}^{j-1} Z_k  + Y_i Y_j \otimes_{k = i+1}^{j-1} Z_k)$ can be reformulated as
$$
 \hat{T}_{ij} = h_{ij} ( 2 \SWAP_{i,j} \otimes_{k = i+1}^{j-1} Z_k - \frac{1}{2} \otimes_{k = i}^{j} Z_k).
 $$
The potential term $\hat{V}_{ij kl}$ can be similarly reformulated such that each individual term commutes with the particle number operator.
Although the Hamiltonian is reformulated in the basis of \textsc{swap}  and Pauli-$Z$ operators, the implementation of the Trotter formulae remains the same as that in the Pauli basis.
The depth complexity for the electronic structure problem in \autoref{eq:molecular_hamil_main} is shown to be $d = \mathcal{O}(n)$ (see Methods). 


\vspace{10pt}

\noindent
\textbf{Asymptotic gate and depth complexity}

\noindent
Our results with randomised composite LCU can nearly match and, in some cases, outperform the previous best methods for gate complexity with respect to $\Delta$, $\varepsilon$, $\lambda$, and $n$, although the sample complexity becomes worse in $\eta$. 
The gate complexity for estimating generic Hamiltonian's eigenstate properties is summarised as follows.

\begin{theorem}[Gate complexity for general cases (Informal)]
\label{thm:observ_estimation_main}

\textbf{Observable estimation} (\autoref{problem:observ}): To achieve the error of observable's expectation on the eigenstate $\ket{u_j}$ within $\varepsilon$ with success probability   $1-\vartheta$, the gate complexity in a single circuit is $\mathcal{O} \left(  (  \Delta^{-1} \ln(  \varepsilon^{-1}))^{1+ \frac{1}{4k+1} }  \right )$ provided $N_s = \mc O\left ( \varepsilon^{-2} \|O\|_1^2 \ln(1/\vartheta) \right)$ samples.

\noindent\textbf{Eigenenergy estimation} (\autoref{problem:eigenenergy}): To achieve the eigenenergy estimation error within $\kappa$, the total gate complexity is $\mathcal{O} \left(    \kappa^{-(1+ \frac{1}{4k+1})}   \ln(1/\vartheta) \right )$  with success probability $1-\vartheta$, approaching to the Heisenberg limit. Alternatively, by using the methods proposed in \cite{wang2023quantum} and \autoref{alg1_main}, the maximum gate complexity in a single circuit can be reduced to $\mathcal{O} \left(  (  \Delta^{-1} \ln(  \kappa^{-1}))^{1+ \frac{1}{4k+1} }  \right )$ at the cost of more samples.

 \end{theorem}

Here, the spectral filter in \autoref{eq:spectral_projector_decomp_VS_main} is constructed using the $2k$th-order Trotter formula as a building block in realising the real-time dynamics in \autoref{fig:cartoon}. 
\red{Note that the spectral filter can be constructed with $k = 0$ (i.e., without the Trotter term $S$) and its gate complexity is also covered by \autoref{thm:observ_estimation_main}. Our zeroth-order design with $k = 0$ is similar to that of \cite{wan2021randomized}, though our sampling procedure is simpler and it shows advantages in depth when qubit connectivity is restricted. 
}

As shown in \autoref{table:ObservableComp}, our method can achieve polylogarithmic dependence on inverse precision, outperforming the QPE-based method and matching the result by QSP. As a variant of QSP, QETU can achieve near-optimal ground state preparation by querying real-time evolution. However, these types of methods intrinsically hinge on a coherent implementation of $e^{-iHt}$, which rules out any random sampling method.  Therefore, it is not straightforward to achieve polylogarithmic dependence on inverse precision by these coherent methods. 
As $\eta$ only appears in sample complexity instead of the gate or depth complexity which is more of a concern in NISQ or early FTQC applications, its dependence is not included in \autoref{table:ObservableComp}. A more detailed description of both the gate and sample complexity can be found in Theorem 3 in Supplementary~Sec. III. The actual \textsc{cnot} and non-Clifford gate counts considering circuit synthesis are shown in \autoref{prop:gate_cost_observ_dynamics} in Methods.





Next, we discuss the depth complexity for various physical Hamiltonians. 
An advantage of our method over QSP-based methods is that the commutator relations can be utilised (when $k \geq 1$), such that it could achieve a better system size dependence. For Heisenberg models, the gate complexity of our method is $\mathcal{O}(n^{1+o(1)} )$. In contrast, the solution given by QSP~\cite{lin2020near} is $\mathcal{O} ( n^{2}   \log\varepsilon^{-1}) $. On the other hand, the implementation of time evolution by Trotter methods will undermine the optimal scaling with respect to $\lambda$, $\Delta$ and $\varepsilon$.
Since controlled operations are necessary in QETU (even for their control-free approach), it becomes suboptimal when qubit connectivity is restricted. To summarise, our  method has the following features:
(1) it has better scaling with respect to $\lambda$, $\Delta$ and $\varepsilon$, though $\eta$ dependence is worse than the other advanced methods. 
(2) For the Heisenberg model,  we shall see that our method based on the $2k$th-order Trotter formula has depth $\mathcal{O}( n^{\frac{2}{4k+1}} \varepsilon^{-\frac{1}{4k+1}})$, which has better system-size $n$ and precision $\varepsilon$ scaling over existing strategies.
{It may be worth noting that one may simultaneously achieve near optimal scaling in both the size and precision as $\mathcal{O}( n^{\frac{2}{4k+1}} \log\varepsilon^{-1})$ if the higher-order commutators in the Trotter error remainder could be compensated.
}
The result for depth complexity is summarised in \autoref{thm:observ_estimate_gate_Lattice_main}. A comparison with other methods is shown in the third column of \autoref{table:ObservableComp}.

\begin{figure*}[t!]
\centering
\includegraphics[width=1\textwidth]{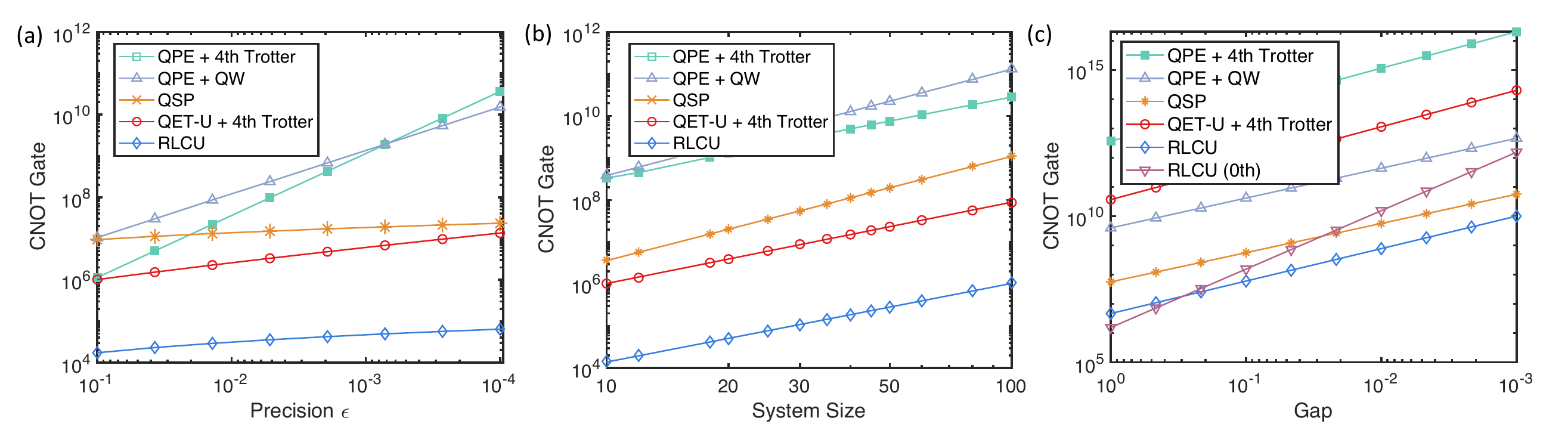}
\caption{ 
\textbf{
Gate count estimates for the eigenstate property estimation tasks for the Heisenberg Hamiltonian and P450 molecule.}
(\textbf{A}) Gate count comparison with different target precision for $20$-site Heisenberg Hamiltonian.
(\textbf{B}) Gate count comparison with increasing system size to achieve precision $0.001$.
The energy gaps are determined through numerical fitting, which agrees well with the results obtained from exact diagonalisation.
(\textbf{C}) Gate count for the P450 with $A$-type active space as a function of the energy gap, which is treated as an independent variable to analyse its effect. 
The pairing orders with both $k = 0$ and $1$ are shown in (\textbf{C}). 
}
\label{fig:SysSize_NC}
\end{figure*}



\begin{theorem}[Gate and depth complexity for lattice Hamiltonians (Informal)]
\label{thm:observ_estimate_gate_Lattice_main}
The gate complexity for estimating  an $n$-qubit Heisenberg Hamiltonian's eigenstate property is $\mathcal{O}(n^{1+\frac{2}{4k+1}}   )$, with circuit depth $\mathcal{O}(n^{\frac{2}{4k+1}}  )$.
\end{theorem} 

The proof idea for \autoref{thm:observ_estimation_main} and \autoref{thm:observ_estimate_gate_Lattice_main} is illustrated in \autoref{fig:section_connection} in Methods. See the formal version of the theorems and the proof in Supplementary~Sec. III.


The result can be extended to the simulation of molecules.  
For electronic problems, the second-quantised electronic-structure Hamiltonian in the plane-wave dual basis, has the form (see \cite{babbush2018encoding,kivlichan2020improved,childs2021theory})
\begin{equation}
H= \hat{T} +\hat{V}  = \sum_{p q} T_{p q} \hat{a}_p^{\dagger} a_q+\sum_p U_p \hat{n}_p+\sum_{p \neq q} V_{p q} \hat{n}_p \hat{n}_q
\label{eq:molecular_hamil_main}
\end{equation}
where $\hat{T}$ and $\hat{V}$ represent the kinetic and potential terms of the fermionic Hamiltonian, respectively,  $\hat a$ and $\hat a^{\dagger}$ are fermionic creation and annihilation operators and $\hat{n}_p$ is the number operator for the corresponding spin-orbital, and  the total number of terms $L = 
\mathcal{O}(n^2)$. 
To estimate the eigenstate property of the Hamiltonian described by \autoref{eq:molecular_hamil_main}, the circuit depth scales as
$\mathcal{O} (n^ {2 + \frac{2}{4k+1}} )$. When restricted to nearest-neighbour (NN) architecture, the circuit depth by QETU scales as
$\Omega (n^ {3 + \frac{1}{2k}} )$. 
More general second-quantised quantum chemistry problems with $L = \mathcal{O}(n^4)$ terms will be discussed in Methods.

As a byproduct, we show that the controlled $e^{-i \theta {H}}$ can be implemented with a linear-depth circuit $d = \mathcal{O}(n)$  comparable to the control-free simulation of electronic problems in \cite{kivlichan2018quantum,babbush2018low}. The result is formalised in \autoref{prop:control_V_NN} with a graphical proof in \autoref{fig:control_V_NN} in Methods. It is directly applicable to a range of quantum algorithms for applications in quantum chemistry, materials and lattice gauge theories, which require controlled unitaries as subroutines.




\vspace{50pt}

\noindent
\textbf{Quantum resource estimation}	









\begin{figure}[t!]
\centering
\includegraphics[width=0.78\linewidth]{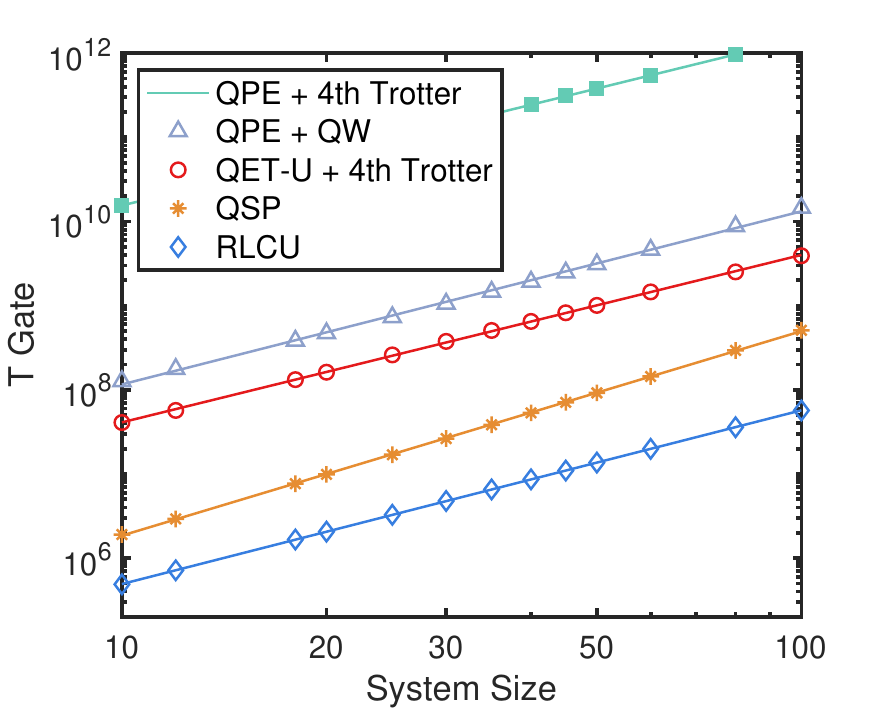}
\caption{ 
\textbf{The scaling of the T gate count with system size for the Heisenberg model.}  The RLCU method involves at most $10^4$ single-qubit Pauli rotation $R_z$ gates and $10^6$ T gates for $20$-qubit Heisenberg model. The circuit synthesis method is detailed in Supplementary Sec. V.
}
\label{fig:Heis_T_main}
\end{figure}
 
\noindent
Here, we compare the resource costs associated with different algorithms listed in Table 1, focusing on \revise{the maximum \textsc{cnot} gates and T gate counts in each circuit} for various physical systems, including lattice models and quantum chemistry problems. 
The QPE method relies on Hamiltonian simulation, which can be realised by the Trotterisation or qubitised quantum walk (QW), with the latter having a better $\varepsilon$ scaling at a cost of a larger overhead. We find that the fourth-order (random) Trotter result performs the best among all the Trotter methods, which is consistent with the result in \cite{childs2018toward}. The quantum circuit is synthesised to \textsc{cnot} gates, single-qubit Clifford gates and non-Clifford gates (including single-qubit $Z$-axis rotation $R_z$ gates and T gates). The elementary operations in the block encoding of $H$ used in QSP are the \textsc{select} and \textsc{prepare} operations~\cite{low2019hamiltonian}. The cost for these two operations is shown in detail in Supplementary~Sec. V, which serves as the basis for analysing the resource requirements of algorithms that query the block encoding of $H$.






We consider the Heisenberg Hamiltonian
\begin{equation}
H = \sum_{i=1}^{n-1}( J_x X_{i}X_{i+1} + J_y Y_{i}Y_{i+1}  + J_z Z_{i}Z_{i+1})  +   H_f.
\label{eq:Hamil_XXZ}
\end{equation}
where the external field is applied $H_f = h_x \sum_{i=1}^n X_i + H_b$ with additional field $H_b$ acting on the boundary. The periodic boundary condition is imposed.
When $h_x = 0$ when an additional field  $H_b = \sqrt{c^2-1}( Z_1 - Z_n )$ is applied, the ferromagnetic Heisenberg Hamiltonian with negative couplings $J_i$  has a constant gap $\Delta(c) = 4(c-1)$  in the infinite size limit (see~\cite{koma1997spectral}). 
Here we consider a more challenging regime with the antiferromagnetic types of couplings with $J_x = J_y = 1$ and $J_z = 2 J_x$ and field $h_x = 0.25$, in which more excited states will emerge and is usually more interesting. Nevertheless, even in this case, we find by numerical fitting that for $n \leq 100$, the energy gap is not very small.
Indeed, the gap can be fitted by a polynomial function with leading order $\Delta = \mathcal{O} (n^{-\frac{1}{2}})$, which agrees quite well with the actual gap at small system sizes, while an exponential decaying function does not agree well, see Supplementary~Sec. VIII for the fitting results. 
\revise{We set the initial state overlap as a constant, as similarly used in \cite{lin2021heisenberg,von2021quantum,ding2023even,wan2021randomized}, whose dependence is analysed in Theorem 5.}

First, we present the gate number  estimates for the above Heisenberg model with different target precisions, as shown in \autoref{fig:SysSize_NC}(a), which validates the high-precision feature of our method. Next, we show the gate count dependence on the system size. 
\autoref{fig:SysSize_NC}(b) clearly shows that our method has a better system size dependence than QSP, both in the gate count and in the asymptotic scaling. The T gate count scaling for the Heisenberg model has a similar behaviour and is shown in \autoref{fig:Heis_T_main}. The large reduction compared to QSP arises from exploiting the commutation relations of Hamiltonian terms. Note that our method only requires rotation gates $R_z (\theta)$ with identical and small angles, and thus further reductions for T gates can be achieved by using the protocols in \cite{campbell2016efficient,duclos2015reducing}.

We further discuss the resource cost for the cytochrome enzyme (P450) molecule with the A-type active space specified in~\cite{goings2022reliably}, as shown in \autoref{fig:SysSize_NC}(c). 
The methods for compiling the single- and double-excitation operators are used to reduce gate overhead  which is more efficient than naively Trotterising each fermionic term.
However, the cost for Trotterisation in every step is not optimised, in which each term is nonlocal and thus contributes substantially to the total gate count. The methods for simulating fermionic Hamiltonians can be directly incorporated~\cite{childs2021theory,kivlichan2020improved,su2021nearly}. Moreover, our result reflects the gate count in the worst-case scenario. The commutation relations among Hamiltonian terms can be naturally exploited to further reduce the computational cost of our method. We leave the improvements about leveraging the properties of molecular Hamiltonians to future work. The T gate counts for the Heisenberg model and P450 molecules are presented in \autoref{fig:Heis_NC_Eps} in Methods.

\vspace{5pt}
\noindent
\textbf{Implementation on IBM devices}

\noindent
We focus on estimating the low-lying eigenenergies of normalised $12$-sites anisotropic Heisenberg Hamiltonians, especially the ground-state energy and the first excited-state energy. The simulations are performed on the IBM processor with the ancilla-free measurement scheme (see \autoref{fig:cartoon}(c2)).  The reference state is set to be orthogonal to the initial antiferromagnetic state.  The eigenenergy estimation algorithm is illustrated in \autoref{fig:cartoon}(b2), where the search by $D_{\tau}(\omega)$ over $\omega$ (based on experimental measurements) only requires classical computation, indicated by the red line in \autoref{fig:IBM}.   
The experimental estimation of the ground-state energy achieves remarkably high accuracy, with an error of about $0.001$ ($0.01J$ with Heisenberg coupling $J$), as shown in the zoom-in region in \autoref{fig:IBM}.  Then, we consider the antiferromagnetic XXZ model and present the search for its low-lying eigenenergies in \autoref{fig:IBM_Methods} in Methods. The experiment involved circuits with up to 2,000 \textsc{cnot} gates and 20,000 single-qubit gates with a relatively long evolution time.
The maximum energy error remains within $0.005$ for both eigenvalues. 
The high precision is attained without relying on error mitigation techniques owing to the intrinsic algorithmic noise robustness.



\begin{figure}[t!]
\centering
\includegraphics[width=0.95\linewidth]{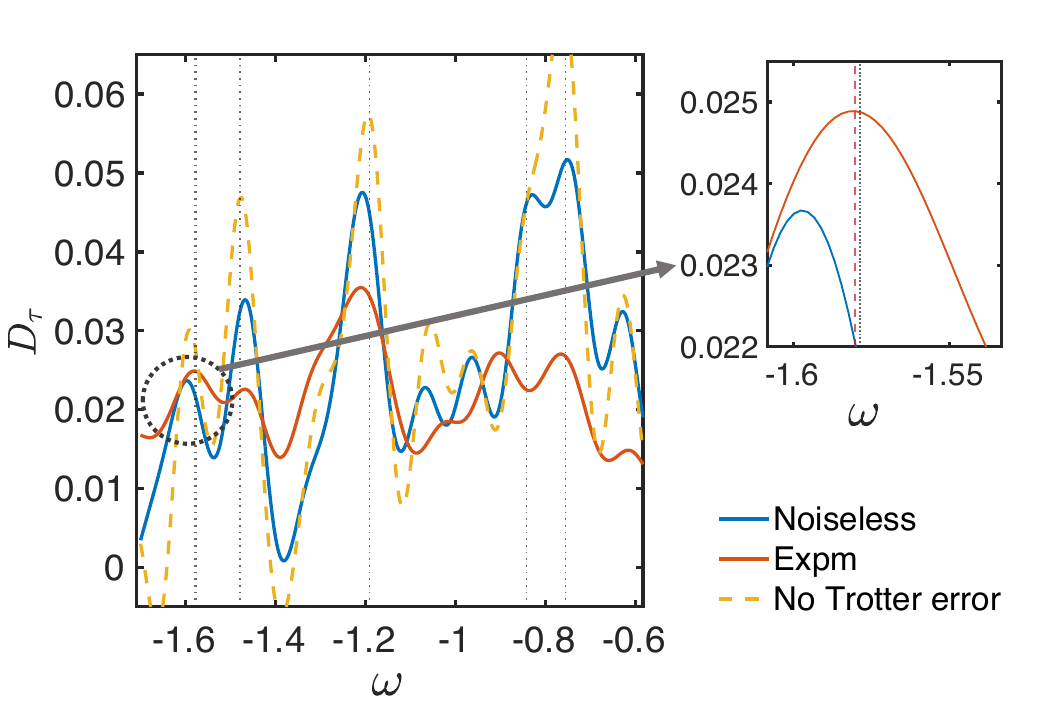}
\caption{
\textbf{Implementation of the RLCU algorithm on IBM quantum devices}.  We consider searching the ground state energy of a 12-qubit normalised anisotropic Heisenberg Hamiltonian \autoref{eq:Hamil_XXZ} without any external field. We present the ideal result with finite $\tau$ and finite cutoff but without any Trotterisation error, represented by the orange dotted lines. We also show the results obtained using the noiseless and noisy Trotterised quantum circuit and the experimental data, denoted by the blue and red lines, respectively. 
For different lines, $D_{\tau} (\omega)$ is computed classically using data points obtained from different setups, including both numerical simulations and experimental measurements.
The right panel is a zoom-in of a smaller range of ground-state energy estimates shown in the left panel. The red dotted line represents the experimentally estimated ground-state energy, which is extremely close to the ideal value shown by the black dotted line, with an error of 0.001. The energy estimation error for the excited state is about 0.005.
}
\label{fig:IBM}
\end{figure}

\section{Discussion}
\noindent
We have provided a full-stack quantum algorithm based on randomised composite LCU for estimating the eigenstate property and eigenenergy of many-body systems. 
While previous works have mostly considered query complexity, our approach can achieve near-optimal precision and system size scaling at the quantum circuit level, even with nearest-neighbour connectivity. We present an ancilla-free strategy by choosing the elementary operators to be symmetry-conserved operators (e.g. \textsc{swap}  and Pauli-$Z$ operators), which is particularly important for implementation on quantum devices, as the restriction of connectivity will incur a large overhead when compiling it into local operations.
Moreover, we show a concrete gate count analysis for various physical models concerning the circuit synthesis, which shows remarkable improvements in both asymptotic scaling and actual gate counts. Our work presents concrete resource estimates for lattice and molecular systems, providing guidance for the computation of physical problems on real devices. 
We would note that it is precisely through adopting this bottom-up quantum circuit design methodology that we can demonstrate deterministic energy estimation in experiments, which achieves higher accuracy and deeper circuit execution~\cite{yoshioka2025krylov,stanisic2022observing,tang2024exploring} without relying on variational ansatz. 

Similar to other spectral filter methods, the random-sampling-based method cannot overcome the fundamental limitations set up by the initial state and the energy gap. Ref.~\cite{lee2023evaluating} numerically investigated its scaling for several many-body examples, which quickly decreases with increasing system size, aligning with the complexity conjecture. Nonetheless, even when provided a sufficiently good initial state, preparing the ground state of  2-local Hamiltonians remains as BQP-hard (known as the guided local Hamiltonian problem)~\cite{gharibian2022dequantizing,gharibian2022improved}. Therefore, it does not rule out the possibility of quantum speedup. In addition, similar to other random-sampling spectral filter methods (\cite{zhang2022computing,lin2021heisenberg}), the sample complexity with respect to $\eta$ is less optimal than that achieved by QSP with amplitude amplification which, nonetheless, also introduces a large circuit compilation cost. On one hand, our algorithm is compatible with initial state preparation methods such as adiabatic evolution and variational methods.
On the other hand, the random sampling method offers non-trivial advantages. In particular, when considering qubit connectivity, the depth complexity is independent of $\eta$ and nearly independent of the system size, which is not achievable by protocols relying on coherent implementations.
Therefore, our approach could be useful for estimating the requirement for early quantum usefulness, where the primary constraint lies in gate count and circuit depth.



This work also provides a user-friendly toolbox for researchers to analyse the individual costs for elementary units in quantum simulation, and thus enables comparison across different eigenstate algorithms with various initial conditions. 
Our framework (decomposing the task into the elementary operations), along with the toolbox for analysing the cost for each elementary operation, is readily useful as a building block for end-to-end resource estimations for a broad class of quantum algorithms.
We also note that recent works consider preparing the ground state by simulation of an open system dynamics described by Lindbladians or a thermal process~\cite{chen2023local}, which may partially overcome the limitations set by the initial state~\cite{lee2023evaluating} though the simulation time~\cite{ding2023single,zhan2025rapid} is longer than the spectral filter method by numerics. As the fundamental building block of these methods  \cite{ding2023single,zhan2025rapid,chen2023local,cubitt2023dissipative} is querying the real-time evolution, the actual resource cost can be analysed similarly with our approach.



\begin{figure*}[ht!]
\centering
\includegraphics[width=1\textwidth]{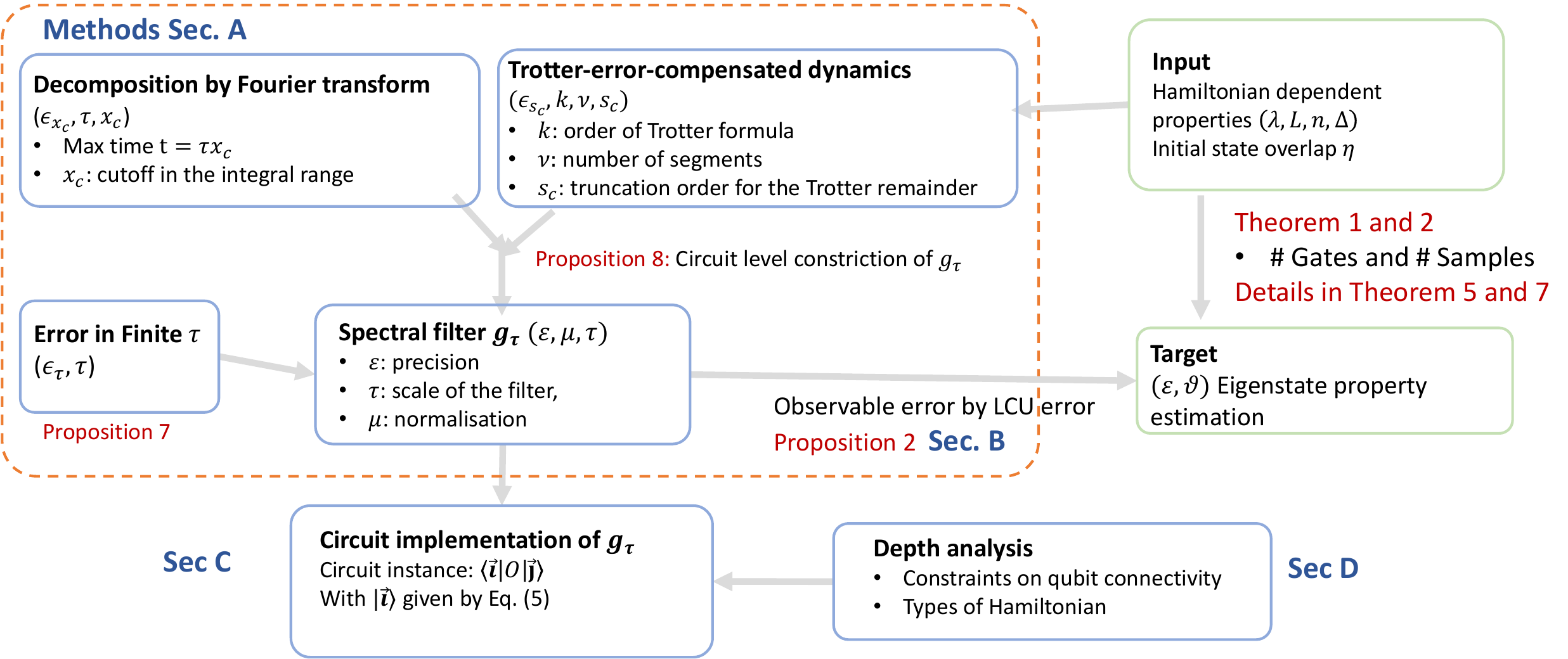}
\caption{
\textbf{Illustration for the hierarchy of various components within the composite LCU.} The figure illustrates the connections between different subsections in Methods. This visualisation aids in understanding how the target problem (\autoref{problem:observ} and \autoref{problem:eigenenergy}) can be achieved, specifically, to achieve $(\varepsilon, \vartheta)$-eigenstate property estimation. Arrows in the figure indicate direct connections between components. The key parameters relevant to the construction of LCU are highlighted.  
}
\label{fig:section_connection}
\end{figure*}

\section{Materials and Methods}
\label{sec:methods}


This section is centred around \autoref{fig:section_connection}, which shows the connections of different subsections in Methods and illustrates the proof idea of \autoref{thm:observ_estimation_main} and \autoref{thm:observ_estimate_gate_Lattice_main}.
In \autoref{sec:methods_LCU}, we show how to construct the composite LCU formula as outlined in \autoref{fig:cartoon} and  \autoref{alg1_main}. The hierarchy of different components in the composite LCU is shown in  \autoref{fig:section_connection}. Then, we will show the methods for eigenenergy estimation and eigenstate property estimation, and provide the circuit depth and gate complexity for the two tasks.
Specifically, \autoref{sec:methods_observable_estimation}   shows how errors in observable expectations can be bounded given the approximation error of the composite LCU formula. \autoref{sec:measurement} shows the circuit realisation of the composite LCU form by presenting the explicit form of each circuit instance. An ancilla-free composite LCU formula and measurement scheme are presented for symmetry-conserved cases.
\autoref{sec:gate_complexity_analysis} analyses the depth and gate complexity for general and Hamiltonian-specific cases.



\subsection{Composite LCU for the spectral filter}
\label{sec:methods_LCU}

{
Now, we show how to construct the composite LCU form of the spectral filter.
We mainly consider the Gaussian filter $g_{\tau}(H) = g(\tau H)$, but present the construction of Hermitian $g$ in a general form following the convention in \cite{zeng2021universal}. Therefore, the results can be readily applied to the investigation of other functions, such as a product of Lorentz and Gaussian functions   \cite{huo2021shallow} and the Gaussian derivative function in \cite{wang2023quantum}. }
To address the eigenstate tasks, we consider the spectral filter $g( \tau(H-\omega) )$ with $\omega = E_j$, which can be expressed as 
$g( \tau(H-\omega) ) = \sum_{i=0}^{2^n-1} g(\tau (E_i - \omega)) \ket{u_i}\bra{u_i}.$

\begin{algorithm}[H]
\begin{algorithmic}[1]
\State{Construct the composite LCU formula as described in \autoref{prop:composite_LCU}.
We expand a Gaussian filter $g$ into a combination of real-time evolution $U(t_i)$ by \autoref{eq:randomLCU_continous_main}, which is then decomposed into Pauli operators as $U(t_i) = (S_{2k}(t_i/\nu ) V_{2k}(t_i/\nu))^\nu$. Overall this takes a composite LCU form, with the LCU form of $ V_{2k}$ in \autoref{eq:TrotterError_LCU_V} and the explicit form in \autoref{eq:time_evo_decomp_main}). 
    }
\State {Realise the composite LCU formula by random sampling. We first generate $t_i$ and $t_j$ (at the level of spectral filter decomposition) and $\vec{i}$ and $\vec{j}$  (at the circuit level, see \autoref{eq:U_ti_LCU}) by sampling from the corresponding probability distributions, which determines the circuit intsance in each run (see \autoref{eq:observable_dynamics_main} for example). The circuit for evaluating $N(O)$ and $D$ (in \autoref{eq:N_O_LCU}) is shown in \autoref{fig:cartoon}(c). 
    The ancilla-free scheme is discussed in Methods \autoref{sec:measurement}.
    }
\State{Estimate the observable expectation value by  $\braket{O} = \frac{N_{\tau}(O)}{D_{\tau}}$, and eigenenergy by $E_j = \argmax_{\omega} D_{\tau}(\omega)$.
The estimator of $N_{\tau}(O)$ and $D_{\tau}$ is given by \autoref{eq:estimator_v_main}. The observable expectation on the eigenstate is given by \autoref{eq:OND_main}. As illustrated in \autoref{fig:cartoon}(b), the eigenenergy could be determined by $\max_{\omega} D_{\tau}(\omega)$ where the search process is evaluated by pure classical computation. The eigenenergy is searched via $\hat{D}_{\tau}(\omega)$ with purely classical search process.}
\end{algorithmic}
\caption{An overview of the random-sampling algorithm for \autoref{problem:observ} and \autoref{problem:eigenenergy}.
} 
\label{alg1_main}
\end{algorithm}

For an input state $\ket{\psi_0}=\sum_i c_i \ket{u_i}$, the state after applying the spectral filter $g( \tau(H-\omega) )$ at a finite $\tau$ is given by
\begin{equation}\label{Eq:coolingdef_main}
	\ket{\psi(\tau)} = \frac{g(\tau (H-\omega) )\ket{\psi_0} }{\|g(\tau (H-\omega))\ket{\psi_0} \|} = \frac{\sum_i g(\tau (E_i - \omega))c_i\ket{u_i}}{\sum_i |c_i|^2 g(\tau (E_i-\omega))^2}.
\end{equation}
It is easy to see that the Gaussian function $g(\tau (E_i-\omega))$ decreases exponentially with $\tau$ and $E_i - \omega$. When taking $\omega = E_j$, the amplitudes of the normalised state $\ket{\psi(\tau)}$ concentrate to the eigenstate with energy $E_j$, and the evolved state asymptotically approximates the  eigenstate $\ket{u_j}$ with nonzero $|\braket{\psi_0|u_j}|^2 \neq 0$ for sufficiently large $\tau$ as  
$ 
    \lim_{\tau\to\infty} g( \tau(H-E_j) ) \ket{\psi_0} \propto \ket{u_j}.
$

There are two important parameters $\tau$ and $\omega  $ in the spectral filter   $g(\cdot)$, in which   $\tau$ is an imaginary time scaling factor, and $\omega = E_j$ indicates a shift of the original function.
We shall see that $\tau$ indicates the timescale for the spectral filtering procedure, and larger $\tau$ will cool the state closer to the target eigenstate.
The shift $\omega$ plays an important role in searching the eigenenergies (see \autoref{fig:cartoon}(b)) and in the eigenstate property estimation. 
From the above equation, readers may wonder if we can still get the concentration around the eigenstate $\ket{u_j}$ if we do not know the value of $E_j$ a priori.
Here, we can find that $\omega$ only appears in classical post-processing and will not be involved in quantum measurement. Therefore, we will classically tune $\omega$ as a variable to find the true eigenenergy without increasing any quantum resource cost.


The Gaussian filter can be decomposed into the basis of real-time evolution as 
\begin{equation}
\begin{aligned}
    \label{Eq:superposed}
 g(\tau(H-\omega) ) = c \int^\infty_{-\infty} \textrm dp(x)  
 e^{ix(\omega -  H)} 
\end{aligned}
\end{equation}
with  $p(x) $ being a Gaussian distribution. 
The quantum state given by \autoref{Eq:coolingdef_main} becomes
\begin{equation}
\begin{aligned}
	\ket{\psi(\tau)}  = c\int^\infty_{-\infty} \textrm dp(x) e^{ i \tau x \omega} \ket{\phi(x\tau)}, \nonumber
\end{aligned}
\end{equation}
which is now a superposition of real-time evolved states $\ket{\phi(x\tau)} = e^{-ix\tau H}\ket{\psi_0}$ with probability distribution $\textrm dp(x) = p(x)\textrm dx$.

Instead of preparing the above  quantum state directly, we focus on the goal of obtaining arbitrary observable expectation values.  
Specifically, we aim to measure any observable $O$ of the evolved state $\ket{\psi(\tau)}$, i.e., 
\begin{equation}\label{eq:OND_main}
	\braket{O}_{\psi(\tau)} = \braket{\psi(\tau)|O|\psi(\tau)} = \frac{N_\tau(O) }{D_\tau(\omega)},
\end{equation}
where 
$ 
D_\tau(\omega) =   \braket{\psi_0| g^2(\tau(H-\omega)) |\psi_0}$ and $
N_\tau(O) =   \braket{\psi_0| g(\tau(H-\omega)) O g(\tau(H-\omega)) |\psi_0}$. Hereafter, we have omitted $\omega$ in $N_\tau$ for simplicity.
The denominator is 
$D_\tau(\omega) = c^2 \int_{-\infty}^{\infty} \textrm{d} \tilde p(x) e^{i\tau x \omega} 	\braket{\psi_0 |  e^{-i \tau x H}|\psi_0}$
 with  $\textrm{d} \tilde p(x) =\frac{1}{2} \textrm{d} x\int^\infty_{-\infty}p(\frac{z+x}{2})p(\frac{z-x}{2})\textrm dz$.  
The expectation values of the denominator and the numerator can be expressed as $v = \tr(O g_{\tau}(H-\omega) \rho_0  g_{\tau}(H-\omega) )$ with $\rho_0 = \ket{\psi_0}\bra{\psi_0}$, and the denominator is obtained by taking $O  = I$.

For the numerator $N_\tau(O)$, we can efficiently obtain it by sampling the distribution $\textrm{d}p(x,x')=\textrm{d}p(x) \textrm{d}p(x')$ and then estimating the mean value $\mathbb E_{x,x'} \braket{\phi(x'\tau)|O|\phi(x\tau )}$, where each term can be measured by the Hadamard test circuit or using the ancilla-free measurement strategy detailed in~\autoref{sec:measurement}. 
We can similarly obtain the denominator $D_\tau$ by estimating $ \mathbb E_x \braket{\psi_0|e^{- i x \tau H}|\psi_0}$ with probability $\textrm d\tilde p(x)$. 

The unnormalised eigenstate can be effectively realised by applying a spectral filter $g_{\tau\rightarrow \infty} (H - E_j)$ to an initial state.  When $\tau \rightarrow \infty$,   the  expectation value of an observable on the eigenstate satisfies $$\braket{O} = \frac{N_{\tau \rightarrow \infty}(O)}{D_{\tau \rightarrow \infty}}. $$ Here, it is also easy to check that the denominator is nonvanishing, given by
$ 
	D_{\tau \rightarrow \infty}(E_j) = \eta$
under the assumption.  Thus we arrive at the ideal observable expectation. 

\sun{
To realise the spectral filter with unitary operations, we   avoid the infinite time length by truncating the time in real-time evolution to a given threshold. The observable when considering a finite $
\tau$ is estimated by
$ 
	\braket{O}_{\tau} = {N_{\tau}(O)} / {D_{\tau}}. 
$ Below we discuss the LCU form of $g$ when considering the cutoff.}
The truncated Gaussian filter takes an explicit composite LCU form of
\begin{equation}
\begin{aligned}
	g_{\tau}(H-\omega) &= c \int_{-x_c}^{+x_c} dx p(x)  e^{ix\tau \omega}   e^{-i\tau x H}
\end{aligned}
\end{equation}
with the integral region $[-x_c, x_c]$.
The integrand is a real-time evolution with  time length $\tau x$.

Let us consider implementing $ e^{-i\tau x H}$ with more elementary gates as illustrated in \autoref{fig:cartoon}(a). Specifically, consider the following LCU decomposition
 \begin{equation}
 	e^{-i\tau x H} =  \mu( x \tau)\sum_{\vec{r} \in \mathcal{K}_x} \Pr(\vec{r}, x\tau, \nu(x \tau)) U_{\vec{r}} 
 \label{eq:time_evo_decomp_main}
 \end{equation} 
where $\vec{r}$ specifies the unitary $U_{\vec{r}}$ involved in the LCU formula of $ e^{-i\tau x H}$,  $U_{\vec{r}}$ is a unitary operator, and $\Pr(\vec{r}, x\tau, \nu(x \tau))$ represents the normalised decomposition coefficients of $U_{\vec{r}}$.
Then, we have 
\begin{equation}
\begin{aligned}
& g_{\tau}(H-\omega) 	\\
&= c(\mu) \int_{-x_c}^{+x_c} dx p_{\mu}(x)  e^{ix\tau \omega} \sum_{\vec{r} \in \mathcal{K}_x} \Pr(\vec{r}, x\tau, \nu(x \tau)) U_{\vec{r}} 
\end{aligned}
\label{eq:g_tau_real_expansion_main}
\end{equation}
where we have defined the normalisation factor 
$ 
	c(\mu) := c \int_{-x_c}^{+x_c} p(x) \mu(x \tau) dx  
$
and $ 
	p_{\mu}(x) = p(x) \mu(x \tau)/c(\mu)
$, see Proposition 5  for the detailed derivation.
 
Then we show how to find the operator $U_{\vec{r}}$  and the corresponding probability distribution, with which the LCU decomposition in \autoref{eq:time_evo_decomp_main} can be determined. The overall idea is to divide the time evolution into $\nu$ segments, and for each segment, we realise both the Trotter evolution and the Trotter error, as proposed in \cite{zeng2022simple}.
For evolution time $t$, let us denote the real-time evolution $U(t) := e^{-i H t} = U_m^{\nu}$ with time interval $m = t  / \nu$. We choose to implement the unitary $U(m) = e^{-i m H }$ by  a deterministic $2k$th-order Trotter formula $ S_{2k}$  and the Trotter remainder   ${V}_{2k}$, which gives us
$
{U}(m)  = {V}_{2k}(m) S_{2k}(m).
 $ As shown in the main text, the LCU formula of ${V}_{2k}$ can be expressed as
 \begin{equation}
 \label{eq:V_2k_m_decomp}
      \tilde{V}_{2k}(m) = \mu(m) \sum_r \Pr(r, m, \nu(m)) W_r 
 \end{equation}
 where $W_r$ is a unitary operator with the error of the formula $\varepsilon_{2k}(m)$.
The overall LCU formula of $U(t)$ is a product of each individual formula
\begin{equation}
    U(t) = (U(m))^{\nu} = \mu(t) \sum_{\vec{r} \in \mathcal{K}_t} \Pr(\vec{r}, t, \nu(t)) U_{\vec{r}} 
    \label{eq:U_t_LCU_form}
\end{equation}
where $\vec{r} := (r_1, r_2, ..., r_{\nu})$ with each $U_{r_i}$ being sampled from the distribution of $\Pr(r, m, \nu(m))$ introduced in \autoref{eq:V_2k_m_decomp} and we denote $ \Pr(\vec{r}, t, \nu(t)) := \prod_{i \leq \nu} \Pr(r_i, m, \nu(m))$ and $U_{\vec{r}} := \prod_{i \leq \nu} W_{r_i} S_{2k}$. With some derivation, one can prove that $\mu = \mu(t)^{\nu}$ and $\varepsilon_{2k}(t) \leq  \nu \mu(m) \varepsilon_{2k}(m)$.
\autoref{eq:U_t_LCU_form} shows how $U_{\vec{r}}$ can be sampled and thus how $g$ can be realised.
This is illustrated in \autoref{fig:cartoon}(a3); in the figure, the index by $i$ is related to the probability $\Pr(i)$ when decomposing $g$ into real-time evolution $U(t_i)$.
 

Using BCH formula, the $2k$th-order remainder   $V_{2k}(m):=U(m)S_{2k}(m)^\dag$ takes an explicit form of
$ V_{2k}(m) = \exp(i \sum_{s=1}^{\infty} E_{2k,s} m^s)$
with Hermitian operators $E_{2k,s} $.
Since the $2k$th-order Trotter error is compensated, we know that $E_{2k,s} = 0$ for $s \leq 2k$. Expanding $V_{2k}(m)$, we have
$ 
V_{2k}(m)=  \sum_{s=0}^\infty F_{2k,s}(m) 
$
where we group the terms by the order of $m$, and  $F_{2k,s}$ denotes the $s$-order expansion term of $V_{2k}(m)$ associated with $m^s$. 



Below, we briefly discuss the error due to truncation by $s_c$.
The truncation error is found to have a quick decrease with an increasing truncation order $s_c$ \cite{zeng2022simple}. Using the fact of $F_{2k,s\leq 2k}(m) = 0$, we can rewrite $V_{2k}^{(s)}(m)$ as
$ 
V_{2k}^{(s)}(m)= I + \sum_{s=2k+1}^{s_c} F_{2k,s}(m).
$
Given the truncation, the  LCU formula for $U(m)$ is
\begin{equation} \label{eq:U_sc_approxi_main}
\begin{aligned}
 {U}_{2k}^{(s_c)}(m)  = {V}_{2k,s}(m) S_{2k}(m),
\end{aligned}
\end{equation}
which consists of a deterministic second-order Trotter formula and the Trotter error compensation term.
The overall LCU formula for $U(t)$ is to repeat the sampling of $\tilde{U}_{2k}(m)$ for $\nu$ times,
$ 
	{U}_{2k}^{(s_c)}(t) = \left( {U}_{2k}^{(s_c)}(m) \right)^\nu
$.

\autoref{eq:g_tau_real_expansion_main} is a composite LCU formula; more specifically, it is a $(c(\mu),0)$ LCU formula. When we consider a finite $x_c, s_c$, it will introduce algorithmic errors. 
When considering a finite $s_c$, the spectral filter becomes
\begin{equation}
\label{eq:def_g_tau_x_s}
	g_{\tau, x_c, s_c}(H-\omega) =  c \int_{-x_c}^{x_c} dx p(x)  e^{ix\tau \omega}   {U}_{2k}^{(s_c)}(x\tau, \nu(x\tau))
	\end{equation}
It is easy to check that $\| g_{\tau, x_c, s_c}(H-\omega)\| \leq c (1+\varepsilon_{s_c})$.
The operator distance between $g_{\tau,x_c} $ and $g_{\tau,x_c, s_c}$ due to  finite $s_c$ is
\begin{equation}
\begin{aligned}
& \|g_{\tau,x_c} - g_{\tau,x_c, s_c}  \| \\
& \leq c \int_{-x_c}^{x_c} dx p(x)  e^{ix\tau \omega}   \| U(x\tau) - {U}_{2k}^{(s_c)}(x\tau, \nu(x\tau)) \| \leq   \varepsilon_{s_c}    
\end{aligned}
\label{eq:g_tau_sc_error}
\end{equation}
when $\| U - {U}_{2k}^{(s_c)}(x\tau, \nu(x\tau)) \| \leq \varepsilon_{s_c} /c$.
In order to achieve an additive error of the approximation  $ 
	\|U(t) - {U}_{2k}^{(s_c)}(t)) \| \leq \varepsilon_{s_c} 
$, Lemma 1 in Supplementary~Sec. III B gives the required segment numbers $\nu$.

Based on the above result, the total error of the spectral filter function $g$ can be bounded by the sum of the individual error terms. 
The derivation process is illustrated within the orange box of \autoref{fig:section_connection}. Specifically, we can bound the error in spectral filter construction by comparing the operator distance between 
$ 
	\| g_{\tau, x_c, s_c}  - g_{\tau, x_c, s_c \rightarrow \infty}  \|,
$
which gives an upper bound for the numerator and the denominator. 
Using the triangular inequality, the operator error between $g_{\tau}$ and $g_{\tau,x_c, s_c} $ defined in \autoref{eq:def_g_tau_x_s} is 
\begin{equation}
\label{eq:g_tau_xc_sc}
	\|g_{\tau} - g_{\tau,x_c, s_c}  \| \leq \varepsilon_{x_c}   + \varepsilon_{s_c}. 
\end{equation}

We find that   the segment number of the order 
\begin{equation}
\label{eq:nu_order_main}
    \nu =  \mathcal{O}((\lambda \tau x_c)^{1+\frac{1}{4k+1}} )
\end{equation} suffices to ensure the LCU construction error up to~$\varepsilon$. The detailed segment number along with the constants is provided in Theorem 5 in Supplementary~Sec. III.
The required segment number for eigenenergy estimation is shown in Theorem 7 in Supplementary~Sec. III. 
Together with Proposition 3 in Supplementary~\autoref{sec:gate_complexity_analysis}, one proves \autoref{thm:observ_estimation_main}.

For lattice models, the segment number can be further reduced to $\nu = \mathcal{O}\left(n^{\frac{2}{4k+1}} (\tau x_c)^{1+\frac{1}{4k+1}} \varepsilon^{-\frac{1}{4k+1}}\right)$. This result leads directly to the proof of \autoref{thm:observ_estimate_gate_Lattice_main}.

\subsection{Observable estimation }
\label{sec:methods_observable_estimation}

The previous section focused on error propagation in the LCU formulation.
Here, given the constructed LCU formula, we analyse the estimation of eigenstate properties and the associated error.
We first show how to measure the expectation value of the observable $O$. Denote the initial state on which the spectral operator acts as $\rho$. The expectation value of the observable $O$ on the normalised state can be expressed by
\begin{equation}
	\braket{O} = \frac{N(O)}{D} = \frac{\tr(g \rho g^{\dagger}O)}{\tr(g \rho g^{\dagger})}. \nonumber
\end{equation}
The numerator can be expressed by
\begin{equation}
    N(O) = \mu^2 \sum_{ij} \Pr(i) \Pr(j) e^{i (t_i - t_j) \omega } \tr(U(t_i) \rho U^{\dagger}(t_j) O).\label{eq:N_O_LCU}
\end{equation}

Denote the expectation value of the estimator of the numerator over measurement outcomes with $U(t_i)$ and $U(t_j)$, $\tr(U(t_i) \rho U(t_j)^{\dagger} O)$, as $\bar N_{ij}(O)$. The expectation value of $N(O)$ is given by sampling over the distribution
\begin{equation}
	\bar N(O) = \mu^2 \mathbb E_{ij} \bar N_{ij}(O)
\end{equation}

When considering finite gate complexity and sample complexity, the eigenstate property is estimated by
\begin{equation}
\label{eq:observ_expec_finite}
	\hat{O}_{\tau, x_c, s_c} = \frac{\hat N_{\tau, x_c, s_c}(O)}{\hat D_{\tau, x_c, s_c}}. 
\end{equation}

The selection of $\tau, x_c, s_c$ can be determined by analysing the error of $\hat{O}_{\tau, x_c, s_c}$ compared to the ideal value, which can be analysed using the operator distance detailed in the above section.
\autoref{prop:RLCU_err_observ_bound} shows the error in observable estimation when we are given a  $(\mu,\varepsilon)$  random sampling formula of~$g$ by $\tilde g$ in \autoref{eq:randomLCU_main}.

 
\begin{proposition}[Observable estimation using the composite LCU formula] \label{prop:RLCU_err_observ_bound}
The estimation error of $N_g(O):= \tr(g\rho g^{\dagger} O)$  is bounded by
$ 
\varepsilon_N := |\hat N_{\tilde g}(O) - N_g(O)| \leq   \|O\|(2\mu^2 \varepsilon + \varepsilon_n),
$
with $N_s =\mu^4\ln(2/\vartheta)/\varepsilon_n^2$ samples and a success probability $1-\vartheta$. 
 The error for the denominator is bounded by 
	$\varepsilon_D := |\hat D_{\tilde g}  - D_{g} | \leq 2\mu^2 \varepsilon + \varepsilon_n$.
The error of observable expectation $\varepsilon_O$ is
bounded by $ \varepsilon_O \leq D_g^{-1} ((\braket{O} + 1)\varepsilon_D +  \varepsilon_N)	$.

\end{proposition}

{Given \autoref{prop:RLCU_err_observ_bound}, the error of the observables can be bounded given the RLCU form.  The proof is shown in Supplementary Sec. II (where the formal version is presented).
}
Using these results of error propagation in the LCU form, we can determine the required segment number $\nu_c$ in order to achieve the desired accuracy as shown in \autoref{eq:nu_order_main}. 
Theorem 5 in Supplementary~Sec. III and Theorem 7 in  Supplementary~Sec. IV present the required segment number $\nu_c$ and the gate count (with the actual prefactors included) for achieving the desired accuracy for the two problems. They can be used to estimate the actual gate numbers needed for physical Hamiltonians.

The total gate count is determined by multiplying  $\nu_c$ and the gate count for realising $S_{2k}$ and $V_{2k}$ within each segment, denoted by $g$ hereafter (with an abuse of notation).
The gate count for lattice models is straightforward $g = \mathcal{O}(n)$. For electronic structure problem in \autoref{eq:molecular_hamil_main}, a naive realisation of $S_{2k}$ and $V_{2k}$ will result in the cost $g = \mc O (n^3)$. However, we could bring down the scaling by considering the properties of fermions.
The kinetic operator is quadratic and thus can be diagonalised by an efficient circuit transformation, either by Givens rotations or fermionic fast Fourier transform (FFFT)~\cite{kivlichan2020improved,kivlichan2018quantum,babbush2018low}.
By using the results in \cite{low2018hamiltonian}, $e^{-i m \hat{T}}$ and $e^{-i m \hat{V}}$ can be implemented with $\mc O(n \log n)$ gates with $m$ being the time length.
The compensation term will cost $\mc O(n )$ gates. Therefore, we have $g = \mathcal{O}(n\log n)$ and  the total gate count for electronic structure problem is $g_{\rm tot} = \tilde{\mc O} ( n (\lambda \Delta^{-1})^{1+\frac{1}{4k+1}}  \log \varepsilon^{-1}) $.

\subsection{Estimation by quantum circuits: general strategy and ancilla-free strategy}
\label{sec:measurement}


In this section, we discuss the circuit realisation for   \autoref{alg1_main}. 
As discussed in the main text, to measure $\braket{\vec {\mathbf{j}} | O | \vec {\mathbf{i}}}$ (introduced in \autoref{eq:observable_dynamics_main}), we can use the Hadamard test circuit with an extra ancillary qubit.
For Hamiltonians with certain symmetries, we do not need this ancillary qubit and thus do not require any controlled operation. Below, we first discuss the general cases by introducing an ancillary qubit. Then we discuss the ancilla-free composite LCU formula and the corresponding measurement schemes.

\vspace{5pt}

\noindent
\textbf{General cases.}
There are two extreme cases:  Case I $t_i = 0$ (or $t_j = 0$) and Case II $t_i = t_j$.
For $t_j = 0$, each component in \autoref{eq:observable_dynamics_main} $\braket{ \vec {\mathbf{  j}} | O | \vec {\mathbf{  i}}}$ can be measured using a Hadamard test circuit, which is the circuit depicted in the green box in \autoref{fig:cartoon}(c).

For $0 < t_j < t_i$, suppose that $U(t_j)$ is divided into $m$ segments and  $U(t_i)$ is divided into $\nu$ segments with the first $m$ segments being set the same as that of $U(t_i)$.   The quantum circuit that can measure $\braket{\vec {\mathbf{j}} | O | \vec {\mathbf{i}}}$ is shown in \autoref{fig:cartoon}.
After post-selecting $\ket{+}$ on the ancillary qubit (but before measurements on $B$), the circuit on system $B$ outputs 
\begin{equation}
    \frac{1}{2}\left(   \prod_{q =1}^{m} W_{j_q} S +   \prod_{q =1}^{\nu} W_{i_q} S \right) \ket{\psi_0}.\nonumber
\end{equation}
The case of $t_i = t_j$ is a special case in \autoref{fig:cartoon}(c1) by only implementing the circuits in the green box.

Here, we can use a single-shot measurement strategy to estimate $N(O)$ in \autoref{eq:N_O_LCU}.
Given the sequence of unitary operators obtained from sampling, the circuit in \autoref{fig:cartoon}(c1) can be used to measure $\braket{\vec {\mathbf{j}} | O | \vec {\mathbf{i}}}$. Specifically, we initialise the ancilla in $\ket{+}$, apply the controlled-unitaries, and perform the measurement on the $X$ basis, with the measurement outcome recorded as $a = \{0, 1\}$. Similarly, we repeat the process but with an inverted phase gate $S^{\dagger}$ applied before measurement, with the measurement outcome  $b = \{0, 1\}$.



One can verify that the estimator   
$ 
	\hat d =  (-1)^a + i (-1)^{b}
$ is unbiased, $\mathbb E_{a, b}  \hat d =  \braket{\vec {\mathbf{j}} | O | \vec {\mathbf{i}}}$. 
Now, take an estimator as
\begin{equation}
	\hat{v} = c^2(\mu)  e^{i \omega(t_i - t_j)} \hat{d}
	\label{eq:estimator_v_main}
\end{equation} which can be proven to be unbiased: 
 \begin{equation}
     \mathbb E_{t_i, t_j, \vec{i}, \vec{j}} \mathbb E_{a,b}\hat{v}  = \braket{\psi_0|g_{\tau}(H-\omega)Og_{\tau}(H-\omega)\psi_0}.
 \end{equation} 
 
By analysing the error dependence on the finite $\tau$, $x_c$, segment number $\nu$ in \autoref{eq:nu_order_main}, we arrive at the result of property estimation in \autoref{thm:observ_estimation_main} (see detailed description in Theorem 5 and Theorem 7). More detailed proof can be found in Supplementary~Sec. III and Sec. IV.



\vspace{5pt}
\noindent
\textbf{Ancilla-free composite LCU formula and the measurements}
In the above section, given a sampled configuration $(t_i, t_j, \vec{i}, \vec{j})$, the real and imaginary part of $\braket{\vec {\mathbf{j}} | O | \vec {\mathbf{i}}}$ can be obtained by the circuit in \autoref{fig:cartoon}(c). 
In cases of the Heisenberg model and electronic structure problems (\autoref{eq:molecular_hamil_main}), the target problem has certain symmetries $\mathcal{S}$ satisfying $[H, \mathcal{S}] = 0$.  
For symmetry-conserved systems with $[U, \mathcal{S}] = 0$, it is possible to estimate the expectation value of a unitary $\braket{\psi_0 | U | \psi_0}$  without ancilla, as proposed in \cite{o2021error,cortes2022quantum}. However, the issue here is that as the unitary operator is realised by a Trotter-LCU expansion in \autoref{eq:U_ti_LCU}, at least there exists a Pauli operator $W_{i_q}$ that does not commute with $\mathcal{S}$. 
We first briefly introduce how to measure  $\braket{\psi_0 | U | \psi_0}$ without ancilla, where the unitary is either $U = e^{-iHt}$ or $U = e^{-iHt_1} \hat O e^{-iHt_2}$, followed by the design of the unitary in \autoref{eq:U_ti_LCU} such that $[U, \mathcal{S}] = 0$.  

The expectation value $\braket{\psi_0 | U | \psi_0}$ can be expressed as $\braket{\psi_0 | U | \psi_0} = r e^{i\theta}$.
If $\ket{\psi_0}$ is a product state, or can be prepared as $\ket{\psi_0} = U_p \ket{0}^{\otimes n}$, the amplitude  of the expectation value $r = |\braket{ \psi_0 |   U    | \psi_0}|$ can be obtained by  measuring $U_p^{\dagger}U U_p \ket{0}^{\otimes n}$ in the computational basis.
The next step is to obtain the phase $\theta$, for which we make use of the fact that the unitary operation conserves the symmetry of $\mathcal{S}$, $[U, \mathcal{S}] = 0$.
To do so,  a reference state is introduced $\ket{\psi_{\mathrm{Ref}}}$, which lies in a different sector of the initial state, such that $\braket{\psi_{\mathrm{Ref}} | U | \psi_0} = 0$. Here, $\braket{\psi_{\mathrm{Ref}}| U | \psi_{\mathrm{Ref}}}$ can be computed classically. 
In addition, the following state can be prepared
$$
    \ket{\phi_0} = U_p \ket{\psi_0} = \frac{1}{\sqrt{2}} (\ket{\psi_{\mathrm{Ref}}} + \ket{\psi_0}).
$$
As $ r_s=|\braket{ \psi_0 | U_p^{\dagger} U U_p  | \psi_0}|$ can be measured on a computational basis, the phase $\theta$ can be computed. 

Note the there will be double solutions for determining  $\theta$ in general and may be hard to distinguish. One may track the dynamics of expectation values to determine the correct phase.

The preparation of the superposition state is remarkably simple, requiring only one additional \textsc{cnot} gate compared to the original state, which is used for real demonstrations on IBM devices. Discussions can be found in Supplementary~Sec. III. Now, to measure $\braket{\vec {\mathbf{j}} | O | \vec {\mathbf{i}}}$, we use the circuit in \autoref{fig:cartoon}(c) to generate 
\begin{equation}
    U_p^{\dagger} \left(    \prod_{q =1}^{m} W_{i_q}^{\dagger} S^{\dagger} \right) O \left(  \prod_{q =1}^{\nu}  W_{i_q} S\right) U_p \ket{\psi_0},\nonumber
\end{equation}
and then measure on a computational basis. 


Recall that our algorithm consists of implementing the Trotter formula $S$ and the Trotter remainder $V$.
For the implementation of $S$, we need to pair the terms to ensure each individual term does not break the symmetry.
Take the quantum chemistry problems, $H = \hat  T + \hat  V$, for example.
We group each term $\hat{T}_{ij} = h_{ij}  (\hat a_i^{\dagger} \hat  a_j + \hat  a_j^{\dagger} \hat  a_i)$ in the kinetic term $\hat T$, which conserve the particle number. 
Its exponentiation $e^{-i \theta \hat{T}_{ij}}$ can be directly implemented, which is similar to the implementation for the Heisenberg model case.
For the potential term $\hat{V}$, we can group its component $g_{ijkl} ( \hat a_i^{\dagger}  \hat a_j^{\dagger}  \hat a_k  \hat a_l +  \hat a_l^{\dagger}  \hat a_k^{\dagger}  \hat a_j  \hat a_i)$ and implement its exponentiation similarly.

Note that operators are typically indicated by using the hat notation. In this work, to avoid confusion with estimators, we reserve the hat notation exclusively for fermionic operators and estimators.
In this paper, the symbol $V$ refers to the Trotter remainder. The potential term in molecular Hamiltonians is denoted as $\hat{V}$ and is distinguished from the Trotter remainder by the hat notation.

The issue is that if we expand the Trotter remainder into  {Pauli bases},
then applying each individual Pauli term will break the symmetry.
One may think of expanding $V(m)$ into
$ 
    V(m) =  \sum_{k,l,i,j} \hat{T}_i^k \hat{V}_j^l m^{k+l}
$
where $\hat{T}_i$ and $\hat{V}_j$ are components of the kinetic and potential terms, respectively.
Although each term conserves the symmetry, the issue is that $\hat{T}_i$ or $\hat{V}_j$ is not unitary, and cannot be implemented directly.
To address this issue, we decompose the Hamiltonian into \textsc{swap}  and tensor products of Pauli-$Z$ operators rather than Pauli operators. 

The explicit formulation for the Fermi-Hubbard model and quantum chemistry problems has been presented in the main text.
Although the Hamiltonian is reformulated in the basis of \textsc{swap}  and Pauli-$Z$ operators, the implementation of the Trotter formulae remains the same as that in the Pauli basis.
Only the compensation term will be different from the case with one ancillary qubit. As shown in \autoref{sec:gate_complexity_analysis}, the ancilla-free measurement strategy shows advantages when qubit connectivity is restricted.

\begin{table*}[ht!]

\centering
\begin{tabular}{c| c c | c | c c} 
\hline
\hline
Hamiltonians           & NN (ancilla-free) & NN (1-ancilla)   & NN (QETU)  \cite{dong2022ground}   & arbitrary    \\
\hline
Heisenberg model  &  $\boldsymbol{d} = \mc O (\boldsymbol{1})$ &  $d = \mc O (n)$ &  $d = \mc O (n)$    & $d = \mc O (1)$       \\
    &   $g = \mc O (n) $  & $g=\mc O(n)$ &  $g = \mc O (n)   $   &  $g = \mc O (n)   $   \\
    \hline
Electronic structure    & $\boldsymbol{d} = \mc O (\boldsymbol{n})$ & $ {d = \Omega (n^2)}$   & $d = \Omega (n^2)$ & $d = \mc O (n)$ &  \\
 (\autoref{eq:molecular_hamil_main})   &   $g = \mc O (n^2) $ & $g = \mc O (n^3) $    & $g = \mc O (n^3) $ & $g = \mc O (n \log n) $ \\
       &  & $\boldsymbol{d} = \mc O (\boldsymbol{n})$  ($[n/2]$ ancilla)   &   &      \\
\hline
Molecular Hamiltonians  & $d = \mc O (n^2)$ & $d = \Omega (n^3)$ &  $d = \Omega (n^3)$ & $d = \mc O (n^2)$  \\
(\autoref{eq:molecule_general_MCH}) &   $g = \mc O (n^3) $ &  $g = \mc O (n^4)   $ &  $g = \mc O (n^4) $ & $g = \mc O (n^2 \log n) $    \\
 
\hline
\hline
\end{tabular}
\caption{ \textbf{Circuit depth $d$ and gate count $g$ in each segment with nearest-neighbour (NN) or arbitrary architecture.}
The ancilla-free or the 1-ancilla schemes are compared when using our method or QETU method. 
}
\label{table:hamil_dep_cost}
\end{table*}

\subsection{Analysis of circuit depth and gate complexity }
\label{sec:gate_complexity_analysis}

We first discuss the gate complexity for general Hamiltonians. Then we analyse the circuit depth and gate complexity for various Hamiltonians of physical relevance when using either our method or other early fault-tolerant quantum algorithms. 
A comparison with a focus on circuit depth for each elementary block is shown in \autoref{table:hamil_dep_cost}.

\noindent
\textbf{Gate complexity for general cases}
The maximum resource appears in the case of $t_j = 0$ in \autoref{fig:cartoon}(c), represented by the orange box.  In this case we need to implement one controlled Trotter formula plus two controlled compensation terms specified by ${V_i}$. The minimum resource appears in the case when $t_j = t_j$, in which case we need to implement the Trotter formula plus two controlled compensation terms specified by ${V_i}$. This is because a controlled Trotter is more costly than a controlled  compensation operation, as analysed in \autoref{prop:gate_cost_observ_dynamics}. 
The resources in other cases with $t_j < t_i$ will be between Case I and Case II.


Define the following characters of the Hamiltonian,
$ \operatorname{wt}(H) :=\sum_{l=1}^{L} \mathrm{wt}\left (P_{l}\right )$, and $ \mathrm{wt}_{m}(H) :=\max _{l} \mathrm{wt}\left (P_{l}\right ) $.
Here, wt$\left (P_{l}\right )$ indicates the weight of the Pauli operator $P_{l}$, i.e., the number of $\{X, Y, Z\}$ terms in $P_l$.
The following result gives the required elementary gates

\begin{proposition}[Elementary gate count]
\label{prop:gate_cost_observ_dynamics}
Suppose   the composite LCU is realised by 
\autoref{alg1_main} and the Hamiltonian dynamics $U = e^{-iHt}$ is realised by \autoref{eq:U_ti_LCU} with maximum segment number $\nu$.
The gate complexity of the circuit instance in eigenstate property estimation is the following: \textsc{cnot} gate number: $\nu(4 \times 5^{k-1} \wt(H) + 4 \wtm (H) + 2 \min(n, s_c \wtm)) - 2L$, single-qubit Pauli rotation gate number: $(2L + 4)\nu $. 

\end{proposition}

We give more details in the following.	
Overall, the circuit within each segment consists of one controlled Trotter plus controlled compensation terms.
Each of the components within the Trotter formula takes the form of $e^{-i H m}$ with evolution time $m = t/\nu$.
The cost by Trotter is given by Lemma 3.

A controlled-Pauli-rotation gate, Ctrl-$\exp(-i P \theta/2)$ with Pauli operator $P$, can be realised by $2(\wt(H)-1)$ \textsc{cnot} gates, a controlled single-qubit  $R_z(\theta)$ gate with rotation angle $\theta$, and some single-qubit Clifford gates. Furthermore,  the controlled $R_z(\theta)$ gate can be decomposed into two single-qubit $Z$-axis rotation gates and two \textsc{cnot} gates.
Therefore, the cost for the controlled Trotter operator Ctrl-$S_{2k}(m)$, in each segment, the number of \textsc{cnot} gates are
$ 
	2 \times 5^{k-1} (2 \wt (H) -2L)
$
and $2L$ controlled single-qubit Pauli rotations.

The compensation term $V$  consists of a controlled multi-qubit Pauli rotation operator and some controlled Pauli operators. The controlled multi-qubit Pauli rotation operator can be realised by $2 \wtm (H)$ \textsc{cnot} gates and two single-qubit Pauli rotations.
The Pauli gates are randomly drawn
from $\{P_l\}_l$ according to the corresponding probability distribution.
In the worst case, the gate sequence length is the truncated value $s_c$, and thus the \textsc{cnot} gate number is bounded by $(s_c \wtm(H))$. Note that saturation occurs when the number of compensation Pauli operations reaches 
$n$.
To sum up, for the compensation term, the number of \textsc{cnot} gates is upper bounded by
\begin{equation}
\label{eq:compensation_CNOT_gate}
	2 \wtm (H) + \min(n, s_c \wtm) \leq 3n
\end{equation}
and 2 single-qubit Pauli rotations.
As there is a saturation of the gate count for the Trotter-error-compensation indicated by \autoref{eq:compensation_CNOT_gate}, one can take $s_c$ to be infinity in deriving the asymptotic scaling in gate complexity in \autoref{thm:observ_estimation_main} and \autoref{thm:observ_estimate_gate_Lattice_main}.

{One can thus check that Case I is more costly than Case II in \autoref{fig:cartoon}.}
For observable dynamics described in \autoref{eq:observable_dynamics_main}, there are two compensation terms, but simply one controlled Trotter Ctrl-$S_{2k}$, as illustrated in \autoref{fig:cartoon}(c1).
Together with the cost for the Trotter formula, the number of \textsc{cnot} gates is upper bounded by
\begin{equation}
	\nu (4 \times 5^k \wt (H) -2 L + 4 \wtm (H) + 2 \min(n, s_c \wtm)   ).
\end{equation}
The number of single-qubit Pauli rotations is $(2 L +4)\nu $.

The gate complexity with the ancilla-free strategy can be analysed in a similar way. Its key advantage is better depth scaling as it allows circuit compilation using only nearest-neighbour gates without overhead. This is discussed in detail in \autoref{sec:gate_complexity_analysis}.

\vspace{5pt}

\noindent
\textbf{Hamiltonian-specific circuit compilation and depth analysis.} The circuit depth and gate count for \autoref{problem:observ} and \autoref{problem:eigenenergy} using either the ancilla-free scheme and the one-ancilla scheme are displayed in \autoref{table:hamil_dep_cost}. Below, we will discuss the result when the qubit connectivity is taken into account.


We consider Hamiltonians including (1) Heisenberg models, (2) electronic structure in the plane wave dual basis in \autoref{eq:molecular_hamil_main}, and (3) second-quantised molecular Hamiltonian in the form of
$ 
    {H}=\sum_{i,j=1}^n h_{ij} \hat{a}_i^\dagger \hat{a}_j + \frac{1}{2}\sum_{i,j,k,l,=1}^n g_{ijkl}\hat{a}_i^\dagger \hat{a}_j^\dagger \hat{a}_k \hat{a}_l.
$
A common strategy for efficiently representing the Hamiltonians with fewer terms and low weights~\cite{berry2019qubitization,von2021quantum,lee2021even,motta2021low} is to reformulate the two-body fermionic operators as a sum of squared one-body operators by the Cholesky decomposition~\cite{motta2021low}, reformulated as 
\begin{equation}
    {H} = \hat{K}+\frac{1}{2} \sum_{\ell}^{\Gamma} \hat{L}_{\ell}^{2},
 \label{eq:molecule_general_MCH}
\end{equation}
where $\hat{K}$ and $ \hat{L}_{\ell}$ are the one-body terms, the number of terms is $\Gamma = \mc O (n)$ \cite{lee2021even}, and the constant has been removed. 


\begin{figure*}[tb!]
\centering
\includegraphics[width=0.75\linewidth]{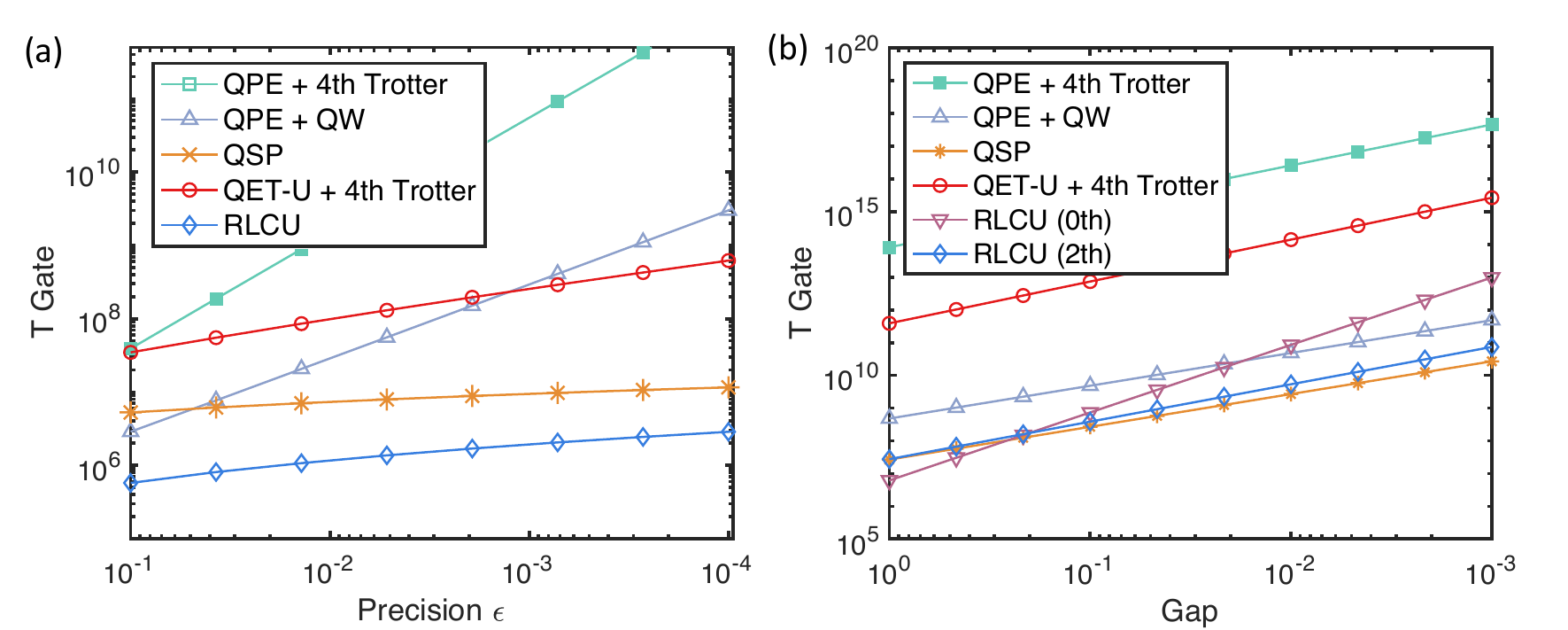}
\caption{ 
\textbf{The T gate count  under the same setup of \autoref{fig:SysSize_NC} in the main text.} (\textbf{A}) The T gate count dependence on precision $\varepsilon$ for the Heisenberg model.
(\textbf{B}) The T gate count dependence on the energy gap for the P450 molecule. As noted in the main text, the simulation of fermionic dynamics with Trotterisation is not optimised in this work. The total weight $\wt(H) = 591$. 
}
\label{fig:Heis_NC_Eps}
\end{figure*}

\begin{figure}[t!]
\centering
\includegraphics[width=0.8\linewidth]{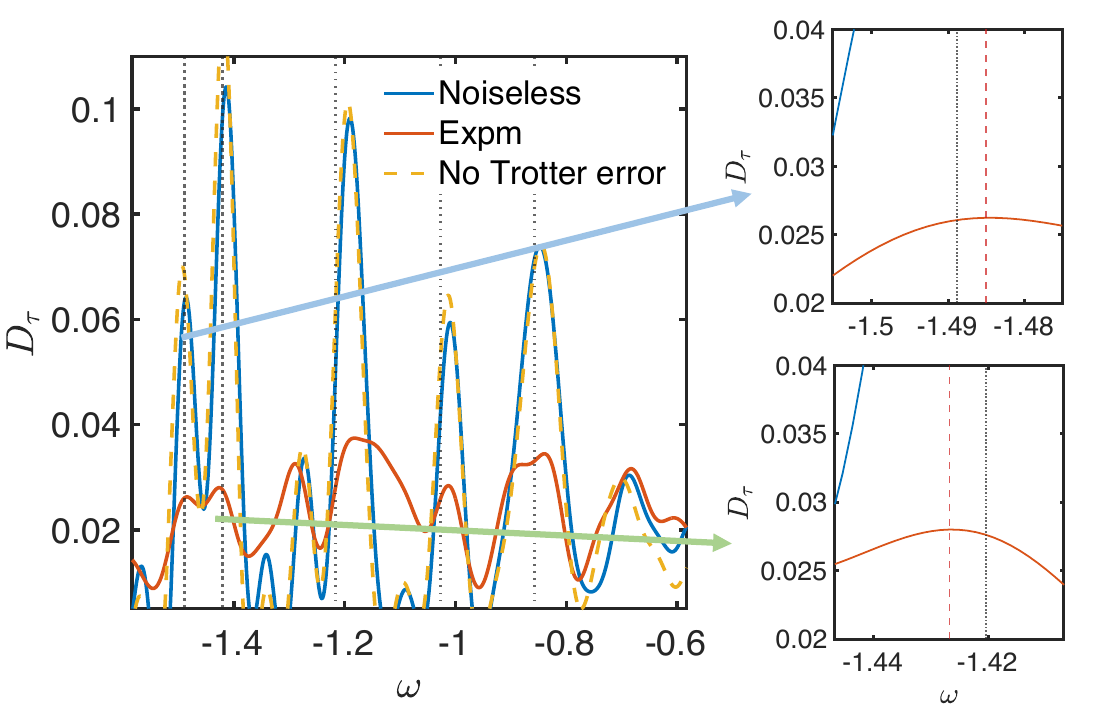}
\caption{ \textbf{
Searching ground state and first excited state energies of the XXZ Heisenberg Hamiltonian on IBM quantum devices.}
The parameters of the Hamiltonian are $J_x = J_y = 0.8$, $J_z = 1$, and $h_z = 0.2$ in \autoref{eq:Hamil_XXZ} in the main text. The initial state is chosen as a Neel state. The reference state is prepared with only one additional \textsc{cnot} gate compared to the original state. The imaginary time and the cutoff are set as $\tau = 4$ and $x_c = 2$.  The figure on the right provides a zoomed-in view of a narrower energy range, highlighting the estimated ground- and first-excited energies shown in the left panel with errors $0.003$ and $0.006$, respectively. 
}
\label{fig:IBM_Methods}
\end{figure}

For general cases, the circuit within each segment consists of two parts: the controlled implementation of the Trotter formula $S$ and the Trotter remainder $V$.
The only difference is the compensation term, which uses the symmetry-conserved gates as the elementary gates, i.e.,  \textsc{swap}, Pauli-$Z$ gates and their exponentiations.
We observe that to realise the compensation term,  at most $d = \mathcal{O}(n)$ and $g = \mathcal{O}(n)$ are required in the worst case, which is fewer than the implementation of the Trotter formula in all physical Hamiltonian cases. \autoref{table:hamil_dep_cost} displays the gate complexity (including both \textsc{CNOT} gates and non-Clifford gates) within each segment when taking the qubit connectivity into account. We elaborate on these results below.

We first discuss the Heisenberg-type Hamiltonians.
Each component $e^{-i \delta t_i (J_1 X_{i}X_{i+1}+ J_1 Y_{i}Y_{i+1}+ J_2 Z_{i}Z_{i+1})}$ with time duration $ \delta t_i$ can be realised by 3 \textsc{cnot} gates and 3 single-qubit Pauli rotation gates.
When restricted to a linear NN architecture, $d_{\textrm{Trotter}} = \mathcal{O}(1)$.
In the ancilla-free scenario, the compensation term is $\SWAP_{i,i+1}$, Pauli-$Z$ operators and the exponentiations, resulting in $d = \mathcal{O}(1)$ and $g = \mathcal{O}(n)$.
In the one-ancilla scenario, since we need to implement controlled two-qubit Pauli rotations, when restricted to linear NN architecture, we need to swap the ancillary qubit sequentially to qubit $1, ...n$ and then perform the controlled rotation. Then, we need to undo the swap to change back the ordering of the qubits. The compensation term can be done when the ancillary qubit is adjacent to the target controlled qubit.
Compared to the ancilla-free scheme, we cannot do it in parallel, resulting in $d = \mathcal{O}(n)$ and $g = \mathcal{O}(n)$.
When removing the restriction of the qubit connectivity, both scenarios have $d = \mathcal{O}(1)$ and $g = \mathcal{O}(n)$.

For the electronic structure problem in \autoref{eq:molecular_hamil_main}, it can be grouped into ${H} = \hat{T} + \hat{V}$. Below, we discuss the cost for the ancilla-free scheme and the one-ancilla scheme with a linear NN architecture.
The kinetic term is a non-interacting term which can be diagonalised as $\hat{T} = \hat{C} (\sum_i \alpha_i \hat{n}_i) \hat{C} ^{\dagger}$.
For the implementation of $e^{- i \theta \hat{T}}$, an NN architecture requires $d = \mathcal{O}(n)$ and $g = \mathcal{O}(n^2)$.
For the interacting potential term $e^{- i \theta\hat{V}}$ which consists of $\mathcal{O}(n^2)$ terms, an NN architecture requires $d = \mathcal{O}(n)$ and $g = \mathcal{O}(n^2)$. Therefore, $d_{\textrm{Trotter}} = \mathcal{O}(n)$ and $g_{\textrm{Trotter}} = \mathcal{O}(n^2)$. 
For the compensation term, since there is only one term in the form of $\SWAP_{i,j} \otimes_{k = i+1}^{j-1} Z_k$ or $  \otimes_{k = i+1}^{j-1} Z_k$ or $  \otimes_{k = i}^{j} Z_k$, which can be implemented at most $d_{\textrm{Remainder}} = \mathcal{O}(n)$ and $g_{\textrm{Remainder}} = \mathcal{O}(n)$. We note that for the 2D Fermi-Hubbard model, which is a special case of \autoref{eq:molecular_hamil_main}, the compensation term can be implemented with $d_{\textrm{Remainder}} = 1$ given a planar nearest-neighbour architecture.
To sum up, an NN architecture requires $d = \mathcal{O}(n)$ and $g = \mathcal{O}(n^2)$.

We then discuss the gate and depth cost when we need a control ancilla (used in both the so-called one-ancilla scheme and QETU). The first observation is that we do not need to implement controlled $\hat{C}$. The rest of the diagonal part is similar to the case of the Heisenberg model, which requires $d = \mathcal{O}(n)$ and $g = \mathcal{O}(n)$.
The implementation of $e^{- i \theta \hat{T}}$ requires $d = \mathcal{O}(n)$ and $g = \mathcal{O}(n^2)$.
However,   the implementation of $e^{- i \theta\hat{V}}$ {requires $d = \Omega (n^2)$ and $g = \Omega (n^2)$ in an NN architecture.}
The challenge of realising the controlled $\hat{V}$ with one control qubit is that at each time, the control qubit can only control one rotation $R_z(\theta_{ij})$ given a pair of $(i,j)$. Therefore, the controlled rotation cannot be realised in parallel. Naively, we could first apply a cyclic swap operation and then realise each individual $R_z(\theta_{ij})$. This results in the total depth $d = \mathcal{O}(n^3)$ and gate count $g= \mathcal{O}(n^3)$.
A lower bound on the depth complexity for realising the controlled $\hat{V}$
may be given by $\Omega (n^2)$.
The compensation term requires at most $d = \mathcal{O}(n)$ and $g = \mathcal{O}(n)$.
Therefore, an NN architecture requires $d = \Omega(n^2)$ and $g = \mathcal{O}(n^3)$. Note that the estimate for gate complexity with 1-ancilla may not be optimal. 

The depth can be reduced to $\mc O (n)$ provided $[n/2]$ ancillary qubits. The result is summarised in \autoref{prop:control_V_NN}.  

\begin{proposition}
\label{prop:control_V_NN}
When restricted to linear nearest-neighbour qubit connectivity, a controlled exponentiation Ctrl-$e^{i \theta {H}}$ with molecular Hamiltonian ${H}$ in \autoref{eq:molecular_hamil_main} can be implemented in depth $\mathcal{O}(n)$ using $[n/2]$ ancillary qubits. The circuit implementation is illustrated in   \autoref{fig:control_V_NN}.  When restricted to only one ancillary qubit, the circuit depth is $\Omega(n^2)$.
\end{proposition}

A graphic proof is shown in \autoref{fig:control_V_NN}. Specifically, 
\autoref{fig:control_V_NN}(a,b) shows the overall structure of the circuit with $\mc O(n)$ depth.
\autoref{fig:control_V_NN}(a) shows how to copy the classical information of $A_1$ to the rest of ancilla and finally give it back to $A_1$. 
By using the circuit (example) illustrated in \autoref{fig:control_V_NN}(c) for 1D architecture and \autoref{fig:control_V_NN}(d) for 2D architecture, the controlled-$H$ can be realised with only NN operations.

When considering an arbitrary architecture, the gate count can be further reduced from $\mathcal{O}(n^2) $ to $\mathcal{O}(n \log n)$ using the results in \cite{low2018hamiltonian}. Note that in \cite{childs2021theory,kivlichan2018quantum,babbush2018low}, the potential term
$\hat V$ is a symmetric translationally-invariant two-body coupling term, specifically $\hat V_{p q}=\sum_{\nu \neq 0} \frac{2 \pi}{\Omega k_\nu^2}  \cos \left(k_\nu \cdot r_{p-q}\right) \hat{n}_p \hat{n}_q$ where  momentum modes are defined as  $k_\nu=2 \pi \nu / \Omega^{1 / d}$, $r_m =m (2 \Omega / N)^{1 / d}$ with momemtum difference $m = p-q$, and the computational cell volume $\Omega$. This reduction from $\mathcal{O}(n^2)$ to  $\mathcal{O}(n \log n)$ uses the translationally-invariant property.
The exponentiation of $e^{-i x \hat{T}}$ can be implemented by FFFT,  $\textrm{FFFT} e^{-i \sum_i t_i \hat{n}_i} \textrm{FFFT}^{\dagger} $, which can be implemented with $g = \mathcal{O}( n \log n)$  and $d = \mathcal{O}(\log n)$. The exponentiation of  $e^{-i x \hat{V}}$ can be implemented with $d = \mathcal{O}(n)$ and $g = \mathcal{O}(n\log n)$.

In addition,   the results for molecular Hamiltonians in \autoref{eq:molecule_general_MCH} are displayed in \autoref{table:hamil_dep_cost} where the dominant cost comes from the potential terms.

Table V in Supplementary Sec. I displays the total gate and depth complexity concerning both \textsc{cnot} and T gates, which can be obtained by using the cost within each segment shown in \autoref{table:hamil_dep_cost}.
We can readily find that our method performs better in circuit depth in all these cases when qubit connectivity is restricted to an NN architecture.

In addition to the \textsc{cnot} gate count estimates in the main text, the corresponding T gate count is shown in \autoref{fig:Heis_NC_Eps}. The search for the eigenenergies of the XXZ Heisenberg model is presented in \autoref{fig:IBM_Methods}.

References in the Supplementary Materials include \cite{haah2021quantum,ross2016optimal,bocharov2015efficient,mosca2021polynomial,maslov2016advantages,childs2019faster,cai2023stochastic}.









\vspace{6pt} 
\noindent
\textbf{Acknowledgments:}
J.S. thanks Xiao Yuan for useful discussion and support on the project, and Yukun Zhang, Andrew Green for useful discussions on the manuscript.
J.S. would like to express gratitude to Georg Schusteritsch for providing the molecular Hamiltonians used in this work.
We acknowledge use of the IBM quantum cloud experience for this work.
The views expressed are those of the authors and do not reflect the official policy or position of IBM or the IBM quantum team.
\noindent\textbf{Funding:}
T.G. thanks  UKRI MR/S031545/1, and EPSRC EP/X018180/1 and EP/W032635/1.
J.S. and M.S.K. acknowledge funding from the UK EPSRC through EP/Z53318X/1, EP/W032643/1 and EP/Y004752/1, the KIST through the Open Innovation fund and the National Research Foundation of Korea grant funded by the Korean government (MSIT) (No. RS-2024-00413957).
\textbf{Author contributions:}
J.S. and P.Z. developed the theoretical aspect of the project. J.S. performed the resource estimations under the supervision of P.Z.. J.S., P.Z. and T.G. contributed to the complexity analysis. J.S. performed simulations on the IBM cloud under the supervision of M.S.K.. T.G. and M.S.K. supervised the project. 
J.S., Z.P. and T.G. wrote the original draft. All the authors contributed to the discussion and writing up the manuscript.
\noindent\textbf{Competing Interests:}
The authors declare that there are no competing interests.
\noindent\textbf{Data and materials availability:}
All data needed to evaluate the conclusions are present in the paper and the Supplementary Materials. 
The source data files and code are available on Zenodo \url{10.5281/zenodo.17854694} and GitHub \url{https://github.com/jinzhao-sun/RLCU}.

 \clearpage

\widetext
 
\renewcommand{\thesection}{S\arabic{section}}
\renewcommand{\theequation}{S\arabic{equation}}
\renewcommand{\thefigure}{S\arabic{figure}}
\renewcommand{\thetable}{S\arabic{table}}

\newpage

\setcounter{page}{1}

\section*{Supplementary Materials for "High-precision and low-depth quantum algorithm design for eigenstate problems"}

\renewcommand{\addcontentsline}[3]{\oldaddcontentsline{#1}{#2}{#3}}

\tableofcontents

\newpage
\vspace{10pt}

In \autoref{appendix:A}, we provide an overview of quantum algorithms for ground state preparation and energy estimation, followed by an overview of the theoretical results in this work. In the remaining supplementary sections, we present the technical details that support the main results of this work.
In \autoref{appendix:spectral_filter_A}, we discuss the construction of the spectral filter with composite LCU. We discuss the error in observation estimation given an LCU form.
In \autoref{sec:observ_estimate}, we discuss the building blocks of the eigenstate property estimation algorithm outlined in \autoref{alg1_main}. We will follow the proof idea illustrated in \autoref{fig:section_connection} and Methods to prove the main theorems for eigenstate property estimation in the main text.
In \autoref{sec:En_estimate}, we prove the result for eigenenergy estimation in the main text.
\autoref{sec:gate_gs_BE} analyses the circuit compilation cost for block-encoding-based methods.
\autoref{sec:gate_gs_QSP} and \autoref{sec:gate_gs_QPE} discuss the cost when using QSP and QPE.





\section{Comparison with other works} 
\label{appendix:A}
\subsection{Overview of quantum algorithms for ground state preparation}

In the main text, we have provided an overview of the quantum algorithms for ground state preparation and property estimation. Here, we reviewed the progress in this rapidly growing field.  
Among these quantum algorithms, spectral filter-based methods offer a rigorous and deterministic solution, with clear assumptions regarding the initial overlap and energy gap, which will be the major focus of this work.
In spectral filter methods, QSP has achieved near-optimal query complexity $\mathcal{O}( \Delta^{-1} \log(\varepsilon^{-1}) )$ for target precision $\varepsilon$, as one of the state-of-the-art algorithms~\cite{lin2020near}. 
While it is favourable in the long term, QSP hinges on querying the block encoding of the Hamiltonian $H$, which is challenging for noisy intermediate-scale quantum (NISQ) or early fault-tolerant quantum computing (FTQC) applications.
In the early FTQC regime~\cite{lin2021heisenberg,katabarwa2023early}, where the number of logical qubits is limited, minimising controlled operations and achieving low circuit depth are essential goals.

Considering the feature of early FTQC~\cite{lin2021heisenberg}, there has been considerable progress on quantum algorithms for ground-state property estimation~\cite{wang2023faster,lin2021heisenberg,zeng2021universal,lu2021algorithms,zhang2022computing,huo2021shallow,he2022quantum,wang2023quantum,ding2024quantum,ding2023even,wang2023faster}. 
In 2021, Lin and Tong proposed a spectral filter algorithm based on random sampling~\cite {lin2021heisenberg}, which achieved the Heisenberg limit for ground-state energy estimation. Using similar techniques, Zeng~\textit{et~al.}~\cite{zeng2021universal} and Zhang~\textit{et~al.}~\cite{zhang2022computing} showed that the time complexity can achieve logarithmic in precision $\log(1/\varepsilon)$ for the ground state property estimation problem. Note that Ref.~\cite{zeng2021universal} extends to eigenstate property estimation. Wan~\textit{et~al.} proposed the randomised algorithm for ground energy estimation, which has a relatively worse gap dependence~\cite{wan2021randomized}. 
On the other hand, following the spirit of QSP, QETU was proposed~\cite{dong2022ground} to prepare the ground state, which achieves near-optimal asymptotic scaling but avoids querying the block-encoding of the Hamiltonian by querying time evolution.
The above algorithms assume the usage of only one ancillary qubit and achieve good asymptotic sample and time complexity, which are competitive for the application of noisy intermediate-scale quantum and early FTQC.
However, these algorithms assume perfect and efficient queries to the time evolution operator $e^{-iHt}$.
It remains unclear whether the good properties of the above early FTQC algorithms can be preserved if we further expand the time evolution operator into elementary gates.
For example, if the time evolution operator is implemented using Trotterisation methods, this will eliminate the advantage of logarithmic scaling in precision. On the other hand, if we introduce advanced block-encoding-based methods to realise the time evolution operator, it needs many ancillary qubits and nonlocal controlled gates, which violates the spirit of early FTQC. 
The system-size dependence of the algorithms is rarely discussed in existing works, as it highly depends on the detailed circuit-level implementation of $e^{-iHt}$ as well as the qubit connectivity of the device.
The central objective of this work is to design full-stack algorithms with high-precision and low-depth features towards   NISQ and early FTQC applications.

\textbf{Summary of the main results.} 
We provide end-to-end gate complexity analysis of the eigenstate property and energy estimation task. Our randomised LCU approach achieves nearly logarithmic scaling on inverse precision, with improved scaling on $\Delta$ and $\lambda$ compared to oracle-free methods.  We have built up a framework based on randomised composite LCU formulae, which contains several hierarchies,  for analysing the gate complexity for the eigenstate problems. 
The comparison with other methods is shown in \autoref{table:ObservableComp} with detailed comparison in \autoref{table:ObservableComp_Methods} and \autoref{table:EnComp}. For eigenstate property estimation, we focus on the comparison for the maximum gate complexity per circuit, while in the literature, the gate complexity may refer to the total number of gates required in the whole algorithm.
For electronic structure problems specified in \autoref{eq:molecular_hamil_main} which is usually compared in the existing literature, the gate complexity scales $\tilde {\mathcal{O}} ( n(\lambda \Delta^{-1} \log \varepsilon^{-1})^{1+ \frac{1}{4k+1}}  )$, which is nearly linear in $n$ (excluding the dependence on $\lambda$) and logarithmic in inverse precision. We note that compared to~\cite{reiher2017elucidating,lee2021even,babbush2018encoding,babbush2018low}, this result is for the depth or gate count in a single-run circuit, which has worse sample complexity.

In addition to addressing gate complexity for generic Hamiltonians, the second contribution of our work is the advantages in circuit depth for various physical problems. These include lattice models and second-quantised plane-wave electronic structures with $n$ spin orbitals and $L = \mathcal{O}(n^2)$ terms in the Hamiltonian. Notably, the 2D Fermi-Hubbard model falls within the problem class. 
In theory, for problems that conserve certain symmetries, we demonstrate the ability to achieve high-precision and low-depth eigenstate property estimation without the need for ancillary qubits. To accomplish this, we design new Trotter and Trotter-error compensation circuits that maintain the system's symmetries.
The random sampling approach can exploit the commutation relation of the target Hamiltonian terms to reduce the gate complexity, outperforming QSP for the 1D lattice model. Moreover, our approach favours a linear nearest-neighbour  architecture. When restricted to nearest-neighbour architecture, for the 1D lattice model, the circuit depth scales $d = \mathcal{O} (n^{\frac{2}{4k+1}})$, while the circuit depth in the QETU-Trotter method scales as $d_{\mathrm{QETU}} = \mathcal{O} (n^{1 + \frac{1}{2k}})$. For electronic structure problems specified in \autoref{eq:molecular_hamil_main}, when considering the commutation relation and restricting to a nearest-neighbour architecture, the circuit depth is $d = \mathcal{O} (n^{2+\frac{2}{4k+1}})$, while QETU requires $d = \Omega (n^{3+\frac{1}{2k}})$. A side product is that controlled $e^{-i \theta {H}}$ can be implemented by a linear-depth circuit $d = \mathcal{O}(n)$, summarised in \autoref{prop:control_V_NN}, comparable to the control-free simulation of electronic problems in \cite{kivlichan2018quantum,babbush2018low}. This result can be directly useful for other quantum algorithms which require controlled unitaries as a subroutine.
    
In the existing resource estimation works~\cite{kim2022fault,reiher2017elucidating,goings2022reliably,wecker2014gate,campbell2019random}, the energy estimation or eigenstate preparation for practical problems, such as second-quantised quantum chemistry problems with $L = \mathcal{O}(n^4)$ terms~\cite{lee2021even,babbush2018encoding,berry2019qubitization} and condensed-phase electrons~\cite{kivlichan2020improved,campbell2021early}, is predominantly based on phase estimation. 
A typical strategy is to encode the eigenspectra of the Hamiltonian in a unitary for phase estimation by the evolution $e^{-iHt}$, which is synthesised by Trotterisation~\cite{kivlichan2020improved,campbell2021early}, or a qubitised quantum walk~\cite{babbush2018encoding,lee2021even} with eigenspectrum proportionally to $e^{\pm i \arccos(H/\lambda)} $, where $\lambda$ is a parameter related to the norm of the Hamiltonian.
However, due to the cost of phase estimation, the circuit depth will inevitably be polynomial in the precision, which is not optimal for eigenstate energy and property estimation. \revise{Recent works~\cite{wan2021randomized,kiss2025early} have also estimated resources for algorithms based on the computation of the cumulative distribution function (CDF) of the spectral measure of a Hamiltonian which was originally proposed in \cite{lin2021heisenberg}.}

The third contribution is to present resource estimations with the method developed in this work, which has better asymptotic scaling in precision and system size in theory, and also has a smaller circuit compilation overhead in practice.
In the NISQ applications, the major bottleneck is the number of two-qubit gates for the noisy quantum computer. The major overhead for the error-corrected quantum computer is the T gate count~\cite{babbush2018encoding}, which requires more gates to perform error correction than Clifford gates. 
We provide a systematic comparison with the existing advanced methods, which include QPE combined with Trotter formulae or qubitised quantum walk (as have been used in \cite{lee2021even,babbush2018encoding,goings2022reliably}),  QSP~\cite{lin2020near} and QETU. {
The \textsc{cnot} gate cost for a 20-site Heisenberg model is about $3 \times 10^5$ while the T gate cost is about $6 \times 10^6$.}

For different query models, we provide a useful toolbox for researchers to analyse the individual costs for elementary units (like block encoding and controlled real-time evolution), and thus enable comparison across different eigenstate preparation methods with various initial conditions. We hope that our framework (decomposing the task into the elementary operations) and the toolbox for analysing the cost for each elementary operation can be useful as a building block for resource estimations for other quantum algorithms based on querying the block encoding of $H$ or the time evolution.
We analyse the actual gate count for typical problems and compare it with the state-of-the-art methods.

\subsection{Problem setup and comparison with existing works}

\noindent
\textbf{Problem setup.}
Consider a gapped $n$-qubit Hamiltonian with a Pauli decomposition $H = \sum_{l=1}^L \alpha_l P_l := \lambda \sum_{l=1}^L \tilde{\alpha}_l P_l$ where $P_l$ is a Pauli operator, $\lambda := \sum_l |\alpha_l| $, and $\tilde{\alpha} := \alpha_l/\lambda$.
The eigenstate property estimation task is to estimate the expectation value of an observable $O$ on the $j$th eigenstate of $H$, $\ket{u_j}$. The assumptions are the following. We assume that we have an estimate of $E_j$, $\hat{E}_j$ with a small estimation error $\kappa:= |\hat{E}_j - E_j|$. We assume that we have a good initial state $\ket{\psi_0}$ which has a nonvanishing overlap with the target eigenstate, $\eta:= |\braket{\psi_0|u_j}|^2 = \Omega(1/\poly(n))$. We assume a nonvanishing energy gap $\Delta:= \min (E_{j+1}-E_j, E_j - E_{j-1} )$.

{Note that in \autoref{problem:observ}, we follow the convention in~\cite{zeng2021universal}, which has a slight difference in the definition of the initial overlap $\eta$ compared to Refs.~\cite{lin2020near,dong2022ground}. }

Below, we provide relatively formal descriptions for \autoref{thm:observ_estimation_main} and \autoref{thm:observ_estimate_gate_Lattice_main}, respectively.

\begin{theorem}[Eigenstate property estimation for generic Hamiltonians (formal version of \autoref{thm:observ_estimation_main})]
\label{thm:observ_estimation_main_methods}

Suppose we use the method in \autoref{alg1_main} and the conditions and assumptions in \autoref{problem:observ} hold.
Observable estimation (\autoref{problem:observ}): To achieve the error of observable's expectation on the eigenstate $\ket{u_j}$ within $\varepsilon$, the gate complexity in a single circuit is $\mathcal{O} \left( 5^{k-1} L (\lambda \Delta^{-1} \ln(\eta^{-1} \varepsilon^{-1}))^{1+ \frac{1}{4k+1} }  \right )$ when the number of samples is $N_s = \mc O\left ( {\eta ^{-2}\varepsilon^{-2}} \|O\|_1^2 \ln(1/\vartheta) \right)$  with a success probability at least $1-\vartheta$.

Eigenenergy estimation (\autoref{problem:eigenenergy}): To achieve the eigenenergy estimation error within $\kappa$ using \autoref{alg1_main}, the  gate complexity in a single circuit is $\mathcal{O} \left( 5^{k-1} L (\lambda  \kappa^{-1} \ln(\eta^{-1}))^{1+ \frac{1}{4k+1} }  \right )$, with number of samples $N_s = \mc O\left ( \eta ^{-2}  \ln(1/\vartheta) \right)$ (independent of $\kappa$) with a success probability at least $1-\vartheta$, approaching to the Heisenberg limit. Alternatively,  by using the methods proposed in \cite{wang2023quantum} and \autoref{alg1_main},  the gate complexity in a single circuit is $\mathcal{O} \left( 5^{k-1} L (\lambda \Delta^{-1} \ln(\eta^{-1} \kappa^{-1}))^{1+ \frac{1}{4k+1} }  \right )$ at the cost of more number of samples $N_s = \mc O\left (\eta^{-2} \Delta^{4}\kappa^{-4} (\ln ({\kappa^{-2} \eta^{-1}}) )^{2} \ln (1/\vartheta) ) \right)$  with a success probability at least $1-\vartheta$.

 \end{theorem}

\begin{theorem}[Gate and depth complexity for lattice Hamiltonians]
\label{thm:observ_estimate_gate_Lattice_main_Methods}
For $n$-qubit Heisenberg Hamiltonians, to estimate the observable on the eigenstate with a precision $\varepsilon$ and a success probability $1 -\vartheta$, in a single run, the gate complexity is $\mathcal{O}(n^{1+\frac{2}{4k+1}} \Delta^{-(1+\frac{1}{4k+1})} \varepsilon^{-\frac{1}{4k+1}} \log(\vartheta^{-1}))$. The circuit depth when compiled on qubits with nearest-neighbour geometry is $\mathcal{O}(n^{\frac{2}{4k+1}} \Delta^{-(1+\frac{1}{4k+1})} \varepsilon^{-\frac{1}{4k+1}} \log(\vartheta^{-1}))$.
\end{theorem} 

More detailed versions with the actual gate overheads are presented in \autoref{thm:observ_estimation_detailed} and \autoref{thm:En_estimation_detailed}, respectively. The eigenenergy is first assumed to be known a priori. This is the case for linear algebra tasks, in which the true solution corresponds to the ground state of a constructed Hamiltonian with the eigenenergy $E_0$ being exactly zero. The task with an unknown eigenenergy will be discussed in Supplementary \autoref{sec:unknown_energy}. 

The comparisons with advanced methods for eigenstate property and eigenenergy estimation are displayed in \autoref{table:ObservableComp_Methods} and \autoref{table:EnComp}, respectively.
\revise{The near-optimal eigenstate property estimation is indeed enabled by our algorithm design. Specifically, our algorithm design does not rely on a coherent implementation of phase estimation, allowing the use of different Hamiltonian simulation strategies. }

\revise{As the implementation of real-time evolution in Fig. 1(a2) is a subroutine in the eigenstate algorithm, one could employ the recent advances in Hamiltonian simulation algorithms. There are various approaches that can achieve favourable scaling in the system size and other key parameters, e.g. \cite{haah2021quantum,rendon2024improved,hejazi2024better,watson2025exponentially}. }

The reason why the Trotter error compensation method is chosen and its particular suitability are the following: 1)  It integrates naturally into the randomised linear-combination-of-unitaries construction; 2) The actual overhead is small -  the actual overhead with prefactors are explicitly calculated in Theorem 5 and Theorem 7; 3) In each circuit run, the circuit structure is deterministic, determined by sampling the operations according to a well-defined randomisation procedure.
For the interleaved forward and backward time-evolutions \cite{haah2021quantum}, they have shown the scaling advantage but the drawback is that the actual gate count may be even worse than that of either standard Trotter or QSP in small sizes. 

In addition, in an ongoing project, we expect that the Trotter error compensation may indeed simultaneously achieve optimal scaling in the system size and precision for $k$-local Hamiltonians. Therefore we expect that the gate complexity in Theorem 2 may be improved to $\mathcal{O}(n^{1+\frac{2}{4k+1}} \log(\varepsilon^{-1}))$

\revise{
It may be worth noting here about the advantages of the above Point 3. For each sampled time $t$, we do not need to change the step-size as in \cite{watson2025exponentially} (in which the stepsize $\{s_i\}_{i=1}^m$ needs to be changed for $m$ times). This has a two-fold advantage (1) it does not incur additional measurement overhead could bring bias and fluctuations to the estimated result). (2) in theory single-shot measurement is sufficient (i.e. we do not require the expectation value of $\braket{O(t)}$).  On the other hand, extrapolation-based methods typically require additional measurements at multiple evolution times, which increases the variance of observable estimation and also may be more susceptible to measurement noise. }

As noted in the main text, the spectral filter can be constructed with the Trotter expansion order $k = 0$ (i.e., without the Trotter term $S$). The corresponding gate complexity is covered in both tables. Our zeroth-order design with $k = 0$ is similar to that of \cite{wan2021randomized}, though our sampling process is simpler, as we only need to sample the first-order terms. In contrast, Ref.~\cite{wan2021randomized} employs full-order pairing and needs to sample from higher-order terms. Another difference is that our zeroth-order design may be advantages when qubit connectivity is restricted. 

To summarise, our work provides a full-stack solution, from high-level query-based design down to end-to-end algorithmic design. This shifts in perspectives which in our view is an appropriate approach for bridging near-term capabilities and long-term goals of fault-tolerant quantum computing. It is precisely through adopting this bottom-up methodology that we are able to achieve the near-optimality at the gate level, and are able to demonstrate the deterministic eigenstate solution on current quantum devices.
 

\begin{table*}[ht!]
\centering
\begin{tabular}{c c c c c} 
\hline
\hline
Methods              & Gate complexity   & Depth complexity (lattice models)    & Extra qubits   \\
\hline
QPE + Trotter ($2k$th-order)    & $ {\mathcal{O}}( L\tilde \Delta^{-(1+\frac{1}{2k})} \varepsilon^{-(1 + \frac{1}{2k})} )$  & 
 $ {\mathcal{O}}( n^{1+\frac{1}{2k}} \Delta^{-(1+\frac{1}{2k})} \varepsilon^{-(1 + \frac{1}{2k})} )$ & $\log(\varepsilon^{-1})+\log (\Delta^{-1} )$     \\
QPE + QW~\cite{lee2021even}   & $ \tilde{\mathcal{O}}( L \tilde \Delta^{-1} \varepsilon^{-1})$ & $ {\mathcal{O}}( n^2 \Delta^{-1} \varepsilon^{-1})$   & $\log(L) + \log(\varepsilon^{-1})+\log(\Delta^{-1})$      \\
QSP~\cite{lin2020near} & $ {\mathcal{O}}( L \tilde \Delta^{-1} \log\varepsilon^{-1})$ & $ {\mathcal{O}}( n^2 \Delta^{-1} \log\varepsilon^{-1})$ &   $\log(L) + \log(\varepsilon^{-1})+\log(\Delta^{-1})$ \\
QETU~\cite{dong2022ground} & $\tilde {\mathcal{O}}( L \tilde \Delta^{-(1+\frac{1}{2k})} \varepsilon^{-\frac{1}{2k}})$ & $ {\mathcal{O}}(n^{1+\frac{1}{2k}}\Delta^{-(1+\frac{1}{2k})} \varepsilon^{-\frac{1}{2k}})$  & 1 \\
This work ($2k$th-order)            &  $\tilde{\mathcal{O}}(L  \tilde\Delta^{-(1+\frac{1}{4k+1})} \log \varepsilon^{-1}) )$ & $\tilde {\mathcal{O}}(n^{\frac{2}{4k+1}}  \Delta^{-(1+\frac{1}{4k+1})}  \varepsilon^{-\frac{1}{4k+1}})$  & 0 or 1     \\
    (zeroth-order)          &  $\mathcal{O}(\tilde\Delta^{-2} \log^2 \varepsilon^{-1})$  & $\mathcal{O}(n^2 \Delta^{-2} \log^2 \varepsilon^{-1})$ & 0 or 1     \\
\hline
\hline
\end{tabular}
\caption{Comparison of observable estimation on the eigenstate of a Hamiltonian (\autoref{problem:observ}). Here we compare the gate complexity of the algorithm in a single coherent run. The results in the second column in the table are based on \autoref{alg1_main} and \autoref{thm:observ_estimation_main} using the $2k$th-order Trotter error compensation. Here, $\tilde\Delta := \lambda/\Delta$.
The results in the third column in the table are based on \autoref{thm:observ_estimate_gate_Lattice_main} when the $2k$th-order Trotter formula is used.
The eigenenergy is assumed to be known a priori and the unknown eigenenergy case is analysed in \autoref{sec:unknown_energy}.
The dependence on $\eta$ is not included in the table since it only appears in sample complexity for our method.
Similar to other random-sampling spectral filter methods (see \cite{zhang2022computing,lin2021heisenberg}), the sample complexity with respect to $\eta$ is $\mathcal{O}(\eta^{-2})$ for which the optimal scaling is $\mathcal{O}(\eta^{-1/2})$ achieved by QSP and QETU with amplitude amplification.
As noted in the main text, one may simultaneously achieve near optimal scaling in both the size and precision as $\mathcal{O}( n^{\frac{2}{4k+1}} \log\varepsilon^{-1})$ if the higher-order commutators in the Trotter error remainder could be compensated.
}
\label{table:ObservableComp_Methods}
\end{table*}

\begin{table*}[t!]
\centering
\begin{tabular}{c c c c c} 
\hline
\hline
Methods              & Gate complexity   & Depth complexity (lattice models)    & Extra qubits   \\
\hline
QPE + Trotter ($2k$th-order)    & $ {\mathcal{O}}( L \lambda^{1+\frac{1}{2k}} \varepsilon^{-(1+\frac{1}{2k})}  )$  & 
 $ {\mathcal{O}}( n^{1+\frac{1}{2k}}   \varepsilon^{-(1 + \frac{1}{2k})} )$ & $\log(\varepsilon^{-1})+\log (\Delta^{-1} )$     \\
QPE + QW~\cite{lee2021even}   & $ \tilde{\mathcal{O}}( L \lambda \varepsilon^{-1}  )$ & $ \tilde{\mathcal{O}}( n^2   \varepsilon^{-1})$   & $\log(L) + \log(\varepsilon^{-1})+\log(\Delta^{-1})$      \\
QSP~\cite{lin2020near} & $ \tilde{\mathcal{O}}( L \lambda  \varepsilon^{-1}  )$ & $ \tilde{\mathcal{O}}( n^2 \varepsilon^{-1}  )$ &   $\log(L) +  \log(\varepsilon^{-1})+\log(\Delta^{-1})$ \\
QETU~\cite{dong2022ground} & $ \tilde{\mathcal{O}}( L \lambda^{1+\frac{1}{2k}} \varepsilon^{-(1+\frac{1}{2k})}  )$ & $ \tilde{\mathcal{O}}( n^{1+\frac{1}{2k}} \varepsilon^{-(1+\frac{1}{2k})}  )$  & 1 \\
This work ($2k$th-order)        &  $\tilde {\mathcal{O}}(L   \lambda^{1+\frac{1}{4k+1}}  \varepsilon^{-(1+\frac{1}{4k+1})}    )$ & $\tilde {\mathcal{O}}(n^{\frac{2}{4k+1}}  \varepsilon^{-(1+\frac{1}{4k+1})} )  $  & 0 or 1     \\
    (zeroth-order)   and \cite{wan2021randomized}         &  $\tilde {\mathcal{O}}(\lambda^2 \varepsilon^{-2} )$  & $\tilde {\mathcal{O}}(n^2 \varepsilon^{-2} )$ & 0 or 1     \\
\hline
\hline
\end{tabular}
\caption{Total gate complexity in eigenenergy estimation up to estimation error $\varepsilon$ (\autoref{problem:eigenenergy}). Note that this table presents the total gate complexity, which includes the sample complexity.
The dependence on $\eta$ is not included in the table, as similarly discussed in \autoref{table:ObservableComp}. 
}
\label{table:EnComp}
\end{table*}

\begin{table*}[t!]
\centering
\begin{tabular}{c| c c | c } 
\hline
\hline
Hamiltonians      & Ancilla-free method (NN)     & QETU (NN)  \cite{dong2022ground}   & arbitrary     \\
\hline
1D lattice models  &  $d = \mc O (\mathbf{ n^{\frac{2}{4k+1}} } \Delta^{-(1+\frac{1}{4k+1})}  \varepsilon^{-\frac{1}{4k+1}}) $   & $d = \mc O (n^{1+\frac{1}{2k}} \Delta^{-(1+\frac{1}{2k})}  \varepsilon^{-\frac{1}{2k}}) $  &  $d = \mc O (n^{\frac{2}{4k+1}} \Delta^{-(1+\frac{1}{4k+1})}   \varepsilon^{-\frac{1}{4k+1}} )$   \\
    &   $g = \mc O (n^{1+\frac{2}{4k+1}} \Delta^{-(1+\frac{1}{4k+1})}  \varepsilon^{-\frac{1}{4k+1}}) $   &  $g = \mc O (n^{1+\frac{1}{2k}} \Delta^{-(1+\frac{1}{2k})}  \varepsilon^{-\frac{1}{2k}}) $ & $g = \mc O (n^{1+\frac{2}{4k+1}} \Delta^{-(1+\frac{1}{4k+1})}  \varepsilon^{-\frac{1}{4k+1}}) $ \\
    \hline
Electronic structure    & $d = \mc O (\mathbf{ n^{2 + \frac{2}{4k+1}} } \Delta^{-(1+\frac{1}{4k+1})}  \varepsilon^{-\frac{1}{4k+1}})$  & $d = \Omega (n^{3+\frac{1}{2k}} \Delta^{-(1+\frac{1}{2k})}  \varepsilon^{-\frac{1}{2k}}) $ & $d = \mc O (n^{2+\frac{2}{4k+1}} \Delta^{-(1+\frac{1}{4k+1})}  \varepsilon^{-\frac{1}{4k+1}})$ \\
 (\autoref{eq:molecular_hamil_main})   &   $g = \mc O (n^{3+\frac{2}{4k+1}} \Delta^{-(1+\frac{1}{4k+1})}  \varepsilon^{-\frac{1}{4k+1}}) $ &  $g = \mc O (n^{4+\frac{1}{2k}} \Delta^{-(1+\frac{1}{2k})}  \varepsilon^{-\frac{1}{2k}}) $  &  $g = \tilde{\mc O} (n^{2+\frac{2}{4k+1}} \Delta^{-(1+\frac{1}{4k+1})}  \varepsilon^{-\frac{1}{4k+1}}) $ \\
\hline
\hline
\end{tabular}
\caption{Gate complexity with respect to the energy gap $\Delta $, target precision $\varepsilon$, and system size $n$ for different Hamiltonians' eigenstate property estimation.
In the second and third columns, the qubit connectivity is restricted to a linear nearest-neighbour (NN) architecture. The fourth column is the result of our method when there is no restriction on qubit connectivity.
In this table, the commutation relation of the Hamiltonian terms is used to improve the system-size scaling.
It is worth mentioning that without restriction on connectivity, the gate complexity of electronic structure problems studied in this table can be $g = \tilde{\mc O} ( n (\lambda \Delta^{-1} \log \varepsilon^{-1})^{1+\frac{1}{4k+1}} )  $ which is logarithmic in inverse precision. 
The circuit depth results for different physical Hamiltonians are displayed in \autoref{table:dep_cost_tot}. The gate-complexity estimate for electronic Hamiltonians using QETU may not be optimal. QSP and methods based on full-order pairing \cite{wan2021randomized} have worse system-size scaling, and thus they are not included in \autoref{table:dep_cost_tot}.
}  
\label{table:dep_cost_tot}
\end{table*}

\section{Construction of the spectral filter}
\label{appendix:spectral_filter_A}

\subsection{Composite LCU formulae for decomposing a nonunitary operator} \label{ssc:LCUdecompNonunit}

  
In the main text, we have introduced a few tools to analyse the property of the composite form of an LCU formula. \autoref{prop:composite_LCU} shows how to bound the error of a composite LCU formula written in the discretised form. We provide the proof here.

\begin{proof}(of \autoref{prop:composite_LCU})
	
Recall that the LCU formula of $g$ can be written as
\begin{equation}
     g_2 =\mu_1   \sum_i \Pr(i) (U(t_i/\nu))^\nu.
\end{equation}	
Given the  LCU formula of $U(t_i/\nu)$, 
$g_2$ can be written as
\begin{equation}
    g_2 =\mu_1 \mu_2^\nu \sum_i \Pr(i) \left( \sum_{i_{\ell}} \Pr(i_{\ell}) P_{i_{\ell}} \right)^\nu = \mu_1 \mu_2^\nu \sum_i \Pr(i)  \sum_{ \{i_{\ell} \}} \prod_{q = 1}^{\nu}\Pr(i_{\ell_q}) \prod_{q = 1}^{\nu} P_{i_{\ell_q}} 
\end{equation}
with some abuse of notation.

Given a $(\mu_2, \varepsilon_2)$-LCU form of $U(\delta t)$, one can prove that the product of LCU formula $U(\delta t)^{\nu}$ has   a normalisation factor $\mu_2' = \mu_2^\nu$ and an error  $\varepsilon_2' \leq \nu \mu_2' \varepsilon_2$. 
We can prove the result by using the triangle inequality $$\|g_2 - g\| \leq  \|g_1 - g\| + \|g_2 - g_1\| \leq \varepsilon_1 + \nu \mu_1  \mu_2^\nu \varepsilon_2,$$
which completes the proof.
\end{proof}

To implement LCU in practice, we could consider either a discretised form or a continuous form. As shown in the main text,  the Fourier transform gives an explicit form for decomposing the spectral filter into unitary operators in a continuous form.
When the integral form of a spectral filter has a well-defined probability distribution, it can be well-characterised by \autoref{eq:randomLCU_main}.
Therefore, we use the continuous form for the Gaussian spectral filter. 
We provide a discretised version of the spectral filter in \autoref{sec:discretised_version} and show  that the discretisation error for the Gaussian spectral filter can be sufficiently small.

We provide the composite LCU formula in a continuous form for completeness.

\begin{proposition}[Composite LCU formula in a continuous form]
\label{prop:composite_LCU_continous}
	Suppose $ g_1$ is a  $(\mu_1,\varepsilon_1)$-LCU formula of $g$,
	\begin{equation}
		g_1 = \mu_1 \int dx\, p(x)U(x)
	\end{equation} Suppose that each of the summand $U(x)$ has a   $(\mu_2,\varepsilon_2)$-LCU formula, 
	\begin{equation}
		\tilde U(x) = \mu_2(x)  \int dy\, q(x,y) V(x,y)
	\end{equation}
	Then the formula
	\begin{equation}
		 g_2 = \mu_1 \int dx\, p(x)\tilde U(x)  	\end{equation}
is a $(\mu,\varepsilon)$-LCU formula of $g$, with $\mu  := \mu_1 \int_{-\infty}^{\infty} p(x) \mu_2(x) dx
$, and $\varepsilon = \varepsilon_1 + \mu_1 \varepsilon_2$.
\end{proposition}

\begin{proof}(of \autoref{prop:composite_LCU_continous})
	
The formula of $g_2$ can be written as
	\begin{equation}
		 g_2 =\mu_1 \int dx p(x) \mu_2(x) \int dy q(x,y) V(x,y)
	\end{equation}	
	
We define
\begin{equation}
	p_{\mu}(x) = \mu_1 \mu_2(x) p(x) /\mu
\end{equation}
with $\mu  := \mu_1 \int_{-\infty}^{\infty} p(x) \mu_2(x) dx
$.
Then $g_2$ takes the form of
	\begin{equation}
		 g_2 = \mu\int dx p_{
		 \mu}(x)  \int dy q(x,y) V(x,y).
   \label{eq:composite_form}
	\end{equation}	
	
The proof is straightforward by using the triangle inequality $\|g_2 - g\| \leq  \|g_1 - g\| + \|g_2 - g_1\| \leq \varepsilon_1 + \mu_1 \varepsilon_2$.
	Note that $p_{\mu}(x)$ and  $q(x,y)$ (for given $x$) are both normalised and can thus be regarded as probability distributions. Therefore, $g_2$ can be realised in a random-sampling way by sampling from the distribution $p_{\mu}(x)$ and then $q(x,y)$.

\end{proof}

\autoref{eq:composite_form} is a general $(\mu,\varepsilon)$ composite LCU form of $g$.
As we shall see in the later discussion, to reduce the maximum evolution time (related to $x$) we usually set a truncation of $x$ in the integral, i.e.,
\begin{equation}
    g_1 = \mu_1 \int_{-x_c}^{x_c} dx p(x) U(x)
\end{equation}
with  $\tilde \varepsilon_1 = \varepsilon_1 + \varepsilon_c$ and $\varepsilon_c$ being a small truncation error.
Then we may set the constant $\mu_2 $ to be $\mu_2 = \max_x 
\mu_2(x)$ in the LCU formula of 
 $U(x)$,
$ 
		\tilde U(x) = \mu_2  \int dy\, q(x,y) V(x,y).
$
Then  the LCU formula of 
 $g$ could be simplified as
	\begin{equation}
		 g_2 = \mu\int dx p
		 (x)  \int dy q(x,y) V(x,y),
   \label{eq:composite_form_const_mu}
	\end{equation}	
with $\mu = \mu_1 \mu_2$.

\subsection{Spectral filter by randomised composite LCU formulae}

Recall that we choose a Gaussian spectral filter $g_{\tau }(H) = e^{- \tau ^2 H^2}$. Below, we elaborate on a few properties of the spectral filter.
A general matrix function acting on the Hamiltonian is defined as
 \begin{equation} \label{eq:g_H}
    g( H ) := \sum_{i=0}^{N-1} g (E_i ) \ket{u_i}\bra{u_i}.
\end{equation}
where $g(h) : \mathbb R  \rightarrow \mathbb C$ is a generic continuous-variable function determining the transformation of the spectrum of the Hamiltonian.
As a spectral filter, the function  $g(h)$ is required to satisfy strictly non-increasing absolute value, $|g(h')|< |g(h)|,~\forall |h'|> |h|$, and vanishing asymptotic value, $\lim_{\tau\rightarrow\infty}|g(\tau h')/g(\tau h)|=0,~\forall |h'|> |h|$, and is an even function, $g(h) = g(-h)$. In this work, we choose the Gaussian function $g(h)=e^{-h^2}$, corresponding to a generalised imaginary-time evolution $g(\tau H)=e^{-\tau^2 H^2}$.  
As shown in the main text, the Gaussian spectral filter is decomposed into the basis of real-time evolution and is further decomposed into elementary operations, either Pauli operators for general cases or symmetry-conserved operators for ancilla-free consideration.

Given an LCU form of $g$, \autoref{prop:RLCU_err_observ_bound} shows the error in observable estimation. Below we provide the proof of \autoref{prop:RLCU_err_observ_bound}.

\begin{proposition}[Formal version of \autoref{prop:RLCU_err_observ_bound}] 
For a target operator $g$ and its $(\mu,\varepsilon)$-randomised LCU formula defined in \autoref{eq:randomLCU_main}, if we estimate the value on the unnormalised state $N_g(O):= \tr(g\rho g^{\dagger} O)$ with an initial state $\rho$ and observable $O$,   then the distance between the mean estimator value $\hat{O}$ and the true value $N_g(O)$ is bounded by
\begin{equation} \label{eq:Obound}
\varepsilon_N := |\hat N_{\tilde g}(O) - N_g(O)| \leq   \|O\|(2\mu^2 \varepsilon + \varepsilon_n),
\end{equation}
with a success probability $1-\vartheta$. Here, we use the variant of the Hadamard test circuit for $N_s=\mu^4\ln(2/\vartheta)/\varepsilon_n^2$ times and $\|O\|$ is the spectral norm of $O$.
 The error for the denominator is bounded as 
	$\varepsilon_D := |\hat D_{\tilde g}  - D_{g} | \leq 2\mu^2 \varepsilon + \varepsilon_n$.
Given the error  $\varepsilon_D$ and $\varepsilon_N$, the error for the normalised observable expectation $\braket{O}_g = N_g(O)/D_g$ is
\begin{equation}
\begin{aligned}
|\frac{\hat N_{\tilde g}(O)}{\hat D_{\tilde g}}  -\frac{N_g(O)}{D_g} |  \leq \frac{1}{D_g  } ((\braket{O} + 1)\varepsilon_D +  \varepsilon_N)	\end{aligned}.
\end{equation}

\end{proposition}

\begin{proof}(of \autoref{prop:RLCU_err_observ_bound})

Suppose $\tilde{g}$ is a $(\mu,\varepsilon)$-RLCU formula of $g$. 
We first have 
\begin{equation}
	 N_{\tilde g}(O) = \tr(\tilde{g}_{\tau}(H-E_0)O \tilde{g}_{\tau}(H-E_0)) 
\end{equation}
 we have 
\begin{equation}
\begin{aligned}
	| N_{\tilde g}(O) - N_g(O)| &= |\tr(\tilde{g}_{\tau}(H-E_0)O \tilde{g}_{\tau}(H-E_0)) - \tr({g}_{\tau}(H-E_0)O{g}_{\tau}(H-E_0)) | \\
	& \leq \varepsilon (\|\tilde{g}_{\tau} \|+\|{g}_{\tau}\|) \|O\|_{\infty}
	\\
	& \leq 2  \mu^2 \varepsilon\|O\|_{\infty}
\end{aligned}
\end{equation}
Here, we have use the fact that $\|\rho\|\leq 1$, $\|\tilde{g}_{\tau}\|\leq \mu^2$, and $g = g^{\dagger}$.

Suppose we have the estimator $ \hat N_{\tilde g}(O)$ defined in \autoref{eq:estimator_v_main}.
The expectation range is $[-\sqrt{2}\mu^2  \|O\|_{\infty}, \sqrt{2}\mu^2  \|O\|_{\infty}]$.
Using the Hoeffding bound, we have the following probability tail bound for the mean estimator $ \hat N_{\tilde g}(O)$,
\begin{equation}
	\Pr(| \hat N_{\tilde g}(O) -  N_{\tilde g}(O)|) \leq 2 \exp(-\frac{N_s \varepsilon_n^2}{ \mu^4})
\end{equation}
By setting $N_s:=\mu^4\ln(2/\vartheta)/\varepsilon_n^2$, we have the estimation error
\begin{equation}
\varepsilon_N := |\hat N_{\tilde g}(O) - N_g(O)| \leq   \|O\|(2\mu^2 \varepsilon + \varepsilon_n),
\end{equation}
with a success probability $1-\vartheta$.

The result for the denominator can be similarly derived since \begin{equation}
	D_{\tilde g}   := \braket{\tilde{g}_{\tau}^2(H-E_0)} = N_{\tilde g}(I).
\end{equation}
We have
\begin{equation}
\begin{aligned}
	\varepsilon_D : = |\hat D_{\tilde g}  - D_{g} | & \leq 2\mu^2 \varepsilon + \varepsilon_n.
 \end{aligned}
\end{equation}
The error for the normalised observable expectation is
\begin{equation}\label{eq:err_N_div_D}
\begin{aligned}
|\frac{\hat N_{\tilde g}(O)}{\hat D_{\tilde g}}  -\frac{N_g(O)}{D_g} | & = |\frac{(\hat N_{\tilde g}(O) D_g -   N_g(O)D_g ) + ( N_g(O)D_g - \hat D_{\tilde g}(O) N_g)}{\hat D_{\tilde g}D_g}  |   \\
&\leq |\frac{ N_g(O)\varepsilon_D +  D_g \varepsilon_N}{\hat D_{\tilde g}D_g}  |  \\
	& \leq  |\frac{(D_g + N_g(O))\varepsilon_D +  D_g \varepsilon_N}{D_g^2  } | =  \frac{1}{D_g  } ((\braket{O}_g + 1)\varepsilon_D +  \varepsilon_N)
	\end{aligned}
\end{equation}


\end{proof}

The performance of our method in eigenstate property estimation can be evaluated by the error in constructing the RLCU formula. 
We note that the random sampling algorithm is not a deterministic state preparation method, as it cannot prepare the target state $\sigma$. However, when we focus on the property estimation of the target state, our scheme has similar performance to the normal Hamiltonian simulation methods. Specifically, the sample complexities of the former algorithms and the RCLU algorithm to learn the observable properties of the state are similar, as long as the norm of the estimator is a constant. This is guaranteed by \autoref{prop:RLCU_err_observ_bound}.
One can thus compare our method with other deterministic schemes at the same level.   





\section{Eigenstate property estimation}
\label{sec:observ_estimate}

In this section, we provide error analysis for eigenstate property estimation described in \autoref{problem:observ}. We will provide the proof of gate complexity (\autoref{thm:observ_estimation_main}) in the main text.

\subsection{Estimator}

Recall that the task in \autoref{problem:observ} is to estimate the observable expectation on the target eigenstate $\ket{u_j}  $ up to a certain precision $\varepsilon$, which is characterised by 
\begin{equation}
\label{eq:OND_ideal}
	\braket{O} = \frac{N(O)}{D} = \frac{\braket{u_i |O|u_i}}{\braket{u_i |u_i}}
\end{equation}
where  the denominator $D$ and the numerator $N$ is defined in respect to the eigenstate $\ket{u_i}$.
The  unnormalised eigenstate can be effectively realised by applying a spectral filter $g_{\tau\rightarrow \infty} (H - \omega)$ to an initial state, which holds when $\omega = E_j$ and the imaginary-time becomes infinity. 
In this section, the eigenenergy is first assumed to be known a priori. The task with an unknown eigenenergy will be discussed in \autoref{sec:unknown_energy}.

It is easy to see that we arrive at the ideal observable expectation  when $\tau \rightarrow \infty$,
\begin{equation}
	\braket{O} = \frac{N_{\tau \rightarrow \infty}(O)}{D_{\tau \rightarrow \infty}},
 \end{equation}
given a nonvanishing denominator, which is 
$ 
	D_{\tau \rightarrow \infty}(E_j) = |c_j|^2 = \eta. 
$
Note that we assume that the spectral weight $\braket{u_i|\psi_0}$ is nonvanishing.	
The observable when considering a finite $
\tau$ is estimated by
\begin{equation}
	\braket{O}_{\tau} = \frac{N_{\tau}(O)}{D_{\tau}}. 
\end{equation}
The denominator and the numerator have the same definition in \autoref{eq:OND_main}. 
In practice, when considering finite gate complexity and sample complexity, the eigenstate property is estimated by
\begin{equation}
\label{eq:observ_expec_finite}
	\hat{O}_{\tau, x_c, s_c} = \frac{\hat N_{\tau, x_c, s_c}(O)}{\hat D_{\tau, x_c, s_c}}. 
\end{equation}

The selection of $\tau, x_c, s_c$ can be determined by analysing the error of $\hat{O}_{\tau, x_c, s_c}$ compared to the ideal value.
As shown in \autoref{prop:composite_LCU}, the spectral filter can be written as an RLCU formula. The spectral filter takes an explicit form of
\begin{equation}
\begin{aligned}
	g_{\tau}(H-\omega) &= c \int_{-\infty}^{\infty} dx p(x)  e^{ix\tau \omega}   e^{-i\tau x H}\\
\end{aligned}
\end{equation}

The integrand is a real-time evolution with total real-time $\tau x$. Suppose we further use LCU formula to implement $ e^{-i\tau x H}$, which takes the form of
 \begin{equation}
 	e^{-i\tau x H} =  \mu( x \tau)\sum_{\vec{r} \in \mathcal{K}_x} \Pr(\vec{r}, x\tau, \nu(x \tau)) U_{\vec{r}} 
 \label{eq:time_evo_decomp}
 \end{equation} 
where we have follow the definition in \autoref{eq:g_tau_real_expansion_main}: $\vec{r}$ specifies the unitary $U_{\vec{r}}$ involved in the LCU formula of $ e^{-i\tau x H}$,   $\Pr(\vec{r}, x\tau, \nu(x \tau))$ represents the normalised decomposition coefficients of $U_{\vec{r}}$.
Then, we have 
\begin{equation}
\begin{aligned}
	g_{\tau}(H-\omega)	&= c(\mu) \int_{-\infty}^{\infty} dx p_{\mu}(x) e^{ix\tau \omega} \sum_{\vec{r} \in \mathcal{K}_x} \Pr(\vec{r}, x\tau, \nu(x \tau)) U_{\vec{r}}.
\end{aligned}
\label{eq:g_tau_real_expansion}
\end{equation}

\subsection{Error analysis}
\label{sec:observ_estimate_error_analysis}

Now, we discuss the errors when considering finite $\tau$, $x_c$, $s_c$.
Note that \autoref{eq:g_tau_real_expansion} is a composite LCU formula, more specifically, it is a $(c(\mu),0)$-LCU formula. When we consider a finite $x_c, s_c$, it will introduce some errors. The key idea of analysing the errors is to compare the operator distance between 
\begin{equation}
	\| g_{\tau, x_c, s_c}  - g_{\tau, x_c, s_c \rightarrow \infty}  \|
\end{equation}
 which will gives a bound for the numerator and the denominator.
 
 \paragraph{Error due to finite $\tau$.}
 
The imaginary-time $\tau$ determines the strength of the spectral filter. The spectral weight on the unwanted eigenstates, namely those energies away from the pre-set parameters $E$, will be exponentially suppressed. In the infinite time limit, $g$ will effectively project out all the spectral weights on the unwanted eigenstates, given that the initial state has a nonvanishing spectral weight on the target eigenstate.
The error of the denominator and numerator can be analysed by considering the operator distance between $g_{\tau}$ and $g_{\tau \rightarrow \infty}$. In this section, we choose $\omega = E_j$ and omit $\omega$ when there is no ambiguity.

\begin{proposition}[Error due to a finite $\tau$]
When $\tau \geq \frac{1}{\Delta} \sqrt{\ln(2/\varepsilon_{\tau})}$, the error of the denominator and the numerator that are defined in \autoref{eq:OND_main} when compared with those defined with respect to the ideal eigenstates in \autoref{eq:OND_ideal} satisfy
	$|D_{\tau} - D | \leq \varepsilon_{\tau}$, $|N_{\tau}(O) - N(O)| \leq \varepsilon_{\tau} \|O\|$.  
\end{proposition}

\begin{proof}
	
The distance between $g_{\tau}(H-E_j)$ and $g_{\tau \rightarrow \infty}(H-E_j)$ can be bounded by
\begin{equation}
	g_{\tau}(H-E_j) - g_{\tau \rightarrow \infty}(H-E_j) = \sum_{i \neq j} g_{\tau}(E_i-E_j) \ket{u_i}\bra{u_i}
	\label{eq:g_tau_eq}
	\end{equation}
where $g_{\tau}(H-\omega)$ is defined in \autoref{eq:g_tau_real_expansion}. 
According to the definition of  $g_\tau$, when $\tau \geq \frac{1}{\Delta} \sqrt{\ln(2/\varepsilon_{\tau})}$, it is easy to check that
\begin{equation}
	g_\tau (E_i - E_j) \leq \varepsilon_{\tau} /2,~\forall i\neq j
\end{equation}
and thus we have 
\begin{equation}
	\|g_{\tau}(H-E_j) - g_{\tau \rightarrow \infty}(H-E_j)\| \leq \varepsilon_{\tau} /2 
	\end{equation}

Using the result in \autoref{prop:RLCU_err_observ_bound}, the error of the numerator can  be bounded by
\begin{equation}
	N_{\tau}(O) - N(O) \leq \varepsilon\|O\|
\end{equation}
The error of the denominator can be obtained straightforwardly by taking $O = I$.

\end{proof}

 \paragraph{Error due to finite truncation in real-time evolution.}
 
Then we analyse the error due to finite $x_c$ and $s_c$.
In practice when the integral is truncated to from infinite range $[-\infty, \infty]$ a finite time length $[-x_c,x_c]$, $g_{\tau}(H-\omega)$ becomes
\begin{equation}
\label{eq:def_g_tau_x}
	g_{\tau, x_c}(H-\omega) =  c \int_{-x_c}^{x_c} dx p(x)  e^{ix\tau \omega}   e^{-i\tau x H}.
\end{equation}
Using the property of the Gaussian tail, it is easy to verify that the truncation error of $g_{\tau}$ with a finite $x_c$ is given by
\begin{equation}
	\|g_{\tau} -g_{\tau,x_c}\| \leq \varepsilon_{x_c}
	\label{eq:g_xm_error}
\end{equation}
when $x_c \geq 2 \sqrt{\ln (1/\varepsilon_{x_c})}$.
 
As shown in the main text, we consider dealing with the real-time evolution using Trotter-LCU expansion.
For evolution time $t$, let us denote the real-time evolution $U(t) := e^{-i H t} = U_m^{\nu}$ where $$U(m) = e^{-i m H } = \mu(m) \mu(m) \sum_r \Pr(r, m, \nu(m)) W_r  $$ is a short time dynamics during the time interval $m = t/\nu$.
In a $2k$th-order Trotter-LCU algorithm, the overall LCU formula for $U(m)$ is
\begin{equation} \label{eq:tildeU2kxPTS_SM}
\begin{aligned}
{U}(m)  = {V}_{2k}(m) S_{2k}(m),
\end{aligned}
\end{equation}
which consists of a deterministic $2k$th-order Trotter formula $S_{2k}$ and the Trotter error compensation term ${V}_{2k}$.
We consider a truncated Trotter-LCU formula to realize a $(\mu,\varepsilon)$-LCU formula for $U(t) = e^{-iHt}$  based on $\nu$ segments of $\tilde{U}_{2k} $ in \autoref{eq:tildeU2kxPTS_SM}. 

We will use the results established in \cite{zeng2022simple} to analyse the segment number $\nu$ to control the error to a certain level.
The complexity of the zeroth-order leading-order-pairing algorithm scales quadratically to $t$, which is undesirable for simulation with a long time.
If we apply $2k$th-order Trotter formula $S_{2k}(m) $ in each segment, the $2k$th-order remainder is $V_{2k}(m):=U(m)S_{2k}(m)^\dag$.
Using BCH formula, we can have an explicit form of $V_{2k}(m)$  
\begin{equation}
	V_{2k}(m) = \exp(i \sum_{s=2k+1}^{\infty} E_{2k,s} m^s) = \sum_{s=0}^\infty F_{2k,s}(m) =  \sum_{s=0}^\infty \|F_{2k,s}(m) \| V_{2k}^{(s)} (m)
\end{equation}
where we group the terms by the order of $m$, and  $F_{1,s}$ denotes the $s$-order expansion term of $V_{2k}(m)$ with $m^s$. Here, $\|F_{2k,s}(m)\|_1$ is the $1$-norm of the $s$-order expansion formula $F_{2k,s}(m)$ and $V_{2k,s}(m)$ is the normalised LCU formula for the $s$-order terms.
It is easy to see that
\begin{equation}
	 F_{2k,s}(m) = i m^s E_{2k,s},~s \in [2k+1, 4k+1]
\end{equation}
Formally, we can express 
\begin{equation}
	E_{2k,s} =  \lambda_r  \sum_{r} \Pr(r|s) P_{2k}(r).
\end{equation}
 
By pairing the $s$th order expansion term with the identity $I$,  the order of $F_{2k,s}(m)$ can be doubled as
\begin{equation}
	I + F_{2k,s}(m) = \lambda_r \sqrt{1 + m^{2s}} \sum_{r} \Pr(r|s) \exp( i \arctan (m^{2s}) P_{2k}(r).
	\end{equation} 
Therefore, the norm of $F_{2k,s}$  decreases from $\mc O(1+m^{s})$ to $\mc O(1+m^{2s})$. Here we note that this holds for the symmetry-conserved LCU decomposition where the elementary operations are \textsc{swap} and Pauli-$Z$ operators.	


The error can be shown to have a quick decrease with an increasing truncation order $s_c$. Considering of a finite truncation and the vanishing of $F_{2k,j\leq 2k}(m) = 0$, we can rewrite $V_{2k}^{(s)}(m)$ as
\begin{equation} \label{eq:DTS2kNew}
V_{2k}^{(s)}(m)= I + \sum_{s=2k+1}^{s_c} F_{2k,s}(m).
\end{equation}


Given the truncation, the  LCU formula for $U(m)$ is
\begin{equation} \label{eq:tildeU2kxPTS_SM}
\begin{aligned}
 {U}_{2k}^{(s_c)}(m)  = {V}_{2k,s}(m) S_{2k}(m).
\end{aligned}
\end{equation}
The overall LCU formula for $U(t)$ is to repeat the sampling of $\tilde{U}_{2k}(x)$ for $\nu$ times.
The overall LCU formula for time $t$ is
\begin{equation}
	{U}_{2k}^{(s_c)}(t) = \left( {U}_{2k}^{(s_c)}(m) \right)^\nu
	\label{eq:U_sc_approxi}
\end{equation}
with $\nu = t/m$. Here  $ {U}_{2k}^{(s_c)}(t)$ denotes the RLCU formula of $U$ when truncating it to $s_c$ orders and $t$ is the total evolution time, and ${U}_{2k}^{(s_c)}(t)$ is a $(\mu_{2k,tot}(t), \varepsilon_{2k,tot}(t) )$-LCU formula of $U(t)$ with
\begin{equation} \label{eq:muEps2kthDS}
\begin{aligned}
\mu_{2k,tot}(t) &= \mu(m)^\nu, ~ \varepsilon_{2k,tot}(t)  \leq \nu \mu_{2k,tot}(t) \varepsilon_{2k}(m)  
\end{aligned}
\end{equation}

The key component is to analyse the quantum resources needed to achieve an additive error of the approximation given by \autoref{eq:U_sc_approxi_main}
 \begin{equation}
	\|U(t) - {U}_{2k}^{(s_c)}(t)) \| \leq \varepsilon_{s_c}. 
\end{equation}
\autoref{lemma:TrotterLCU_error_2k} gives the Hamiltonian simulation error $ \varepsilon_{s_c}$  when using 2kth-order paired Taylor-series compensation with finite $s_c$. 

\begin{lemma}[Hamiltonian simulation error using 2kth-order paired Taylor-series compensation (Theorem 1 in \cite{zeng2022simple})]
\label{lemma:TrotterLCU_error_2k}

In a $2k$th-order Trotter-LCU algorithm, 
if the segment number $\nu$ and the truncation order $s_c$ satisfy
\begin{equation} \label{eq:nuScBoundKthLOP}
\begin{aligned}
\nu &\geq \left(\frac{2(e+c_k)\lambda t}{\ln\mu}\right)^{\frac{1}{4k+1}} 2\lambda t, \\
s_c &\geq \max\left\{ \left\lceil \frac{\ln\left( \frac{\mu}{\varepsilon} \nu_{2k}(t) \right)}{W_0\left( \frac{1}{2e\lambda t}\nu_{2k}(t) \ln\left( \frac{\mu}{\varepsilon} \nu_{2k}(t) \right) \right)} - 1 \right\rceil, 4k+1 \right\},
\end{aligned}
\end{equation}
we can then realize a $(\mu,\varepsilon)$-LCU formula for $U(t) = e^{-iHt}$, i.e.,
	$\|U(t) - {U}_{2k}^{(s_c)}(t)) \| \leq \varepsilon_{s_c} $,  based on $\nu$ segments of $\tilde{U}_{2k}$ in \autoref{eq:tildeU2kxPTS_SM}. Here,  $c_k$ is defined as $
c_k := \frac{1}{2}\left( \frac{e}{2k+1} \right)^{4k+2}.$

\end{lemma}

When considering a finite $s_c$, the spectral filter becomes
\begin{equation}
\label{eq:def_g_tau_x_s_SM}
	g_{\tau, x_c, s_c}(H-\omega) =  c \int_{-x_c}^{x_c} dx p(x)  e^{ix\tau E}   {U}_{2k}^{(s_c)}(x\tau, \nu(x\tau))
	\end{equation}
It is easy to check that $\| g_{\tau, x_c, s_c}(H-\omega)\| \leq c (1+\varepsilon_{s_c})$.
The operator distance between $g_{\tau,x_c} $ and $g_{\tau,x_c, s_c}$ due to  finite $s_c$ is
\begin{equation}
\label{eq:g_tau_sc_error}
	\|g_{\tau,x_c} - g_{\tau,x_c, s_c}  \| \leq c \int_{-x_c}^{x_c} dx p(x)  e^{ix\tau \omega}   \| U(x\tau) - {U}_{2k}^{(s_c)}(x\tau, \nu(x\tau)) \| \leq   \varepsilon_{s_c}
\end{equation}
when $\| U - {U}_{2k}^{(s_c)}(x\tau, \nu(x\tau)) \| \leq \varepsilon_{s_c} /c$.
 
Using the triangular inequality, the operator error between $g_{\tau}$ and $g_{\tau,x_c, s_c} $ defined in \autoref{eq:def_g_tau_x_s_SM} is 
\begin{equation}
\label{eq:g_tau_xc_sc}
	\|g_{\tau} - g_{\tau,x_c, s_c}  \| \leq \varepsilon_{x_c}   + \varepsilon_{s_c} 
\end{equation}

With the error due to a finite time when evaluating the integral \autoref{eq:g_xm_error}
 and the truncation order $s_c$ \autoref{eq:g_tau_sc_error}, we arrive at the following \autoref{prop:error_finite_time_x_s}.

\begin{proposition}[Error due to finite time length and truncation]
\label{prop:error_finite_time_x_s}
 When $x_c \geq 2\sqrt{\ln(2/\varepsilon_{x_c})}$, $\nu_c$ and $s_c$ satisfying \autoref{eq:nuScBoundKthLOP}, we have
	$|D_{\tau, x_c, s_c} - D_{\tau}| \leq 3 \varepsilon_{c}$, $|N_{\tau, x_c, s_c}(O) - N_{\tau}(O)| \leq 3 \varepsilon_{c} \|O\|$ with $\varepsilon_{c} = \varepsilon_{x_c} + \varepsilon_{s_c}$

\end{proposition}

\begin{proof}
	
The estimator 
\begin{equation}
\begin{aligned}
		| N_{\tau}(O) -  N_{\tau,x_c,s_c}(O)| = & |\braket{\psi_0|g_{\tau}(H-\omega) O g_{\tau}(H-\omega)|\psi_0} - \braket{\psi_0|g_{\tau,x_c,s_c}(H-\omega) O g_{\tau,x_c,s_c}(H-\omega)|\psi_0}\\
		 \leq & |\braket{\psi_0|g_{\tau}(H-\omega) O g_{\tau}(H-\omega)|\psi_0} 
		- \braket{\psi_0|g_{\tau,x_c,s_c}(H-\omega) O g_{\tau}(H-\omega)|\psi_0}| \\ 
		& +|\braket{\psi_0|g_{\tau,x_c,s_c}(H-\omega) O g_{\tau}(H-\omega)|\psi_0} -\braket{\psi_0|g_{\tau,x_c,s_c}(H-\omega) O g_{\tau,x_c,s_c}(H-\omega)|\psi_0} \\
		\leq & \varepsilon_{c} (\|g_{\tau} \| + \|g_{\tau,x_c,s_c}\| ) \|O\| \leq  3 \varepsilon_{c} \|O\|
\end{aligned} 
\end{equation}
where $\varepsilon_{c} := \varepsilon_{x_c}  + \varepsilon_{s_c}$.
In the third inequality, we have used the inequality by \autoref{eq:g_tau_xc_sc}.

\end{proof}

The gate complexity of the overall algorithm can then be estimated. To construct controlled-$U(t)$, we split it to $\nu$ segments. In each segment, we need to implement $k$th-order Trotter circuits and random Taylor-series sampling circuits. The number of gates in the random Taylor-series sampling circuit is $\mc{O}(s_c)$. Therefore, the gate complexity of the overall algorithm using $K$th Trotter formula ($K=0,1,2k$) is given by 
\begin{equation}
\begin{aligned}
N_{K} &= \mc{O}(\nu(\kappa_K L + s_c)),
\end{aligned}
\end{equation}
where
\begin{equation}
\kappa_K = 
\begin{cases}
K, & K = 0,~1 \\
2\cdot 5^{K/2-1}, & K=2k,~k\in \mathbb{N}_+.
\end{cases}
\end{equation}

\paragraph{Unbiased estimator.} 

 \revise{ 
Recall that to estimate $\braket{\vec {\mathbf{j}} | O | \vec {\mathbf{i}}}$, the Hadamard-test type of circuit in Fig. 1(c1) in the main text is used.
In Fig. 1(c1),} we first prepare the state $\ket{\psi_0}$ and an extra ancillary qubit prepared on $|+\rangle$. Afterwards, we perform a C-U gate from ancillary to $\ket{\psi_0}$. If we directly measure the ancillary qubit on the X-basis, the outcome $a$ will be $0$ with a probability of $\Pr (a=0) = \frac{1}{2}(1 + \RE(\braket{\vec {\mathbf{j}} | O | \vec {\mathbf{i}}})$ and 1 with a probability of $\Pr (a=1) = \frac{1}{2}(1 - \RE(\braket{\vec {\mathbf{j}} | O | \vec {\mathbf{i}}})$.
Then, similarly, we repeat the circuit but with an inverted phase gate $S^{\dagger} = \Rz(-\frac{\pi}{2})$ before measurement.
The outcome $b$ will be $0$ with a probability of $\Pr (b=0) = \frac{1}{2}(1 + \IM(\braket{\vec {\mathbf{j}} | O | \vec {\mathbf{i}}})$ and 1 with a probability of $\Pr (b=1) = \frac{1}{2}(1 - \IM(\braket{\vec {\mathbf{j}} | O | \vec {\mathbf{i}}})$.
\begin{proposition}
\label{prop:estimator}
The estimator $\hat v$ defined in \autoref{eq:estimator_v_main} is an unbiased estimator of $\braket{\psi_0|g_{\tau}(H-\omega)Og_{\tau}(H-\omega)\psi_0}$.

\end{proposition}

\begin{proof}

Taking an average over $a$, $b$, $r$ we have
\begin{equation}
\begin{aligned}
		\mathbb E_{a,b,r} \hat{v} &= c^2(\mu)  e^{i\tau \omega(t_i - t_j)}\mathbb E_{a,b,r}  \hat{d} = c^2(\mu)  e^{i\tau \omega(t_i - t_j)} = c^2(\mu)  e^{i\tau \omega(t_i - t_j)} \braket{\vec {\mathbf{j}} | O | \vec {\mathbf{i}}}	
\end{aligned}
\end{equation}

Taking an average of $t_i, t_j, \vec i, \vec j$, we have
\begin{equation}
\begin{aligned}
		\mathbb E_{t_i, t_j, \vec i, \vec j} \mathbb E_{a,b,r}\hat{v} &= c^2(\mu) \mathbb E_{t_i, t_j} e^{i\tau \omega(t_i - t_j)} \mathbb E_{\vec i, \vec j}   \braket{\vec {\mathbf{j}} | O | \vec {\mathbf{i}}}	\\
		&= c^2(\mu)  \sum_{i, j} \Pr(i) \Pr(j) e^{i\tau \omega(t_i - t_j)} \sum_{\vec i, \vec j} \Pr(\vec i,\tau t_i, \nu) \Pr( \vec j, \tau t_j, \nu) 	 \braket{\vec {\mathbf{j}} | O | \vec {\mathbf{i}}}	\\
		&= \braket{\psi_0|g_{\tau}(H-\omega)Og_{\tau}(H-\omega)\psi_0} = N_{\tau} (O).
\end{aligned}
\end{equation}
	
\end{proof}

In the other case, where $O = \sum_{l} o_l P_l $ is composed of many terms as described in \autoref{problem:observ},   an importance sampling method can be used to estimate the observable.  
Compared to measuring each term separately, the measurement cost using importance sampling is independent of $L$. In particular, we sample the observables' index $l$ from the probability distribution $o_l/ \|O\|_1$. 
Given a sampled $l$, we define the estimator, $$\hat{d}_l  =  (-1)^{a_l} + i (-1)^{b_l},$$ in a similar way to that in Methods.   It is easy to see that $\hat{d}_l$ is an unbiased estimator of $\braket{\vec {\mathbf{j}} | O | \vec {\mathbf{i}}}	 $ as
\begin{equation}
	\mathbb{E}_{a,b,l} \hat{d}_l = \braket{\vec {\mathbf{j}} | O_l | \vec {\mathbf{i}}},
\end{equation}  
as one can verify that 
\begin{equation}
	\mathbb E_{a_l} (-1)^{a_l} = \RE(\braket{\vec {\mathbf{j}} | O_l | \vec {\mathbf{i}}} ),~\mathbb E_{b_l} (-1)^{b_l} = \IM( \braket{\vec {\mathbf{j}} | O_l | \vec {\mathbf{i}}} ).
\end{equation}

Similarly, we define the estimator as
\begin{equation}
\hat v_l = c^2(\mu)  e^{i\tau \omega(t_i - t_j)} \hat{d}_l
\end{equation}
Taking the average of  $t_i,t_j,\vec i,\vec j,a,b,r,l$, we have
\begin{equation}
	\mathbb E_{t_i,t_j,\vec i,\vec j,a,b,r,l} \hat v_l = N_{\tau} (O).
\end{equation}

The range of the estimator is $[-\sqrt{2} c^2(\mu)\| O\|_1, \sqrt{2} c^2(\mu)\|O\|_1]$.   Using Hoefflding inequality, the estimation error can be bounded 
\begin{equation}
	\begin{aligned}
		|\hat{D}_{\tau, x_c, s_c} - {D}_{\tau, x_c, s_c} | &\leq \varepsilon_n \\
		|\hat{N}_{\tau, x_c, s_c}(P_l) - {D}_{\tau, x_c, s_c}(P_l) | &\leq \varepsilon_n \| O\|_1 \\
	\end{aligned}
\end{equation}
with a success probability of  $1-\vartheta$
when the number of measurements $N_s \geq 2 c^4(\mu) \frac{1}{\varepsilon_n^2}\ln(1/\vartheta) $. In this work, we will set $c(\mu) = 2$.
We have the following result.

\begin{proposition}[Error due to finite measurements using Hadamard test]\label{prop:measurment_error}
$|\hat D_{\tau, x_c, s_c} - D_{\tau, x_c, s_c}| \leq \varepsilon_n
$
 and  $|\hat N_{\tau, x_c, s_c}(O) - N_{\tau, x_c, s_c} (O)| \leq \varepsilon_n
$ has a success probability of  $1-\vartheta$
when the number of measurements $N_s \geq 2 c^4(\mu)\frac{\|O\|_1^2}{\varepsilon_n^2}\ln(1/\vartheta) $.
\end{proposition}

\paragraph{Measurement strategy.}


To measure $\braket{\vec {\mathbf{j}} | O | \vec {\mathbf{i}}}$, we can use the circuit in \autoref{fig:cartoon}(c) to generate 
\begin{equation}
    U_p^{\dagger} (    \prod_{q =1}^{\nu} W_{i_q} S)^{\dagger} O  \prod_{q =1}^{m} W_{i_q} S U_p \ket{\psi_0},
\end{equation}
and then measure on a computational basis. 

In many practical cases, such as Heisenberg models and electronic structure problems (\autoref{eq:molecular_hamil_main}), the target problem has certain symmetries $\mathcal{S}$ satisfying $[H, \mathcal{S}] = 0$. Consequently, the state can be divided into state spaces with different symmetry sectors. For instance, for fermionic problems, the Hamiltonian has particle number symmetry,  $[H, \hat N] = 0$ with $\hat   N = \sum_i \hat{a}_i^{\dagger} \hat{a}_i$. In this case, the state can be divided into different sectors in the Hilbert space with different particle numbers,  $\mathcal{H}: \mathrm{span} \{ |i\rangle\}$ with  $\ket{i}$ representing the state with $i$ particles. We can use the ancilla-free measurement strategy in \autoref{fig:cartoon}(c2). Note that in this case, the operator $W_{i_q}$ must be symmetry preserving.

Recall that the requirements of this protocol include the preparation of the superposition state and efficient computation of $\braket{\psi_{\mathrm{Ref}}| U | \psi_{\mathrm{Ref}}}$.
Let us take the fermionic Hamiltonian as an example (note that the Heisenberg model can be regarded as a 1D Fermi-Hubbard model after the Jordan-Wigner transformation).
If the reference state is the vacuum state with the number of particles being zero, then  $\braket{\psi_{\mathrm{Ref}}| U | \psi_{\mathrm{Ref}}}$ can be computed classically.  
The Hartree-Fock state takes the simple product-state form $\ket{\psi_0} = \ket{1}^{\otimes N_e } \ket{0}^{\otimes N-N_e} \in \{\ket{N_e}\}$,
where the first $N_e$ qubits are prepared in the $\ket{1}$ states and the rest of the qubits remain in the  $\ket{0}$ states.  To prepare the target superposition state $\ket{\phi_0}$, a Hadamard gate is first applied to the initial qubit, followed by a chain of \textsc{cnot} gates up to the $N_e$th qubit, requiring a total of $N_e - 1$ \textsc{cnot} operations. However, we can alternatively choose a reference state that differs from the initial state by only a single site. In this case, the preparation of the superposition state becomes  simpler, requiring only one additional \textsc{cnot} gate compared to the original state. 

\revise{
We have adopted this ancilla-free measurement strategy in the experiments, which avoids decomposing long-range controlled gates into local operations and hence has a low compilation overhead. That said, in our resource estimates we used a relatively conservative evaluation, ensuring that the reported results remain applicable even if the first method are employed instead.}

\subsection{Proof of \autoref{thm:observ_estimation_main} (\autoref{problem:observ}) }

When $\tau$ is taken as a finite value, $g_{\tau}$ can be regarded a $(1, \varepsilon_{\tau})$-RLCU formula of $g$.	
The integral is evaluated within a finite time range,
\begin{equation}
		g_{\tau,x_c} =  \int_{-x_c}^{x_c} dx p(x)U(x\tau)
	\end{equation}

One can check that $g_{\tau,x_c}$ is a $(1, \varepsilon_{x_c})$-RLCU formula of $g_{\tau}$ defined in \autoref{Eq:superposed}, and is a  $(1, \varepsilon_{x_c} + \varepsilon_{\tau})$-RLCU formula of $g_{\tau \rightarrow \infty}$.
Each of the integrand $U(x\tau )$ in $g_{\tau,x_c}$ is approximated by $ U_{2k}^{(s_c)}(x \tau )$ defined in \autoref{eq:tildeU2kxPTS_SM}.
According to \autoref{prop:error_finite_time_x_s} and \autoref{lemma:TrotterLCU_error_2k}, when considering a segment number $\nu_c$ and the truncation order $s_c$ defined in \autoref{eq:sc_requirement},  $ U_{2k}^{(s_c)}(x \tau )$ is a   $(\mu_2(x),\varepsilon_x)$-LCU formula, where   upper bounds of $\mu_2$ and $\varepsilon_x$ can be obtained by \autoref{prop:composite_LCU}.
	
According to \autoref{prop:composite_LCU}, $g_{\tau,x_c,s_c}$, which is defined as
\begin{equation}
		g_{\tau,x_c,s_c}=  \int_{-x_c}^{x_c} dx p(x)\tilde U(x \tau) 
		\end{equation} 
		is a $(\mu,\varepsilon)$-LCU formula of $g$, with $\mu  := \int_{-\infty}^{\infty} p(x) \mu_2(x) dx \leq \max_x \mu_2(x)
$, and $\varepsilon = \varepsilon_{c} + \varepsilon_{\tau}  +   \varepsilon_{s_c}$.

Next, we analyse the error due to finite measurements.
In the case where we can directly measure in the eigenbases of the observable $O$, we can take the estimator in \autoref{eq:estimator_v_main}.  Then the error of the expectation value of observable can be bounded by directly using \autoref{prop:RLCU_err_observ_bound}.
The error distance can be bounded in a similar fashion to \autoref{prop:RLCU_err_observ_bound}. 
For the case with observables composed of many Pauli terms. Suppose we measure it using importance sampling; the only difference is the amplification of the prefactor by $\|O\|_1$.

Equipped with all these results, we arrive at the following result of eigenstate property estimation. 

\begin{theorem}[Eigenstate property estimation for generic Hamiltonians]
\label{thm:observ_estimation_detailed}

Suppose the conditions and assumptions in \autoref{problem:observ} hold.  
Suppose we use the method in \autoref{alg1_main} where we choose the time-segment number function $\nu_c = \mathcal{O} \left( (\lambda \Delta^{-1} \ln(\eta^{-1} \varepsilon^{-1}) )^{1+ \frac{1}{4k+1}} \right )$ for realising the real-time evolution (in \autoref{eq:nu_c_form}) and the truncation order $s_c = \mathcal{O} (\ln(\nu_c/\varepsilon_c)/\ln \ln(\nu_c/\varepsilon_c)))$.  
We can achieve the error of observable's expectation on the eigenstate $\ket{u_j}$ within $\varepsilon$, 
$|\hat v - \braket{u_j | O |u_j} | \leq  \varepsilon $ when the sampling number is $N_s = \mc O\left ( {\eta ^{-2}\varepsilon^{-2}} \|O\|_1^2 \ln(1/\vartheta) \right)$ (in \autoref{eq:Ns_form_generic}), where $\hat v$ is the estimator defined in \autoref{eq:estimator_v_main} with a success probability at least $1-\vartheta$.

 \end{theorem}

\begin{proof}

We set the imaginary-time as 
\begin{equation}
	\tau \geq \frac{1}{\Delta} \sqrt{\ln \frac{2}{\varepsilon_{\tau}}}
	\label{eq:tau_explicit}
\end{equation}
and set truncation time $x_c$ 
\begin{equation}
	x_c \geq 2\sqrt{\ln(2/\varepsilon_c)}. 
\end{equation}
The maximum real-time  is 
\begin{equation}
	t_c = \tau x_c = \frac{2}{\Delta} \sqrt{\ln(2/\varepsilon_\tau)} \sqrt{\ln(2/\varepsilon_c)}
	\label{eq:maximum_real_time} 
\end{equation}

According to \autoref{prop:error_finite_time_x_s}, when the segment number is set as
\begin{equation}
\begin{aligned}
	\label{eq:nu_c_form}
		\nu_c &= 2 \left(\frac{2(e+c_k)}{\ln 2}\right)^{\frac{1}{4k+1}} (\lambda \tau x_c)^{1+\frac{1}{4k+1}} = 4 \left(\frac{4(e+c_k)}{\ln 2}\right)^{\frac{1}{4k+1}} (\frac{\lambda}{\Delta} \sqrt{\ln(2/\varepsilon_\tau)\ln(2/\varepsilon_c)})^{1+\frac{1}{4k+1}}
	\end{aligned}
\end{equation}
and set \begin{equation}
	s_c  =   \mc O \left(   \frac{\ln\left( \frac{4\nu_c}{\varepsilon_{s_c}}  \right)}{\ln\left(  \nu_c^{\frac{1}{4k+2}} \ln\left( \frac{4\nu_c}{\varepsilon_{s_c}} \right) \right)}  \right ) 
	\label{eq:sc_requirement}
\end{equation}
 ${U}_{2k}^{(s_c)}(t_i)$ is a $(2, \varepsilon_{s_c})$-LCU formula of ${U}(t_i)$, and the approximation error is given by
\begin{equation}
	\max_{t_i} \|{U}_{2k}^{(s_c)}(t_i) - U(t_i)\| \leq \varepsilon_{s_c}.
\end{equation}

According to \autoref{prop:error_finite_time_x_s}, we have 
$|D_{\tau, x_c, s_c} - D_{\tau}| \leq   \varepsilon_{c}$, $|N_{\tau, x_c, s_c}(O) - N_{\tau}(O)| \leq   \varepsilon_{c} \|O\|$ with $\varepsilon_{c} = 3(\varepsilon_{x_c} + \varepsilon_{s_c})$

We set the sampling number as
\begin{equation}
\label{eq:Ns_form_generic}
	N_s = \frac{2c^4(\mu) \| O \|_1^2 }{\varepsilon_n^2} \ln(1/\vartheta).
\end{equation}
According to \autoref{prop:measurment_error}, $|\hat D_{\tau, x_c, s_c} - D_{\tau, x_c, s_c}| \leq \varepsilon_n
$
 and  $|\hat N_{\tau, x_c, s_c}(O) - N_{\tau, x_c, s_c} (O)| \leq \varepsilon_n
$ has a success probability of  $1-\vartheta$
when the number of measurements $N_s \geq 32\frac{\|O\|_1^2}{\varepsilon_n^2}\ln(1/\vartheta) $.
Using the triangular inequality, we have
\begin{equation}
\begin{aligned}
	\varepsilon_N &:= |\hat N_{\tau, x_c, s_c}(O) - N_{\tau\rightarrow \infty} (O)| \leq \|O\|( \varepsilon_\tau+\varepsilon_{c} ) +\|O\|_1 \varepsilon_n, \\
	\varepsilon_D &:= |\hat D_{\tau, x_c, s_c}  - D_{\tau\rightarrow \infty}  | \leq \varepsilon_\tau+\varepsilon_{c}+\varepsilon_n, 
\end{aligned}
\end{equation}

The observable's error in respect to $\varepsilon_{\tau}$, $\varepsilon_{c}$ and $\varepsilon_n$ is given by
 
\begin{equation}
\begin{aligned}
\label{eq:derive_O_err}
\left|\langle\hat{O} \rangle_{\tau, x_c, s_c}-\left\langle O\right\rangle\right| & =\left|\frac{\hat{N}_{\tau, x_c, s_c}\left(O \right)}{\hat{D}_{\tau,x_c, s_c}}-\frac{N\left(O\right)}{D}\right| \\
& =\left|\frac{D \hat{N}_{\tau, x_c, s_c}(O)-N\left(O\right) \hat{D}_{\tau, x_c, s_c}}{D \hat{D}_{\tau, x_c, s_c}}\right| \\
& \leq\left|\frac{D\left(N\left(O\right)+\varepsilon_N\left(O\right)\right)-N\left(O\right)\left(D-\varepsilon_D\right)}{D\left(D-\varepsilon_D\right)}\right| \\
& =\left|\frac{D \varepsilon_N\left(O\right)+N\left(O\right) \varepsilon_D}{D^2-D \varepsilon_D}\right| \\
& \leq\left|\frac{D \varepsilon_N\left(O\right)+\left(N\left(O\right)+D\right) \varepsilon_D}{D^2}\right| \\
& =\eta^{-1}\left(\left(\varepsilon_\tau+\varepsilon_c\right)\left\|O\right\|+\varepsilon_n\left\|O\right\|_{1}\right)+\eta^{-1}\left(\left\langle O\right\rangle+1\right)\left(\varepsilon_\tau+\varepsilon_c+\varepsilon_n\right) \\
& \leq \eta^{-1}\left(2\left\|O\right\|+1\right)\left(\varepsilon_\tau+\varepsilon_c\right)+\eta^{-1}\left(\left\|O\right\|+\left\|O\right\|_1+1\right) \varepsilon_n .
\end{aligned}
\end{equation}

If we set the precision for each component
\begin{equation}
	\begin{aligned}
	\label{eq:epsilon_requirement}
& \varepsilon_\tau=\varepsilon_c=\frac{1}{3} \eta \varepsilon  \left(\frac{1}{2\left\|O\right\|_{\infty}+1}\right), \\
& \varepsilon_n=\frac{1}{3} \eta \varepsilon   \frac{1}{\left\|O\right\|_{\infty}+\left\|O\right\|_1+1},
\end{aligned}
\end{equation}
then we have 
\begin{equation}
\left|\langle\hat{O} \rangle_{\tau, x_c, s_c}-\left\langle O\right\rangle\right| \leq \varepsilon	
\end{equation}
for observable $O$
 with a success probability bounded by $1- \vartheta$.
 
Under the precision requirement in \autoref{eq:epsilon_requirement}, we know from \autoref{eq:nu_c_form} that the segment number scales as
\begin{equation}
	\nu_c =   \mathcal{O} \left( (\lambda \Delta^{-1} \ln(\eta^{-1} \varepsilon^{-1}) )^{1+ \frac{1}{4k+1}} \right )
\end{equation}
and the sampling number scales as
\begin{equation}
	N_s = \mc O\left ( {\eta ^{-2}\varepsilon^{-2}} \|O\|_1^2 \ln(1/\vartheta) \right).
\end{equation}



\end{proof}

Importantly, the actual quantum resources can be directly estimated by \autoref{eq:nu_c_form} and \autoref{eq:epsilon_requirement} (and thus \autoref{eq:nu_c_actual}) and \autoref{eq:Ns_form_generic}. 
 

\subsection{Gate complexity analysis}


Using the results in \autoref{prop:gate_cost_observ_dynamics} and \autoref{thm:observ_estimation_detailed}, we have the following result for generic Hamiltonians.

\begin{theorem}[Gate complexity for generic Hamiltonians' eigenstate property estimation]
Suppose we are given the same condition in \autoref{thm:observ_estimation_detailed}.
Suppose we synthesise the circuit to \textsc{cnot} gates, single-qubit Clifford gates and single-qubit Z rotation gates, we have 
  gate count

\begin{equation}
\label{eq:gate_each_sample}
\GateC_{\CNOT} =	\mathcal{O} ( \wt(H) \nu_c ),~\GateC_{\Rz} =	\mathcal{O} ( L \nu_c )
\end{equation}
with $\nu_c$ given by \autoref{eq:nu_c_actual}.

\end{theorem}

\begin{proof}

We can calculate the maximum required evolution time based on Theorem \autoref{thm:observ_estimation_detailed} with Gaussian function if we want to estimate an observables $O$ with an accuracy $\varepsilon$.

For the ground state preparation, the maximum real-time is given by \autoref{eq:maximum_real_time}. 
Suppose we divide the time slice into $\nu$ segments and use $2k$th order Trotterisation, in this context, the segment number $\nu$ can be chosen as 
$ 
	\nu_c = \mathcal{O} ( (\lambda \Delta^{-1}  \ln(\eta^{-1} \varepsilon^{-1}))^{1+ \frac{1}{4k+1}}).
$
We note that since there is a saturation of the gate count for the Trotter-error-compensation indicated by \autoref{eq:compensation_CNOT_gate}, we take $s_c$ to be infinity in deriving the asymptotic scaling, and the gate count is less than or at the same order of $L$ or $\wt(H)$.
Therefore, the gate complexity in each sample is given by \autoref{eq:gate_each_sample}.   
In cases where we could implement each $e^{-iH_l t}$, then the gate complexity in each single run is 
\begin{equation}
	\mathcal{O} ( L (\lambda \Delta^{-1}  \ln(\eta^{-1} \varepsilon^{-1}))^{1+ \frac{1}{4k+1}}).
\end{equation}

Thus far, we have completed the proof of \autoref{thm:observ_estimation_main}.

\end{proof}

 Suppose $O$ is simply a Pauli operator, $\|O\| = 1$.
A more explicit form for the segment number can be obtained using \autoref{eq:epsilon_requirement}. With some derivation, we have
\begin{equation}
\label{eq:nu_c_actual}
	\nu_c = 4 \left(\frac{4(e+c_k)}{\ln 2}\right)^{\frac{1}{4k+1}} (\frac{\lambda}{\Delta}  \ln(\frac{9}{\eta \varepsilon} ) )^{1+\frac{1}{4k+1}},
\end{equation}
which can be used to carry out resource estimations when given the parameters of the Hamiltonian.

\autoref{thm:observ_estimation_detailed} gives an upper bound on the gate complexity required for estimating generic Hamiltonians' eigenstate properties in relation to $\lambda$, $L$, and $\wt(H)$.
For physical Hamiltonians,  we could reduce the gate complexity by using the properties of the Hamiltonians.
As shown in Theorem 2  in \cite{zeng2022simple}, the segment number can be reduced to
\begin{equation}
\label{eq:nu_lattice}
    \nu_c = \mathcal{O} ( n^{\frac{2}{4k+1}} ( \Delta^{-1} \ln(\eta^{-1} \varepsilon^{-1}))^{1+ \frac{1}{4k+1}}).
\end{equation}
Using \autoref{prop:gate_cost_observ_dynamics}, we can obtain the gate complexity for lattice models. 
In short, by using the commutation relation of the Hamiltonian terms in real-time evolution, the gate complexity for eigenstate property estimation using \autoref{alg1_main} scales $\mc O(n ^{1 + \frac{2}{4k+1}})$. When using the ancilla-free LCU formulae and the corresponding measurement strategy, the gate operations within each segment can be implemented in parallel, as discussed in \autoref{sec:gate_complexity_analysis}. Therefore, the depth complexity is   $\mc O(n ^{ \frac{2}{4k+1}})$.
Therefore, we have completed the proof of \autoref{thm:observ_estimate_gate_Lattice_main} in the main text


\subsection{Discussions on molecular Hamiltonians}

The following discusses the resource cost for fermionic problems with the Hamiltonian \autoref{eq:molecular_hamil_main} in the main text and its qubit form given by \autoref{eq:molecular_hamil_qubit}.
Note that each term in the expansion of $V_i$ is a Pauli operator. Therefore, at each segment, $\delta t$, the gates that effectively implement the remainder will merge into a single Pauli operator. This means that the gate count for Trotter-error compensation will saturate to $n$, regardless of the truncation order $s_c$.

The Trotter formula will be implemented in a split-operator way. The nonlocal kinetic term will be converted to local terms with an additional cost for implementing the diagonalisation. However, to derive the remainder, we still need to expand all the terms in the Pauli basis. 
The second-quantised Hamiltonian given by \autoref{eq:molecular_hamil_main} can be mapped by the Jordan-Wigner transform to a qubit Hamiltonian of the form
\begin{equation}
\label{eq:molecular_hamil_qubit}
H= \hat{T} +  \hat{V} = \sum_{p q} \tilde{T}_{p q}\left(X_p Z_{p+1} Z_{p+2} \cdots Z_{q-1} X_q+Y_p Z_{p+1} Z_{p+2} \cdots  Z_{q-1} Y_q\right)+\sum_p \tilde{U}_p Z_p+\sum_{p q} \tilde{V}_{p q} Z_p Z_q 
\end{equation}

A direct application of our zeroth-order method gives
\textsc{cnot} gate count: $  \wtm(H) (s_c+2) \nu = \mathcal{O} (n (s_c + 2) \nu)  = \mathcal{O} (n (s_c + 2) (\lambda t)^2) $, and
\sun{
2kth order \textsc{cnot} gate count gives $  (\wtm(H) s_c +2 \wt(H) ) \nu = \mathcal{O} ( n^3  (\lambda t)^{1+\frac{1}{4k+1}}) $ with $\wt(H) = nL = n^3$.}
However, the kinetic operator is quadratic and thus can be diagonalised   by an efficient circuit transformation $C$

\begin{equation}
\label{eq:diag_kinetic}
\hat{T} = C\left(\sum_p T_p Z_p\right) C^{\dagger}
\end{equation}

The second-order Trotterised time evolution is
\begin{equation}
	S_2 = (e^{-i T x/2} e^{-i V x} e^{-i T x/2})^{\nu} =  (C e^{-i T x/2} C^{\dagger} e^{-i V x} C e^{-i T x/2} C^{\dagger})^{\nu}
\end{equation}
For the error compensation term, the gate cost shown in \autoref{eq:compensation_CNOT_gate} has a saturation, which is upper bounded by $3n$.
As analysed in Methods, the dominant cost is from the Trotterisation.
For Hamiltonian in \autoref{eq:molecular_hamil_main}, one can use the fermionic swap network to simulate the Hamiltonian dynamics with depth $\mathcal{O} (n)$ and gate count  $\tilde {\mathcal{O}} (n)$  in each segment. Using the results in \autoref{table:dep_cost_tot}, we can estimate the gate complexity for molecular systems in  \autoref{eq:molecular_hamil_main}.


\begin{corollary}[Eigenstate property estimation for molecular systems]
\label{thm:observ_estimate_gate_chem}

Suppose we aim to estimate the observable $O$ on the eigenstate $\ket{u_i}$ of an $n$-qubit second-quantised plane wave Hamiltonian specified in \autoref{eq:molecular_hamil_main},  $\braket{u_i | \hat O| u_i}$.
To achieve an estimation error within $\varepsilon$, it is sufficient to have the gate count of $\tilde{\mathcal{O}} (n^2   ( \Delta^{-1})^{1+\frac{1}{4k+1}} \log(\eta^{-1} \varepsilon^{-1}) \log(  \vartheta^{-1})  )$   and a success probability $1 -\vartheta$.
\end{corollary}





\sun{ 
For quantum molecular systems, the Hamiltonian ${H}$ takes the following form
\begin{equation}
\label{eq:hamiltonian_molecule_general}
    {H}:=\hat{T}+\frac{1}{2}\hat{V}+C=\sum_{i,j=1}^n h_{ij} \hat{a}_i^\dagger \hat{a}_j + \frac{1}{2}\sum_{i,j,k,l,=1}^n g_{ijkl}\hat{a}_i^\dagger \hat{a}_j^\dagger \hat{a}_k \hat{a}_l,
\end{equation}
where $n$ is the number of spin orbitals of the molecular system; $\hat{a}_i^\dagger$ and $\hat{a}_i$ are the fermionic generation and annihilation operators, respectively; $h_{ij}$ and $g_{ijkl}$ are the corresponding coefficients for the one-body and two-body interactions, respectively.
Note that the identity term in the Hamiltonian is a trivial term, so it is removed in the Hamiltonian in this work for simplicity.
For quantum chemistry Hamiltonians, there has been considerable progress in efficiently representing the Hamiltonians with fewer terms and low weights.}

These results, such as single factorisation~\cite{berry2019qubitization} and double factorisation~\cite{von2021quantum}, can be directly applied to reduce the cost.
A common strategy is to reformulate the two-body fermion operators as a sum of squared one-body operators by Cholesky decomposition, as has been used in AFQMC. 
The Hamiltonian is reformulated as 
\begin{equation}
    {H} =  \hat{K} + \hat{V} :=\hat{K}+\frac{1}{2} \sum_{\ell}^{\Gamma} \hat{L}_{\ell}^{2},
    \label{eq:MCH}
\end{equation}
with
$ 
    \hat{K}:=\sum_{i,j=1}^{n} \left[h_{i j}-\frac{1}{2} \sum_{\ell=1}^{\Gamma} \sum_{k=1}^{n} L_{(i k) \ell} L_{(j k) \ell}^{*}\right] \hat{a}_{i}^{\dagger} \hat{a}_{j},
$
and
$ \hat{L}_{\ell}:=\sum_{i,l=1}^{n}L_{(i l) \ell} \hat{a}_{i}^{\dagger} \hat{a}_{l}$.
Here, the constant is removed and $\Gamma = \mathcal{O}(n)$ is the number of terms of $\hat{L}_\ell$.  

Below, we use the first-order Trotter formula as an example to analyse the cost within each Trotter segment
\begin{equation}
	e^{-i H t} \approx  e^{-i \hat{K} t} \prod_{\ell = 1}^\Gamma e^{-\frac{it}{2} U_{\ell}   \sum_{p} (f_p \hat{n}_p)^2  U_{\ell}^{\dagger}} = e^{-i \hat{K} t}
	  \prod_{\ell = 1}^\Gamma U_{\ell} e^{-\frac{it}{2} \sum_{p} (f_p \hat{n}_p)^2 } U_{\ell}^{\dagger}
\end{equation}
In the first line, Trotterisation is used, and thus, this is an approximation with some Trotter errors up to the second order.
In the second line, some derivations have been abbreviated and the key facts that we used are $[U_{\ell} \hat{n}_p U_{\ell}^{\dagger}, U_{\ell} \hat{n}_q U_{\ell}^{\dagger} ] =0 $ and $ e^{ - i t U_{\ell} \hat{n}_p  U_{\ell}^{\dagger} } = U_{\ell} e^{ - i t  \hat{n}_p } U_{\ell}^{\dagger}$.

Using the idea in \cite{low2018hamiltonian}, the implementation of $e^{-\frac{it}{2} \sum_{p} (f_p \hat{n}_p)^2 } $  may  be done with $\mathcal{O} (n)$ depth circuit.  
Therefore in total, for each time segment, we require $\Gamma \times \mathcal{O} (n) =  \mathcal{O} (n^2)$ depth circuit.

A very inefficient way to implement the above process may be like this.
We first compute $Z_j$ by applying a controlled NOT gate and store the information on the ancillary qubit $\ket{j_1} ... \ket{j_n} $. Then implement as follows:
\begin{equation}
    O_a \bigotimes_i \ket{n_i} \ket{0}_o \rightarrow \bigotimes_i \ket{n_i} \ket{ \sum_i f_i n_i} \rightarrow \bigotimes_i \ket{n_i} \ket{ (\sum_i f_i n_i)^2} \rightarrow \bigotimes_i \ket{n_i} e^{-i(\sum_i f_i n_i)^2 }\ket{ (\sum_i f_i n_i)^2} 
\end{equation}

We require the following circuit oracle
\begin{equation}
    O_a \ket{j_1} ... \ket{j_n} \ket{0}_{r} =  \ket{j_1} ... \ket{j_n} \ket{(-1)^{j_1} + (-1)^{j_2} +... + (-1)^{j_n}}_r. 
\end{equation}
\begin{equation}
O_{\vec{A}}|j\rangle|0\rangle_o|0\rangle_{\text {garb }}=|j\rangle\left|A_j\right\rangle_o|g(j)\rangle_{\text {garbage }}
\end{equation}
with $A_j = j^2$.
Then
\begin{equation}
|j\rangle|0\rangle_o|0\rangle_{\text {garb }}|0\rangle \underset{O_{\vec{A}}}{\rightarrow}|j\rangle\left|A_j\right\rangle_o|g(j)\rangle_{\text {garb }}|0\rangle \underset{\text { PHASE }}{\rightarrow} e^{-i A_j t}|j\rangle\left|A_j\right\rangle_o|g(j)\rangle_{\text {garb }}|0\rangle \underset{O_{\vec{A}}^{\dagger}}{\rightarrow} e^{-i A_j t}|j\rangle|0\rangle_o|0\rangle_{\text {garb }}|0\rangle
\end{equation}

\subsection{Effect of energy error on eigenstate property expectation estimation}
\label{sec:unknown_energy}
In this section, we analyse the effect of energy error on eigenstate property expectation estimation. Our result indicates that as long as the energy estimation is $\epsilon$-close to the true energy, the observable error can be bounded.

Suppose the energy has an estimation error $\kappa:=|\hat{E}_j - E_j | \leq \Delta$.
The observable expectation will become
\begin{equation}
\label{eq:observ_expec_finite}
	\hat{O}_{\tau, x_c, s_c}(\hat{E}_j) = \frac{\hat N_{\tau, x_c, s_c}(O,\hat{E}_j)}{\hat D_{\tau, x_c, s_c}(\hat{E}_j)}. 
\end{equation}

In the presence of estimation error $\kappa \neq 0$, $g_{\tau  }(H-\omega) $ will tend to be zero $\|g_{\tau \rightarrow \infty}(H-\omega) \| = 0$.
In such a case, we consider the projector $\hat P_i = \ket{u_i}\bra{u_i}$.
\autoref{eq:g_tau_eq} will need to be modified a little. For a general $E$, we have
\begin{equation}
	g_{\tau}(H-\omega)   = \sum_{i} g_{\tau}(E_i-\omega) \hat P_i.
	\label{eq:g_tau_eq_EnErr}
	\end{equation}

When $\tau \geq \frac{1}{\Delta} \sqrt{\ln(2/\varepsilon_{\tau})}$, we have 
\begin{equation}
	\|g_{\tau}(H-\omega) - g_{\tau}(E_j - \omega)\hat P_j \| \leq \varepsilon_{\tau} /2 
	\end{equation}
The error of the numerator can be bounded by
\begin{equation}
	|N_{\tau}(O,E) - g_{\tau}^2(E_j - \omega)N(O)| \leq \varepsilon_{\tau}\|O\|
	\label{eq:N_O_eq_EnErr}
\end{equation}

Compared to the result with an accurate estimation, the only difference is that $N(O)$ and $D$ is coupled with an additional factor $ g_{\tau}^2(\kappa)$. 


Recall that the objective is to estimate $N(O)$. Therefore, we put the factor $g_{\tau}^{-2}(\kappa) $ coupled with $N_{\tau}$, and denote
\begin{equation}
	\varepsilon_N :=  |g_{\tau}^{-2}(\kappa) N_{\tau}(O,E) - N(O)|.
\end{equation}
Since $g_{\tau}^{-2} > 1$, the error $\varepsilon_N$ compared to the previous estimation error of the numerator is amplified.

The results concerning a finite cutoff and number of samples can be  derived similarly to that in \autoref{thm:observ_estimation_detailed}.
More specifically, we have
\begin{equation}
\begin{aligned}
\left|\langle\hat{O} \rangle_{\tau, x_c, s_c}-\left\langle O\right\rangle\right| & =\left|\frac{\hat{N}_{\tau, x_c, s_c}\left(O,\omega \right)}{\hat{D}_{\tau,x_c, s_c}(\omega)}-\frac{N\left(O,\omega\right)}{D(\omega)}\right| \\
& =\left|\frac{D(\omega) \hat{N}_{\tau, x_c, s_c}(O,\omega)-N\left(O,\omega\right) \hat{D}_{\tau, x_c, s_c}(\omega)}{D(\omega) \hat{D}_{\tau, x_c, s_c}(\omega)}\right| \\
& \leq \frac{1}{|D(\omega) \hat{D}_{\tau, x_c, s_c}(\omega)|} \bigg( \left| D(\omega) \hat{N}_{\tau, x_c, s_c}(O) - \frac{\hat{N}_{\tau, x_c, s_c}\left(O,\omega \right) \hat{D}_{\tau, x_c, s_c}(\omega)}{g^2_{\tau}(\kappa)}  \right |\\
&  +\left| \frac{\hat{N}_{\tau, x_c, s_c}\left(O,\omega \right) \hat{D}_{\tau, x_c, s_c}(\omega)}{g^2_{\tau}(\kappa)}  -N(O,\omega) \hat{D}_{\tau, x_c, s_c}(\omega) \right | \bigg) \\
& \leq  \left|\frac{D(\omega) \varepsilon_N(O,\omega)+N(O,\omega) \varepsilon_D}{g^2_{\tau}(\kappa) ( D(\omega)^2-D(\omega) \varepsilon_D ) }\right| \\
& \leq\left|\frac{D(\omega) \varepsilon_N(O,\omega)+\left(N(O,\omega)+D(\omega)\right) \varepsilon_D}{g^2_{\tau}(\kappa) D^2}\right| \\
& \leq g^2_{\tau} \eta^{-1}\left(2\left\|O\right\|+1\right)\left(\varepsilon_\tau+\varepsilon_c\right)+ g^2_{\tau} \eta^{-1}\left(\left\|O\right\|+\left\|O\right\|_1+1\right) \varepsilon_n.
\end{aligned}
\end{equation}

Compared to the case with known eigenenergy, the observable estimation error $\varepsilon$ will be amplified by a factor $g_{\tau}^{-2}(\kappa) = \exp(2 \tau^2 \kappa^2)$. 

To ensure the estimation is nonvanishing, we require $\tau \kappa \leq c$, which indicates the energy precision needs to satisfy 
\begin{equation}
	\kappa \leq c \mathcal{O} ( \Delta \log^{-1}((\eta\varepsilon)^{-1}).
\end{equation}

\subsection{Discretised version}
\label{sec:discretised_version}

For the Gaussian-type spectral filter, we show a deterministic version of the LCU decomposition of the spectral filter in \autoref{Eq:coolingdef_main}. Using the Gaussian integral, also known as the Hubbard-Stratonavich transformation which is widely used in field theories and auxiliary-field quantum Monte Carlo~\cite{huggins2022unbiasing}, we have
\begin{equation}
	g_{\tau}(H) = e^{-\tau^2H^2} = \frac{1}{\sqrt{2 \pi}} \int dx e^{-x^2/4} e^{-i \tau x H }. 
\end{equation}
By converting the integral into a summation, we define a discretised version of $g_{\tau,x_c}^{(D)}(H-\omega)$ with  a maximum evolution time $x_c$ as
 \begin{equation}
 \label{eq:g_tau_xc_discretised}
 	g_{\tau,x_c}^{(D)}(H-\omega) = \frac{1}{\sqrt{2\pi}} \sum_{j=-N_m}^{N_m}  b e^{-x_j^2/4} e^{i \tau x_j (H- \omega) }
 \end{equation}
 with total number of steps $N_m$, the step size $b = x_c / N_m$, and  $x_j = jb$, a superscript $D$ denoting a discretised version.
 An infinite sum of $g_{\tau,x_c}^{(D)}(H-\omega)$ is given by  \begin{equation}
 	g_{\tau}^{(D)}(H-\omega) = \frac{1}{\sqrt{2\pi}} \sum_{j=-\infty}^{\infty}  b e^{-x_j^2/4} e^{i \tau x_j (H- \omega) }
 \end{equation}
 with total number of steps $N_m$, the step size $b = x_c / N_m$, and  $x_j = jb$.

The discretisation error can be bounded by
\begin{equation}
	\varepsilon_{d,x_c} := |g_{\tau,x_c} - g_{\tau,x_c}^{(D)}| \leq |g_{\tau} - g_{\tau}^{(D)}| + |g_{\tau} - g_{\tau,x_c}|
 + |g_{\tau}^{(D)} - g_{\tau,x_c}^{(D)}|
\end{equation}
We define the discretisation error with an infinite expansion as $\varepsilon_{d} = |g_{\tau} - g_{\tau}^{(D)}| $ the truncation error $\varepsilon_{x_c}:=  |g_{\tau} - g_{\tau,x_c}|$ and its discretisation as $\varepsilon_{x_c}^{(D)} := |g_{\tau}^{(D)} - g_{\tau,x_c}^{(D)}| $

Since $e^{-x^2}$ is a monotonic function, the discretisation form is less than the integral. From the definition of $ g_{\tau,x_c}^{(D)}$ \autoref{eq:g_tau_xc_discretised}, we can check that $\varepsilon_{x_c}^{(D)} \leq\varepsilon_{x_c} $. From \autoref{prop:error_finite_time_x_s}, we know that $\varepsilon_{s_c } \leq \exp(- x_c^2/2)$ and thus $\varepsilon_{x_c}^{(D)} $ is bounded by $\exp(- x_c^2/2)$ as well.
With the result derived in \cite{keen2021quantum}, $\varepsilon_{d}$ can be bounded by
\begin{equation}
	\varepsilon_{d} \leq \exp \left ((\frac{2\pi}{b} - \tau)/2  \right)^2
\end{equation}
For an equal distribution of error, we choose to set
$	\varepsilon_d = \varepsilon_{x_c} = \varepsilon/3$.
When the stepsize is set as $b = 2\pi /(x_c + \tau) = \mathcal{O} (\Delta { (\ln(\varepsilon^{-1})})^{-1/2})$, the total error is bounded by $\varepsilon$. The total number of steps is 
\begin{equation}
    N_m = x_c(x_c+\tau)/2\pi \leq \frac{2}{\pi \Delta
} \ln \frac{2}{\varepsilon_{\tau}} = \mathcal{O} (\Delta^{-1}  \ln(\varepsilon^{-1}))
\end{equation}
with $\varepsilon_\tau$ defined in \autoref{eq:epsilon_requirement}.


\section{Eigenenergy estimation }	
\label{sec:En_estimate}


In this section, we discuss the gate complexity of eigenenergy estimation described in \autoref{problem:eigenenergy}. We provide a proof for the second part of \autoref{thm:observ_estimation_main}.
We first discuss how to use the denominator to estimate the eigenenergy $E_j$.
Intuitively, we can find that $D_{\tau}(\omega)$ indeed shows a coarse-grained energy spectrum.
For the initial state $\ket{\psi_0} = \sum_i c_i \ket{u_i}$, the spectrum of the initial state can be characterised by 
\begin{equation}
	P(E) = \sum_i |c_i|^2 \delta(\omega - E_i). 
\end{equation}
One can prove that $D_{\tau} (E)=[g_{\tau}^2 \star P ] (E)$ where $ \star $ denotes the convolution of two functions.

Suppose we have a prior knowledge of $E_j \in [E_j^L, E_j^R]$.
It is worth noting that we cannot distinguish eigenenergies that are very close to each other. The eigenenergies that are close to each other could be merged and regarded as a broadened eigenenergy.
Here, we assume that the target resolution is less than the energy gap, i.e., $\kappa < \Delta$, and $ E_j^R < E_{j} + \Delta/2$, $ E_j^L > E_{j} - \Delta/2$. 
Given this range, the $j$th eigenenergy can be searched by
\begin{equation}
\label{eq:E_search_ideal}
    E_j = \argmax_{\omega\in[E^L_j, E^U_j]} \hat{D}_\tau(\omega).
\end{equation}
The maximum of $D_\tau(\omega)$ within the range $[E^L_j, E^U_j]$ gives us an estimate of eigenenergy $E_j$.

In practice, we can only obtain an estimation $\hat{D}^{(x_c)}_\tau(E)$ of $D_\tau(E)$, when considering finite cutoff time $x_c$, segment number $
\nu_c$, truncation $s_c$ and number of samples $N_s$. The eigenenergy is determined by
\begin{equation}
\label{eq:E_search_noisy}
\hat{E}_j := \argmax_{\omega\in[E^L_j, E^U_j]}  \hat{D}_{\tau,x_c, s_c}(\omega). 
\end{equation}

 
Similar to property estimation in \autoref{sec:observ_estimate}, the error sources include a finite imaginary time, a finite cutoff of real-time evolution, discretization error, and the statistical error due to Hamiltonian simulation and finite number of samples 
\begin{equation}
	\hat D_{\tau,x_c, s_c}(\hat E_j) - D_{\tau} (E_j) \leq \varepsilon_{\tau} + \varepsilon_{x_c} + \varepsilon_{s_c} + \varepsilon_{n}
	\end{equation}
Based on the error dependence, we can estimate the resource requirements (i.e., circuit depth and number of samples) for eigenenergy estimation. 
Compared to \autoref{sec:observ_estimate}, the only difference is that the denominator is a function of $\omega$  in energy estimation, while we take $\omega = E_j$ in property estimation in \autoref{sec:observ_estimate}.
The following lemma establishes the error due to a finite $\tau$.

\begin{lemma}[Error due to a finite $\tau$ (Proposition 3 in \cite{zeng2021universal})]
\label{lemma:error_finite_tau}
When $\tau \geq \frac{2}{\Delta} \sqrt{\ln(2/\varepsilon_{\tau})}$, 	$|D_{\tau}(\omega) - \eta g_{\tau}^2 (\omega - E_j)| \leq \varepsilon_{\tau}$.
\end{lemma}

The rest of the proof is nearly identical to \autoref{sec:observ_estimate}.
We give the result here.
 

\begin{theorem}[Eigenenergy estimation for generic Hamiltonians]
\label{thm:En_estimation_detailed}

Suppose that $\hat E_j$ is determined by \autoref{eq:E_search_noisy} and the conditions and assumptions in \autoref{problem:eigenenergy} hold.  

Case I: Suppose that we choose the time-segment number function $\nu_c = \mathcal{O} \left( (\lambda \kappa^{-1} \ln(\eta^{-1} ))^{1+ \frac{1}{4k+1}}  \right )$ for realising the real-time evolution using the method in \autoref{alg1_main} (with the truncation order $s_c = \mathcal{O} (\ln(\nu_c)/\ln \ln(\nu_c)))$).  
We can achieve the error of eigenenergy estimation within $\kappa$, 
$|\hat E_j - E_j | \leq  \kappa $ with a success probability at least $1-\vartheta$ when the number of samples is $N_s = \mc O\left ( \eta ^{-2}  \ln(1/\vartheta) \right)$.

Case II: When we choose the time-segment number function $\nu_c = \mathcal{O} \left( (\lambda \Delta^{-1} \ln(\eta^{-1} \kappa^{-1}))^{1+ \frac{1}{4k+1}}  \right )$ for realising the real-time evolution, one can achieve the error of eigenenergy estimation within $\kappa$, 
$|\hat E_j - E_j | \leq  \kappa $ with a success probability at least $1-\vartheta$ when the number of samples is $N_s = \mc O\left (\eta^{-2} \Delta^{4}\kappa^{-4} (\ln ({\kappa^{-2} \eta^{-1}}) )^{2} \ln (1/\vartheta)  \right)$. 

 \end{theorem}

\begin{proof}

We start by proving the result in Case I. 
From \autoref{eq:E_search_noisy}, we have
\begin{equation}  
\hat{D}_{\tau, x_c,s_c}(\hat{E}_j) \geq \hat{D}_{\tau,x_c,s_c}(E_j).
\end{equation}
Then we bound the difference between $D_\tau (\omega)$ and $\eta g_{\tau}^2(\omega - E_j)$ using \autoref{lemma:error_finite_tau}.
That is, when we set the imaginary-time as 
$
	\tau \geq \frac{2}{\Delta} \sqrt{\ln \frac{2}{\eta \varepsilon_{\tau}}}
 $
we have
\begin{equation}
    |D_{\tau}(E_j) - \eta| \leq \eta \varepsilon_{\tau},~ |D_\tau(\hat{E}_0)-\eta g_\tau(\hat{E}_0-E_0)^2| \leq \eta \varepsilon_\tau.
\end{equation}
Here it is worth noting that we do not have to know the value of $\eta$.

The maximum real-time  is 
\begin{equation} \label{eq:maximum_real_time_En} 
    t_c = \tau x_c = \frac{4}{\Delta} \sqrt{\ln(2/\eta \varepsilon_\tau)} \sqrt{\ln(2/\eta \varepsilon_c)}.
\end{equation} 
According to \autoref{prop:error_finite_time_x_s}, when the segment number is set by
\begin{equation}
\begin{aligned}
		\nu_c &= 4 \left(\frac{2(e+c_k)}{\ln 2}\right)^{\frac{1}{4k+1}} (\lambda \tau x_c)^{1+\frac{1}{4k+1}} 
	\end{aligned}
\end{equation}
 ${U}_{2k}^{(s_c)}(t_i)$ is a $(2, \varepsilon_{s_c})$-LCU formula of ${U}(t_i)$ with the approximation error given by
$
	\max_{t_i} \|{U}_{2k}^{(s_c)}(t_i) - U(t_i)\| \leq \varepsilon_{s_c},
 $ and we have $|D_{\tau, x_c, s_c} (\omega) - D_{\tau} (\omega) | \leq   \eta \varepsilon_{c}$  with $\varepsilon_{c} = 3(\varepsilon_{x_c} + \varepsilon_{s_c})$

According to \autoref{prop:measurment_error},  
when the number of measurements $N_s = {32} \eta^{-2}{\varepsilon_n^{-2}} \ln(1/\vartheta)$, we have 
\begin{equation}
\begin{aligned}
& \left|D_\tau\left(E_j\right)-\hat{D}_{\tau,x_c, s_c}\left(E_j \right)\right| \leq  \eta (\varepsilon_c + \varepsilon_n) \\
& \left|D_\tau(\hat E_j)-\hat{D}_{\tau,x_c, s_c}(\hat E_j )\right| \leq  \eta (\varepsilon_c + \varepsilon_n)\\
\end{aligned}
\end{equation}
with a success probability of  $1-\vartheta$.
According to \cite{zeng2021universal},  the following inequality holds
\begin{equation}
\begin{aligned}
|\hat{E}_j-E_j| & \leq \frac{1}{\tau} g_\tau^{-1}\left(\sqrt{1-2\left(\varepsilon_\tau+\varepsilon_c + \varepsilon_n \right)}\right). 
\end{aligned}
\end{equation}
When we set the precision for each component
\begin{equation}
	\begin{aligned}
	\label{eq:epsilon_requirement_En}
& \varepsilon_\tau=\varepsilon_c=\varepsilon_n \leq \frac{1}{6} (1 - 
e^{-1}) \leq 0.1
\end{aligned}
\end{equation} 
and $\tau \geq \kappa^{-1}$, we can make sure that the energy estimation precision is no greater than $\kappa$.
From  the precision requirement \autoref{eq:epsilon_requirement_En} and \autoref{eq:maximum_real_time_En}, we require 
\begin{equation}
    t_c \geq \kappa^{-1} \ln(20 \eta^{-1} )
\end{equation}
when $\kappa \leq \Delta/4$.
Using \autoref{prop:error_finite_time_x_s}, the segment number is set to be
\begin{equation}
\begin{aligned}
		\nu_c &= 4 \left(\frac{2(e+c_k)}{\ln 2}\right)^{\frac{1}{4k+1}} (\lambda \kappa^{-1} \ln(20 \eta^{-1} ) )^{1+\frac{1}{4k+1}} = \mathcal{O} \left( (\lambda \kappa^{-1} \ln(\eta^{-1} ))^{1+ \frac{1}{4k+1}}  \right )
	\end{aligned}
\end{equation}
and set  $s_c$ by 	\autoref{eq:sc_requirement}.
 
In \cite{wang2023quantum}, Wang et al. showed that when using a Gaussian filter, the precision dependence can be improved, where the eigenenergy is also determined by \autoref{eq:E_search_noisy}.
Compared to Case I, the key difference is that show that the distance  $| D_{\tau}(\omega)  - \eta|$ is modulated by the estimation error $\omega - E_j$.
When we set the imaginary time
\begin{equation}
    \tau = \frac{1}{0.9 \Delta} \sqrt{ \ln ({20}{\kappa^{-2} \eta^{-1}})} = \mc O(\Delta^{-1} (\ln ({\kappa^{-2} \eta^{-1}}) )^{-\frac{1}{2}} )
\end{equation}
we can distinguish $E_j$ from the others.
The cutoff by $x_c$ can be similarly obtained. 
Using \autoref{prop:error_finite_time_x_s}, when the segment number is set to be
\begin{equation}
\begin{aligned}
	\mathcal{O} \left( (\lambda \Delta^{-1} \ln(\eta^{-1} \kappa^{-1}))^{1+ \frac{1}{4k+1}}  \right )
	\end{aligned}
\end{equation}
the eigenenergy error can be bounded by $\kappa$ by using the results in \cite{wang2023quantum}.
Note that in \cite{wang2023quantum}, the error is
$\varepsilon_n = \tau \kappa$. Thus, the number of samples 
\begin{equation}
    N_s = \mc O (\eta^{-2} \tau^{-4}\varepsilon^{-4} \ln (1/\vartheta)) =  \mc O (\eta^{-2} \Delta^{4}\kappa^{-4} (\ln ({\kappa^{-2} \eta^{-1}}) )^{2} \ln (1/\vartheta) )  
\end{equation}


\end{proof}

\revise{
For eigenstate property estimation, the gate complexity depends on the energy gap $\Delta$. 
For eigenenergy estimation, the total gate count nearly reaches the Heisenberg limit $\mathcal{O}(\kappa^{-(1+o(1))})$~\cite{lin2021heisenberg,zeng2021universal}, which may not be directly related to $\Delta$.}

Given a Hamiltonian with parameters $n$, $L$, $\wt(H)$, and $\wtm(H)$, the gate complexity of eigenenergy estimation can be similarly obtained by using \autoref{prop:gate_cost_observ_dynamics}.

\section{Circuit compilation and gate cost for block-encoding-based methods}
\label{sec:gate_gs_BE}

\subsection{Stage Setting}

In this section, we briefly introduce and estimate the gate cost for block-encoding-based methods. 
Here we mainly discuss the cost based on the result by Google's team in \cite{babbush2018encoding} which is used in our numerical simulation. There are considerable progress in reducing the cost for block encoding. We will not introduce these advanced techniques which are not the main focus of this work.

To estimate the gate cost of each algorithm, we synthesize the circuits to \textsc{cnot} gates, single-qubit Clifford gates and $T$-gates. The \textsc{cnot} gate number is more important for a near-term application on a quantum computer with no or limited fault tolerance, while the $T$-gate number is more critical for a long-term application on a fully fault-tolerant quantum computer.
In some subroutines of the above algorithms, a direct estimation of the $T$-gate number is hard to obtain. In this case, we first synthesize the circuits to \textsc{cnot} gates, single-qubit Clifford gates and single-qubit $Z$-axis rotation gates $R_{z}(\theta)$. Then we estimate the $T$-gate number $n_{T}$ using the $R_{z}(\theta)$ gate number $n_{R z}$. 

We consider the optimal ancilla-free gate synthesis algorithm in Ref.~\cite{ross2016optimal}, which requires $3 \log _{2}(1 / \varepsilon)+\mathcal{O}(\log \log (1 / \varepsilon))~T$-gates to approximate the $R_{z}(\theta)$ gate to a precision $\varepsilon$. Here, we set the gate synthesis error of each $R_{z}(\theta)$ $\varepsilon_{CS} $ to be a small value compared to the total error. 
In practice, we should determine the resource overhead $c_{T}$ based on the number of $R_{z}$ gates in the quantum algorithm.

We remark that, if we are allowed to introduce extra ancillary qubits and entangling Clifford gates, one can further reduce the required $T$ gates to 
$$1.15 \log _{2}(1 / \varepsilon) + 9.2$$ 
using a repeat-until-success strategy proposed in~\cite{bocharov2015efficient}. However, this will introduce extra ancillary qubit requirements and more \textsc{cnot} gate costs. 

In our resource analysis, to streamline the comparison, we exclude the observable estimation error due to a finite sampling cost. We will focus on the circuit depth to achieve a certain level of accuracy of the RCLU formula. 
Here, we remark that while the RLCU method cannot prepare the state, it can effectively prepare the eigenstate at the level of expectation.  
To get the resource cost, the key component is to get the 
segment number $\nu_c$, which is directly related to the maximum real evolution time $t_c$.
The gate count for \textsc{cnot} gates and single-qubit rotation gates can be obtained by using \autoref{prop:gate_cost_observ_dynamics}.

Let us define 
\begin{equation}
\begin{aligned}
\label{eq:H_character}
n_{L} &:=\left \lceil\log _{2} L\right \rceil \\
\Lambda &:=\max _{l} \alpha_{l} \\
\end{aligned}
\end{equation}
In the standard block encoding procedure~\cite{low2019hamiltonian}, the $n$-qubit Hamiltonian $H$ is encoded in a $\left (n_{L}+n\right )$ -qubit unitary, $\operatorname{select}(H)$
\begin{equation}
\operatorname{select}(H):=\sum_{l=1}^{L}|l\rangle\langle l|  \otimes H_{l}
\end{equation}
Denote
\begin{equation}
|G\rangle:= \prep \ket{0} ^{n_L} =\frac{1}{\sqrt{\lambda}} \sum_{l=1}^{L} \sqrt{\alpha_{l}}|l\rangle,
\end{equation}
then we have
\begin{equation}
\label{eq:H_BE}
H=\lambda(\langle G| \otimes I) \operatorname{select}(H)(|G\rangle \otimes I)
\end{equation}
which indicates that $H$ is block-encoded into $\operatorname{select}(H)$.
Here, $\prep$ encodes the amplitude into the state on the ancillary space, and it is also referred to as the amplitude-encoding unitary or \textsc{prepare} operation in the literature. 

Reflection unitary $R$, which is $ 
    R:=(I-2|0\rangle\langle 0|) \otimes I
$
where the operation $(I-2|0\rangle\langle 0|)$ is defined on the ancillary space with dimension $n_L$.

\subsection{Gate cost}
\subsubsection{The index enumeration circuit}

We follow the circuit construction in Ref.~\cite{babbush2018encoding} to build the amplitude encoding operation (denoted by PREP or $B(x)$) and controlled select operation $\mathrm{C}$-select$(H)$. A major gadget of both operations is the following operation,
\begin{equation}
    \text { C-select}(X)=\sum_{a=0}^{1}|a\rangle\langle a|\otimes \sum_{l=1}^{L}| l\rangle\langle l| \otimes\left (X_{l}\right )^{a},
    \label{eq:C-selectX}
\end{equation}
where $X_{l} \in \{I, X\}$ is a single-qubit Pauli operator. The value of $X_{l}$ depends on the value of $l$ stored in the classical register. We can regard $\mathrm{C}$-select$(X)$ as a simplified version of $\mathrm{C}$-select$(H)$, where $H=\sum_{l=1}^{L} X_{l}$ is a single-qubit Hamiltonian where $X_{l}$ is either $I$ or $X$, based on the storage in the classical register.

In Sec. III in Ref.~\cite{babbush2018encoding}, the authors construct a ``sawtooth'' circuit to realise the $\mathrm{C}$-select$(X)$ gate (which is called the indexed operation in the original paper). In the simplified circuit of $\mathrm{C}$-select$(X)$ in Fig. 7 in Ref.~\cite{babbush2018encoding}, we need $(L-1)$ computing AND operations, $(L-1)$ uncomputing AND operations, $L$ control- $X_{l}$ gates, and $(L-1)$ extra \textsc{cnot} gates. Suppose that we decompose the computing and uncomputing AND operations based on Fig. 4 in Ref.~\cite{babbush2018encoding}, and synthesize all the gates to Clifford $+T$ gates. We present the following observation.

\begin{observation}
[Gate cost in the index enumeration circuit]
\label{observ:In_circ}
If we construct the index enumeration circuit $\mathrm{C}$-select$(X)$ defined in \autoref{eq:C-selectX} following the 'sawtooth' way in Ref.~\cite{babbush2018encoding} and synthesize all the gates to \textsc{cnot} gates, single-qubit Clifford gates and $T$ gates, then we can realise $\mathrm{C}$-select $(X)$ using

\begin{enumerate}
  \item $(6 L-5)$ \textsc{cnot} gates;

  \item $(4 L-4)$ T gates.
  \item $(2L-2)$ Hadamard gates.

\end{enumerate}

\end{observation}

\subsection{The amplitude encoding, select gates, and reflection gates}
Now, we estimate the gate cost in amplitude encoding, select, and reflection operations.
The  amplitude-encoding unitary $B$ realises the following transformation,
$$
B|0\rangle=\sum_{l=1}^{L} \sqrt{\omega_{l}}|l\rangle\left|\operatorname{temp}_{l}\right\rangle
$$
where
$
\omega_{l} :={\alpha_{l} / \lambda},
$
is the normalised amplitude of the Hamiltonian. Following Ref.~\cite{babbush2018encoding}, we assume that it is allowed to introduce temporary storage $|\mathrm{ temp}_{l}\rangle$ during the amplitude encoding. This will not cause problems as long as we finally disentangle the system    $|\mathrm{ temp}_{l}\rangle$ during the implementation of $B^{\dagger}$.

The dominant subroutine of the \textsc{prepare} circuit is the \textsc{subprepare} circuit defined in Eq. (48) in Ref.~\cite{babbush2018encoding}, which realises the amplitude encoding to different orbitals, ignoring the spin information first. In our discussion, we first ignore the detailed structure of the Hamiltonian $H$ with respect to different spins. In this case, we can treat \textsc{subprepare} circuit to be the \textsc{prepare} circuit. To realise the \textsc{subprepare} circuit, we use the method introduced in Sec. IIID in Ref.~\cite{babbush2018encoding}. The basic idea is to first prepare ancillaries with uniformly distributed coefficients over indecies $l$ and then use a pre-determined binary representation of a probability $\operatorname{keep}_{l}$, to perform a controlled-swap on the amplitude register $l$ and another predetermined amplitude location alt$_{l}$. With well-designed values of swap probability keep send swap location alt$_{l}$, we can use the circuit in Fig. 11 in Ref.~\cite{babbush2018encoding} to realise the \textsc{subprepare} circuit.

Suppose we want to realise the amplitude encoding with an accuracy of $\varepsilon_{A E}$, that is, to realise the following transformation,
\begin{equation}
    B^{\varepsilon_{A E}}|0\rangle=\sum_{l=1}^{L} \sqrt{\tilde{\omega}_{l}}|l\rangle\left|\operatorname{temp}_{l}\right\rangle
    \label{eq:B_AE}
\end{equation}
where $\tilde{\omega}_{l}$ is a $n_{AE}$ bit approximation to the true value $\omega_{l}$
$$
\left|\tilde{\omega}_{l}-\omega_{l}\right| \leq \varepsilon_{A E}, \quad l=1, \ldots, L .
$$
The number of ancillary qubits is required to be
$n_{AE}=\left \lceil - \log _{2}  {\varepsilon_{AE}}\right \rceil.
$

Due to the relation of the rescaled spectrum  by block encoding, we have
$ 
\varepsilon_{\prep} = \frac{\varepsilon}{\lambda}.
$
The relation of the amplitude encoding error and $\prep$ error could be derived by considering the norm of the Hamiltonian, and a simple relation is given by
$
    \varepsilon_{AE} \sim \frac{\varepsilon_{\prep} }{L}. 
 $
Following  Ref.~\cite{babbush2018encoding} (see Fig.~11), we need to introduce at least $2 n_{AE}+2 n_{L} + 1$ extra ancillary qubits, 
and $n_{L}:=\left \lceil\log _{2} L\right \rceil$. 
To simplify the gate cost, we assume $L$ is a power of 2.  In this case, the first layer of the circuit in Fig. 11 in Ref.~\cite{babbush2018encoding} can be realised using Hadamard gates. If $L$ is not a power of 2, additional quantum resources are needed. 

The second and the third layer of the circuit requires the QROM circuit in Fig. 10 in Ref.~\cite{babbush2018encoding}, which is a modified version of the index enumeration circuit $\mathrm{C}$-select$(X)$ defined in \autoref{eq:C-selectX}.


Based on \autoref{observ:In_circ},  the second layer of date loading requires $5(L-1) + L(n_L + n_{AE}) $ \textsc{cnot} gates, $4(L-1)$ T gates. The third layer is a coherent inequality test, which requires $(n_{AE}-1)$ AND and uncomputing AND operations, and additional $6(n_{AE}-1)$ \textsc{cnot} gates plus 1 Toffoli gate. Thus, it requries $11n_{AE}-5$ \textsc{cnot} gates and $4n_{AE}+3$ T gates.

The fourth layer is a Fredkin gate, which is a controlled swap gate. Following Fig.~5 in \cite{mosca2021polynomial},  this gate can be synthesised into Clifford + T gates using $8 n_{L}$ \textsc{cnot} gates and $7 n_{L} T$ gates.

\begin{observation}[Ancillary qubit and gate costs in the second-type amplitude encoding operation]
\label{observ:gs_amp_encode}
If we synthesize the $B^{\varepsilon_{A E}} $ unitary defined in \autoref{eq:B_AE}  to \textsc{cnot} gates, single-qubit Clifford gates and $T$ gates, then the approximate ancillary and gate cost of $B^{\varepsilon_{A E}}$ are listed as follows,

\begin{enumerate}
  \item $2 n_{AE}+2 n_{L} + 1$ extra ancillary qubits; 
  \item $  n_L(L+8) + n_{AE}(L+11) + 5L-10 $ \textsc{cnot} gates;

  \item $4\left (L+ {n_{AE}} \right ) + 7 n_L + 3$ T gates;

\end{enumerate}
Here, $n_{AE}:=\left \lceil -\log _{2} {\varepsilon_{A E}}\right \rceil$ and $n_{L}:=\left \lceil\log _{2} L\right \rceil$.

\end{observation}


With the above result, it is easy to analyse the gate cost of the $\mathrm{C}$-select$(H)$ gate. A straightforward implementation of the $\mathrm{C}$-select$(H)$ gate is to replace of $X_{l}$ gate in $\mathrm{C}$-select$(X)$ defined in \autoref{eq:C-selectX} to multi-qubit Pauli gates $P_{l}$. 
For instance, consider the transverse field Ising model
\begin{equation}
    H = J \sum_i \sigma_{i}^z\sigma_{i+1}^z + h \sum_i \sigma_{i}^x
\end{equation}
with the periodic boundary condition. The gate cost for the lattice Hamiltonian is shown in \autoref{observ:gs_select_Ising}.
 



\begin{corollary}[Gate cost in the $\mathrm{C}$-$\selectH$ operation of the lattice model]
\label{observ:gs_select_Ising}

If we construct the controlled-select circuit $\mathrm{C}$-$\selectH$ of the lattice model following the 'sawtooth' way in Ref.~\cite{babbush2018encoding} and synthesize all the gates to \textsc{cnot} gates, single-qubit Clifford gates and $T$ gates, then the approximate gate cost of $\mathrm{C}$-$\selectH$ is listed as follows,

\begin{enumerate}
  \item $5(L-1) + \mathrm{wt}(H)$ \textsc{cnot} gates;

  \item $(4 L-4) $ T gates.

\end{enumerate}
\end{corollary}

Next, we consider the gate cost for the fermionic Hamiltonian in \autoref{eq:hamiltonian_molecule_general}.
The fermionic Hamiltonian can be mapped to a qubit form by JW transformation. We suppose there are $L$ terms with distinctive coefficients in total.
To further improve the gate cost in a fermionic Hamiltonian, Ref.~\cite{babbush2018encoding} introduces an accumulator during the Pauli gate query process (Sec. IIIB and Fig. 8 in Ref.~\cite{babbush2018encoding}). The accumulator will 'accumulate' the effect of the Pauli operators accessed in the previous data queries and save the \textsc{cnot} gate cost. Using this improved select operation, we can reduce the \textsc{cnot} cost for each Pauli operator $P_{l}$ to a constant independent of the weight of $P_{l}$. 
We will use an optimistic estimate of the \textsc{cnot} gate cost for controlled-$P_{l}$ operations, which is $5(L-1) + 3L = 8L - 5$ for QSP in the numerical simulation. The actual cost should be greater than this value.

\begin{observation}[Ancillary qubit and gate costs in the reflection operation~\cite{maslov2016advantages}]
\label{observ:gs_reflect}
If we construct the reflection operation $I-2|0\rangle\langle 0|$ on $n$ qubits following the methods in Proposition 4 in \cite{maslov2016advantages} and synthesise all the gates to \textsc{cnot} gates, single-qubit Clifford gates and $T$ gates, then the approximate ancillary qubit and gate costs are listed as follows,

\begin{enumerate}
  \item $\left \lceil\frac{n-3}{2}\right \rceil$ ancillary qubits

  \item $(6 n-12)$ \textsc{cnot} gates;

  \item $(8 n-17) $ T gates.

\end{enumerate}
\end{observation}

\subsection{Gate cost for Trotter methods}
In the Trotter methods, we first divide the real-time evolution into $\nu$ segments,
\begin{equation}
  e^{-i H t}=\left (e^{-i H x}\right )^{\nu}.
\end{equation}
where $x:=t / \nu$.
The first-order Trotter formula is
\begin{equation}
    S_{1}(x)=\prod_{l=1}^{L} e^{-i x H_{l}}.
\end{equation}
and the second-order Trotter formula is
\begin{equation}
   S_{2}(x)=\prod_{l=L}^{1} e^{-i(x / 2) H_{l}} \prod_{l=1}^{L} e^{-i(x / 2) H_{l}}. 
\end{equation}
The $(2 k)$th-order Trotter formula is
\begin{equation}
S_{2 k}(x)=\left [S_{2 k-2}\left (p_{k} x\right )\right ]^{2} S_{2 k-2}\left (\left (1-4 p_{k}\right ) x\right )\left [S_{2 k-2}\left (p_{k} x\right )\right ]^{2}
\end{equation}
with $p_{k}:=1 /\left (4-4^{1 /(2 k-1)}\right )$ for $k \geq 1$

We use the results for $2 k$th-order Trotter formula from \cite{childs2018toward,childs2019faster} to analyse the Trotter Cost. We put their results below for the ease of readers. 

\begin{lemma}
[Simple Trotter error bound for the $2 k$th-order Trotter formula (\cite{childs2018toward,childs2019faster})]
\label{prop:bound_2k_stTrotter}
Let $H=\sum_{l=1}^{L} H_{l}$ be a Hamiltonian consisting of $L$ summands and $t \geq 0$. We denote
\begin{equation}
\begin{aligned}
&a_{2 k}(\nu):=2 \frac{\left (2 \cdot 5^{k-1} L \Lambda t\right )^{2 k+1}}{(2 k+1) ! \nu^{2 k+1}} e^{2 \cdot 5^{k-1} L \Lambda t / \nu}, b_{2 k}(\nu):=\frac{L^{2 k}\left (2 \cdot 5^{k-1} \Lambda t\right )^{2 k+1}}{(2 k-1) ! \nu^{2 k+1}} e^{2 \cdot 5^{k-1} L \Lambda t / \nu},
\end{aligned}
\end{equation}
where $k \geq 1, \nu$ is the time segment number, $\Lambda$ is defined in Eq. (2). If we set the segment number $\nu$ to be
\begin{equation}
\nu_{2 k}^{\mathrm{det}}=\min \left\{\nu \in \mathbb{N}: \frac{\nu}{2} a_{2 k}(\nu) \leq \varepsilon\right\}
\label{eq:nu_2k_Trotter}
\end{equation}
for the deterministic  Trotter formula, or set $\nu$ to be
\begin{equation}
\nu_{2 k}^{\text {random }}=\min \left\{\nu \in \mathbb{N}: \frac{\nu}{2}\left (a_{2 k}(\nu)^{2}+2 b_{2 k}(\nu)\right ) \leq \varepsilon\right\},
\end{equation}
for the randomised   Trotter formula, then the spectral norm distance of the resulting simulation channel to the unitary channel of $e^{-i H t}$ is at most $\varepsilon$.
\end{lemma}

From \autoref{eq:nu_2k_Trotter}, the time segment can be roughly approximated by
\begin{equation}
    \nu_{2k} \leq  \frac{\left (2 \cdot 5^{k-1} L \Lambda t\right )^{ 1 + \frac{1}{2k}}}{((2 k + 1) !)^{\frac{1}{2k}} \varepsilon^{\frac{1}{2k}}}.  
\label{eq:nu_approxi_2k_Trotter}
\end{equation}
In our numerical simulation, we search the  required segment number by \autoref{eq:nu_2k_Trotter}.

\section{Ground state property estimation by quantum signal processing proposed by Lin and Tong}
\label{sec:gate_gs_QSP}

\subsection{Overview}

In this section, we first review the key ingredient of the seminal algorithm proposed by Lin and Tong \cite{lin2020near}.
Their method relies on the block encoding of a non-unitary matrix in the quantum circuit.
To establish a clear connection to \cite{lin2020near} and facilitate the reader, we will follow the notation and conventions used in \cite{lin2020near}.

To simplify the notations, we denote the \textsc{cnot} gate and T gate required for $\operatorname{select}(H)$ operation as $\SC$ and $\ST$, and the \textsc{cnot} gate and T gate required for $\textsc{prepare}$ operation as $\PC$ and $\PT$, respectively.

A matrix $A \in\mathbb{C}^{N\times N}$ where $N=2^n$ can be encoded in the upper-left corner of an $(n_L+n)$-qubit unitary matrix if 
\begin{equation}
\|A-\alpha(\bra{0^{n_L}}\otimes I) U (\ket{0^{n_L}}\otimes I)\|_2\leq \varepsilon.
\label{eqn:block_encoding}
\end{equation}
and we refer to $U$ as an $(\alpha,n_L,\varepsilon)$-block-encoding of $A$.
In this work, we consider the Hamiltonian written in an LCU form in \autoref{problem:observ}.
In the standard block encoding procedure, the $n$-qubit Hamiltonian $H$  can be explicitly block-encoded into $U_H := \textsc{prepare}^{\dagger} \cdot \operatorname{select}(H) \cdot \textsc{prepare}$, as shown in \autoref{eq:H_BE}.

The state preparation algorithm based on QSP is summarised below.
\begin{enumerate}
    \item Obtain the $(\lambda,n_L,0)$-block-encoding of a Hermitian matrix $H=\sum_{k}E_k \ket{\psi_k}\bra{\psi_k}\in\mathbb{C}^{N\times N}$, $N=2^n$, $\E_k \leq \E_{k+1}$. This block encoding is constructed by $U_H$.
    \item Construct a $(\lambda+|\mu|,n_L+1,0)$-block-encoding of matrix $H-\mu I$ using  of \cite[Lemma 29]{gilyen2019quantum} for any $\mu\in \mathbb{R} $. 
    \item Construct an $(1,n_L+2,\varepsilon)$-block-encoding of $$
R_{<\mu} = \sum_{k:\E_k<\mu} \ket{\psi_k}\bra{\psi_k} - \sum_{k:E_k>\mu} \ket{\psi_k}\bra{\psi_k}.
$$
This is realised by constructing a block encoding of the sign function $-S(\frac{H-\mu I}{\alpha+|\mu|};\delta,\varepsilon)$ for any $\delta$ and $\varepsilon$ where $S(\cdot, \delta, \varepsilon)$ is the sign function of degree $d = \frac{e}{2\delta} \ln (32\pi^{-1/2}\varepsilon^{-1}) $.  Note that if we assume further that $\Delta/2 \leq \min_k |\mu-\E_k|$, then we let $\delta=\frac{\Delta}{4\lambda}$,  all the eigenvalues of $-S(\frac{H-\mu I}{\lambda+|\mu|};\delta,\varepsilon)$ are $\varepsilon$-close to either {-1 or 1}, and thus $-S(\frac{H-\mu I}{\lambda+|\mu|};\delta,\varepsilon)$  is $\varepsilon$-close, in operator norm, to the reflector about the direct sum of eigen-subspaces corresponding to eigenvalues smaller than $\mu$:
    \item Using the block encoding of $R_{<\mu}$, we can construct an $(1,n_L+3,\varepsilon/2)$ block encoding of the projection operator $P_{<\mu} := \frac{1}{2} (R_{<\mu}  + I)$.
    \item Obtain the ground state with a success probability close to 1  by amplitude amplification. 
    \item Observable estimation.
    \end{enumerate}


\subsection{Gate count and depth analysis for QSP and QETU}

Next, we show the resource analysis of each step when compiling into elementary gates.
We denote the resource as $(\cdot, \cdot,\cdot,\cdot)$ with the four elements $(\cdot)$ representing the ancillary qubits, the number of \textsc{cnot} gates,  the number of T gates, and single-qubit $R_z$ rotations.

\begin{enumerate}
	\item Block encoding: $(n_L, \SC + 2\PC$, $ \ST + 2\PT, 0)$
	\item Controlled $\selectH$ and two $\prep$ operations: \sun{$(n_L+1, \SC + 2\PC$, $ \ST + 2\PT, 2)$}
	\item QSP of the sign function and hence the R operator: $ (n_L + 2, d (\SC + 2\PC) + 2d, d( \ST + 2\PT), 3 d)$ with $d = [  \frac{2e \lambda}{\Delta}\ln (32\pi^{-1/2}\varepsilon^{-1})   ]  $ obtained from \autoref{lem:approx_sign_overhead}
	\item \sun{Projector, which is a controlled version of $R$: $(n_L + 3, d(4 +  6 \SC + 2 \PC + 2 \ST + 2L ), d(7 \SC + 5 \ST + 2 \PT + 4 L ), 4 d)$. 
	\item Amplitude amplification. $  (n_L + 3 + \lceil \frac{n-3}{2} \rceil, d \gamma^{-1} (\SC + 2\PC + 6 n -10 ) , d\gamma^{-1} ( \ST + 2\PT + 8n-17), 3d \gamma^{-1})$ with $d =  [\frac{2e \lambda}{\Delta} \ln (32 \pi^{-1/2}\gamma^{-1}\varepsilon^{-1}) ] $.}

\end{enumerate}

In the $4$th step. In each block, the controlled gates:
\begin{enumerate}
    \item  Controlled Phase iterate: 2 \textsc{cnot} + 2 single-qubit rotation.
    Thus, the total single-qubit Pauli rotation gate is $4d$.
    \item Controlled $\selectH$ and 2 $\prep$. CNOT: $6 \SC + 2 \PC + 2 \ST$. The thrid $ 2 \ST$ is from that  one controlled T gate can be synthesised by 2 \textsc{cnot} and 2 $\sqrt{\mathrm{T}}$ gates.
    
    T gate: $7 \SC + 5 \ST + 2 \PT$, where the first $7$ comes from Toffolis gates, the second $5$ is from that
    \sun{one controlled T gate can be synthesised by 2 \textsc{cnot} and 2 $\sqrt{\mathrm{T}}$ gates, and we simply assume that 2 $\sqrt{\mathrm{T}}$ may be catalysed by $5$ T gats using Hamming weight by \cite{kivlichan2020improved}.
    }
    
    Note that the $\selectH$ has the Hadamard gates: each controlled Hadamard gives $2$ T gates and $1$ \textsc{cnot} gate. At least, we have $2 L + n_L$ Hadamard gates in $\selectH$.
    
    \item The other operations are symmetric.
\end{enumerate}

\sun{
In total, QSP requires the number of \textsc{cnot} gates
$$
d(2 + 2 +  6 \SC + 2 \PC + 2 \ST + 2 L)
$$
T gates
$$
d(7 \SC + 5 \ST + 2 \PT + 4 L  )
$$	
}

In the $5$th step, the additional cost is from the reflection from \autoref{observ:gs_reflect}.

The ground state preparation error is composed of two parts: the error from the approximation of the sign function, and the block-encoding error of the $\prep$ operation 
\begin{equation}
    \varepsilon = \varepsilon_{sgn} + \varepsilon_{tot, \prep}.
\end{equation}

Considering the $2 d$ repetition of the $\prep$ operation and the relation indicated by, we have
\begin{equation}
     \varepsilon_{\prep} = \frac{\varepsilon}{4d \lambda}
\end{equation}
we choose to set the amplitude encoding error as 
\begin{equation}
\varepsilon_{AE} \sim \frac{\lambda}{L} \varepsilon_{\prep}    =  \frac{ \varepsilon   }{ 4 L d }   
\label{eq:errorAE_QSP_Lin}
\end{equation}
and $n_{AE} = \lceil- \log_2  n_{AE}\rceil$.


A key component in \cite{lin2020near} is a polynomial approximation of the sign function in the domain $[-1,-\delta]\cup[\delta,1]$.
To derive the actual degree, we use an explicit construction of a polynomial with the same error scaling provided in \cite{low2018hamiltonian} based on the approximation of the $\textit{erf}$ function.


\begin{lemma}[Polynomial approximation to the sign function $\text{sgn}(x)$ (QETU overhead)~\cite{low2018hamiltonian}]\label{lem:approx_sign_overhead}

For any $\delta < 1$, $\varepsilon \leq \sqrt{2/\pi e}$, the polynomial $S(x, \delta, \varepsilon) =  p_{\text{sgn},\delta,n}(x) = p_{\text{erf},k,n}(x)  $ of odd degree $d=\left[ \frac{e}{2\delta} \ln (32{\pi^{-1/2}} \varepsilon^{-1}) \right]   = \mathcal{O}(\delta^{-1}\log{(\varepsilon^{-1})})$ satisfies
\begin{align}
\label{Eq:p_tilde_sgn_shifted}
 \varepsilon_{\text{sgn},\delta,n}&=\max_{x \in [-1,-\delta]\cup [\delta,1]}|p_{\text{sgn},\delta,n}(x)-\text{sgn}(x)| 
\le \varepsilon. 
\end{align}

\end{lemma}

The actual gate cost for the QETU method presented in this work is analysed by \autoref{lem:approx_sign_overhead}.

It is also interesting to note why QETU may not be able to run the real-time evolution in parallel when the qubit connectivity is restricted to nearest-neighbours.
The QETU method requires one ancillary qubit because of the phase iteration.
For certain Hamiltonians, if there exists a single Pauli operator $K_j$ such that it anticommutes with each component of $H$, then the evolution can be implemented in a control-free way. For example, for Heisenberg models, we can divide the Hamiltonian into three terms $H = H_X + H_Y + H_Z$ where $H_X$, $H_Y$, $H_Z$ contain tensor products of Pauli $X$, $Y$, $ Z$, respectively. For $H_X$, we can choose $K_1 = \otimes_{i\in \mathrm{odd}} X_{i} \otimes_{i \in \mathrm{even}}Y_{i} $.
However, the controlled-K operations still need depth $d_K = \mathcal{O} (n)$, as opposed to $d = \mc O (1)$ in our ancilla-free scheme. 
For electronic structure problems, to reduce the Trotter error and implement the operations in parallel, we can group the Hamiltonian into $\hat{T}$ and $\hat{V}$, in which case it is difficult to find a Pauli operation $K$ such that it commutes with the grouped term. If the Hamiltonian is not grouped, then the circuit depth will be increased to $d = \mathcal{O} (n^3)$ when implementing each individual term in a naive way.

\section{Ground state property estimation with phase estimation}
\label{sec:gate_gs_QPE}

\subsection{Complexity of phase estimation}

 For the canonical QPE algorithm, we apply a series of controlled $U$, $U^2$, ..., $U^{2^k-1}$ and an inverse quantum Fourier transform on the ancillary $k$ qubits, such that the state becomes 
\begin{equation}
\sum_i c_i \ket{0^{\otimes k} } \ket{E_i} \rightarrow \sum_i p_i \ket{\operatorname{bin} (E_i)} \ket{E_i}.
\end{equation}

To obtain a binary estimate of the energy precise to $n = \lceil \log_2 \varepsilon^{-1}\rceil$ bits, we require  $k =   \mathcal{ O}(  \log_2 \varepsilon^{-1}  + \log_2 \eta^{-1/2} ) $ 
ancillary qubits~\cite{Ge19}. The coherent runtime for each phase estimation is $ 2^{k+1}\pi = \mathcal{O} (\varepsilon^{-1} \eta^{-1/2} )$, and the number of calls to phase estimation is $\mathcal{O} (\eta^{-1/2})$. The total gate complexity is
\begin{equation}
\mathcal{ O} ({C}_{gate} \eta^{-1} \varepsilon^{-1} )
\end{equation}
where ${C}_{gate}$ is the gate cost within each segment.
 
To obtain a binary estimate of the energy precise to $n = \lceil \log_2 \varepsilon^{-1}\rceil$ bits, 
we require  $k =   \mathcal{ O}(  \log_2 \varepsilon^{-1}  + \log_2 \eta^{-1/2} ) $ 
ancillary qubits. 
The total error is composed of following parts: the error of phase estimation,   Hamiltonian simulation, and circuit synthesis.
$$
\varepsilon_{\operatorname{tot}} = \varepsilon_{PE} + \varepsilon_{HS} + \varepsilon_{CS}
$$
The coherent runtime for each phase estimation is lower bounded by
\begin{equation}
t_{PE}^{En} = \frac{\pi}{2  \eta \varepsilon_{PE} }
\label{eq:t_PE_En}
\end{equation}

In the following, we will discuss the Hamiltonian simulation

\begin{lemma}[Ground state preparation with phase estimation for known ground energy]
\label{prop:gs_scale_QPE}
Using the canonical phase estimation, the state can be prepared $\varepsilon$ close to ground state using $k = \mathcal{ O}(  \log_2 \varepsilon^{-1} +  \log_2 \Delta^{-1}  + \log_2 \eta^{-1/2} ) $. The runtime for each phase estimation is $  2^{k+1}\pi = \mathcal{O} (\varepsilon^{-1} \Delta^{-1} \eta^{-1/2} )$, and the number of calls to phase estimation is $\mathcal{O} (\eta^{-1/2})$ using fixed point search. 
The total  gate complexity is
$ 
\mathcal{O} (\frac{{C}_{gate}}{\eta^2 \Delta \varepsilon })
$.

\end{lemma}

The  coherent runtime for each phase estimation is lower bounded by 
\begin{equation}
t_{PE}^{prep} = \frac{\pi}{2   \Delta \varepsilon_{PE} }
\label{eq:t_PE_prep}
\end{equation}
 
The error of estimation of an observable $O$ consists of the following components:
$$
\varepsilon = \varepsilon_{PE} + \varepsilon_{HS} +  \varepsilon_{CS} +  \varepsilon_{observ}.
$$

The error of estimating observables using $N_s$ samples is given by
\begin{equation}
    \varepsilon_{observ} = \frac{C_{observ} }{\sqrt{ N_s }} 
\end{equation}
Suppose we use the importance sampling to estimate the observable $O = \sum_l o_l P_l$, the measurement overhead is $ C_{observ} = \|O\|_1 $, and we may use other methods to reduce $C_{observ}$, such as Pauli grouping or classical shadow methods to reduce the cost.
In our numerical simulation, we use an { optimistic estimation by only considering the dominant cost from the last operations $C-U^{2^{k-1}}$ only, and neglect the cost by controlled operations.
The gate cost for QFT is neglected as well, which scales as $\mathcal{O}(k)$.}


\subsection{Hamiltonian simulation by Trotterisation}

The overall circuit complexity for $2k$ order achives minima when $\varepsilon_{PE} = \varepsilon_{HS} = \varepsilon/2$
\begin{equation}
\frac{\left (\pi \cdot 5^{k-1} L \Lambda  \eta^{-1} \right )^{ 1 + \frac{1}{2k}}}{((2 k + 1) !)^{\frac{1}{2k}} \varepsilon_{PE}^{1+\frac{1}{2k}}\varepsilon_{HS}^{\frac{1}{2k}}}  
\end{equation}
    
Its minimum is obtained at $$\varepsilon_{PE} = \frac{2k+1}{2(k+1)} \varepsilon,~ \varepsilon_{HS} = \frac{1}{2(k+1)} \varepsilon$$   

Ground state energy estimation with phase estimation + higher-order Trotter:

\textbf{Gate count for eigenenergy estimation.}
\begin{enumerate}
    \item Get the runtime $t_{PE}/2$ with $\varepsilon_{PE} = \varepsilon/2$ in \autoref{eq:t_PE_En}.
    \item Determine the number of segment $\nu$ using \autoref{prop:bound_2k_stTrotter}. 
    \item \textsc{cnot} gates: $2\cdot 5^{k-1}\nu \eta^{-1/2} (2\operatorname{wt} (H)-L + 2)$.
    \item Single-qubit $Z$-axis Pauli rotation gate: $4 \cdot 5^{k-1} \nu \eta^{-1/2} L$.
\end{enumerate}


\textbf{Gate count for eigenstate property estimation.}
\begin{enumerate}
    \item Get the measurement overhead $C_{observ}$, runtime $t_{PE}/2$ as a function of $\varepsilon_{observ}$ and $\varepsilon_{PE}$, respectively. 
    \item Determine the number of segment $\nu$ using \autoref{prop:bound_2k_stTrotter} as a function of $\varepsilon_{HS}$. An approximation of $\nu$ is given by \autoref{eq:nu_approxi_2k_Trotter}.
    \item \textsc{cnot} gates: $2 \cdot 5^{k-1}\nu \eta^{-1/2} C_{observ} \varepsilon_{observ}^{-2} (2\operatorname{wt} (H)-L + 2)$.
    \item Single-qubit $Z$-axis Pauli rotation gate: $4\cdot 5^{k-1}\nu C_{observ} \varepsilon_{observ}^{-2} \eta^{-1/2} L$.
    \item Get the gate count by optimising over the distribution of $\varepsilon$.
\end{enumerate}


\subsection{Hamiltonian simulation by qubitised quantum walk}
The phase estimation combining the qubitisation methods has been discussed in \cite{babbush2018encoding}, which is compared to other algorithms. In the following, we will review the method before analysing the resource costs.
The key idea of qubitised quantum walk is that the spectrum of $H$ can be obtained by performing phase estimation on the Szegedy quantum walk operator, defined as 
\begin{equation}
    \mathcal{W} := (2\ket{G}\bra{G} \otimes I - I) \cdot \selectH
\end{equation}
with $\ket{G} = \prep \ket{\bar 0}$. 
The spectrum has a relation
$
    \operatorname{spectrum}(H) = \lambda \cos (\arg[\operatorname{spectrum}(\mathcal{W})])
 $
with $\arg(e^{i\phi}) = \phi$.
Their results suggest that we can estimate the phase to a number of bits given by
\begin{equation}
    k = \left\lceil \log \left( \frac{ \sqrt{2} \pi \lambda}{2 \varepsilon_{PE}} \right) \right \rceil 
\end{equation}
with $k$ extra ancillary qubits.
Here, we further assume a small error of gate synthesis in $\prep$ and $QFT$.
Using phase estimation, the query number is
\begin{equation}
d := 2^k \leq \frac{\sqrt{2} \pi \lambda}{2 \varepsilon_{PE}} + 1.
\end{equation}

The state preparation error $\varepsilon_{\prep}$ for a single application of $\mathcal{W}$ is 
\begin{equation}
 \varepsilon_{tot,\prep} \leq \| e^{i \arccos(H/\lambda)} - e^{i \arccos(\tilde H/\lambda)}\|
\end{equation}
It is related to the amplitude encoding error $\varepsilon_{AE}$   by
\begin{equation}
 \varepsilon_{\prep} \leq \frac{L \varepsilon_{AE}}{\lambda} \left( 1- \left( \frac{\|H\| + L \varepsilon_{AE}}{\lambda}\right)^2 \right)^{-1/2}
\end{equation}
Suppose that we require the preparation error to be $\varepsilon_{\prep}$. One can choose to set
\begin{equation}
\label{eq:epsilon_AE}
\varepsilon_{AE} = \frac{ \varepsilon_{\prep}}{ (1+\varepsilon_{\prep}^2)L}  
  \end{equation}
 assuming that  $\Omega (\frac{\|H\|}{\lambda}) = 0$.  
The preparation error is set to be 
\begin{equation}
    \varepsilon_{\prep} \leq  \frac{\sqrt{2}}{{2}\lambda } \frac{\varepsilon_{PE}}{  2^{k} } = \frac{\sqrt{2} \varepsilon_{PE}}{ \lambda 2^{k+1} }
\end{equation}
and hence for a single block
\begin{equation}
    \varepsilon_{AE} = \frac{\sqrt{2} \varepsilon}{4 L \lambda d}.
\end{equation}
Note that $2$ \textsc{prepare} is used in one block.
By \autoref{eq:epsilon_AE}, we can determine the amplitude encoding error $\varepsilon_{AE}$ as
$ 
 \varepsilon_{AE}  = \frac{\varepsilon_{PE}^2}{ \pi L \lambda} 
$.
Again, the cost from the QFT is ignored, which scales as $\mathcal{O} (k \log k)$.

The overall gate complexity of the eigenenergy estimation is 
$
\mathcal{O} (\frac{\lambda L}{ \varepsilon })
 $. The total gate count can be estimated by using \autoref{observ:gs_amp_encode},   \autoref{observ:gs_reflect}, and the Hamiltonian dependent $\selectH$, given by  \autoref{observ:gs_select_Ising}.

Each block requires: 1 controlled $\selectH$, 2 $\prep$ and 1 Reflection on $n_L$ qubits, which has the gate count
\begin{equation}
    ( n_L + \max( k, n_L+2n_{AE}+1 ), \SC + 2\PC + 6(n_L - 2), \ST + 2\PT + (8 n_L -17), 0 ).
\end{equation}
and an additional $k$ repetition of controlled reflection, each block (2 preparation) has the cost:
\begin{equation}
    (0, 2\PC + 6(n_L - 1), 2\PT + (8 n_L - 9),0)
\end{equation}

In total $d$ queries and $\eta^{-1} \Delta ^{-1}$ repetitions are required, which results in
\begin{equation}
\begin{aligned}
( n_L + \max( k, ( n_L+2n_{AE}+1 ), & \eta^{-1} \Delta^{-1}\left(d(\SC + 2\PC+6(n_L - 2)) + k(2\PC + 6(n_L - 1))\right),  \\
&\eta^{-1} \Delta^{-1} \left( d(\ST + 2\PT + (8 n_L -17)) + k(2\PT + (8 n_L - 9)) \right), 0 ).  
\end{aligned}
\end{equation}
 

\section{Investigation on the resource cost}
\label{sec:numerics_SM}

\subsection{Numerical setting}

In this section, we first present the details of the numerical simulation. We will also present additional resource estimation results. 

To estimate the gate costs of the RLCU algorithm, in this work we set the normalisation factor $\mu=2$ to ensure that the sample complexity of the RLCU algorithm is similar to other quantum algorithms. 
The  operations involved in the algorithms  are \textsc{cnot} gates and single-qubit Pauli rotation gates $R_z (\theta
)$, which is further decomposed into T gates. The circuit compilation overhead is detailed in Supplementary \autoref{sec:gate_gs_BE}.

\textbf{Energy gap fitting}
For the Heisenberg type of Hamiltonian with the additional field on the boundary in \autoref{eq:Hamil_XXZ} in the main text, when $h_x = 0$, it has a constant gap $\Delta(c) = 4(c-1)$ when the system size is infinite.  The energy degeneracy is $n+1$. 
For example, for $c = 2$, it has a constant energy gap $\Delta = 4$. 

When the external field $h_x$ increases, more excited states will emerge. However, we find by numerical fitting that for $n \leq 100$, the energy gap is not very small. 
We find that the energy gap can be well fitted by a polynomial function. In \autoref{fig:gap_fitting_SI}, we show the fitting by $\Delta = b \cdot n^a$ with $a = -0.50$, which agrees quite well with the actual gap at small system sizes. 
In contrast, the energy gap fitting by an exponential function $\Delta = b\exp(-{n^a})$ does not agree well. \revise{The gap dependence has been considered in our resource estimation.}

\revise{\textbf{Remark on the initial state}.
As the central objective is to reduce the depth, in resource estimation, we mainly focus on the maximum gate count in a single run, whose scaling is logarithmic in initial state overlap $\eta$. In other words, this maximum gate count that needs to be implemented coherently is nearly independent of $\eta$. In contrast, the standard phase-estimation procedure will have a worse gate complexity dependent on $\eta$, $\mathcal{O}( \eta^{-1} \epsilon^{-1})$.
To make a fair comparison with other methods, we set the initial state overlap to be a constant value as similarly used in \cite{ding2023even,wan2021randomized}. For instance,  $\eta$ is set to be 1  for quantum chemistry example FeMoco in \cite{wan2021randomized}, and  $\eta = 0.8$ for Ising models in \cite{ding2023even}.  In \cite{von2021quantum} they used DMRG to find the ground state with $\eta$ around $0.9$. In \cite{lin2021heisenberg} they have used a Hartree-Fock state for a 8-site Fermi-Hubbard model with $\eta$ around $0.4$. 
On the other hand, we have highlighted in the paper that this will mostly affect the sampling numbers, the sampling complexity has a similar dependence (in Theorems 5 and 7) to other selected randomised works whose practical performance is good  (i.e. with few actual gate counts). We also note that the sampling overhead due to randomisation is included in all the plots in the resource estimation to ensure that this method is compared with other deterministic algorithms like QSP at the same level.}

\revise{ 
In a way, initial state preparation can be improved by using methods like adiabatic evolution and dissipative method as well as various physics-inspired or MPS-based methods. 
The paper suggested by the referee is very helpful in this context. The key is to employ MPS techniques for state preparation.
This work focuses on how to obtain suitable initial states and is fully compatible with our approach, i.e., one can use the MPS-based strategy to prepare an appropriate initial state, and then apply the method developed in our work to estimate eigenenergies and eigenstate properties with high precision.  }


The requirements for the gate number with the selected advanced algorithms are estimated. Since the central objective is for the application in the early FTQC or NISQ era,  we mainly focus on the circuit depth in a single-run experiment. Therefore,  the amplitude amplification is not considered in the algorithms, which can deterministically prepare the state closer to the true ground state yet at the cost of a deeper circuit. 
When we consider a real physical model, such as quantum chemistry problems, the coefficient of the Hamiltonian is constructed by calculating the integral and represents the feature of the quantum system. Due to finite precision, there will be an amplitude encoding error when we perform the $\prep$ operation. 
To ensure that the amplitude encoding error in the block encoding procedure is less than a threshold, we require more qubits to encode the coefficient.
However, for the toy models, the absolute value of the coefficient may not be essentially relevant for the actual physics. For instance,  we can manually set the interaction strength when we study the phase transition.
In this work, we include the amplitude encoding error in our analysis when aiming for a realistic application. That is why we require more qubits for the algorithm involving amplitude amplification.

Note that $R_{z}(\theta)$ can be virtually implemented with real physical devices. For superconducting devices, $R_z(\theta)$ indeed does not have to be implemented physically, but rather it can be implemented by changing the phase of the reference frame defined by the multi-level rotating frame. That is, $R_z(\theta)$ is a virtual gate, and therefore, there is no physical error in implementing $R_z(\theta)$.

\begin{figure*}[t!]
\centering
\includegraphics[width=0.78 \textwidth]{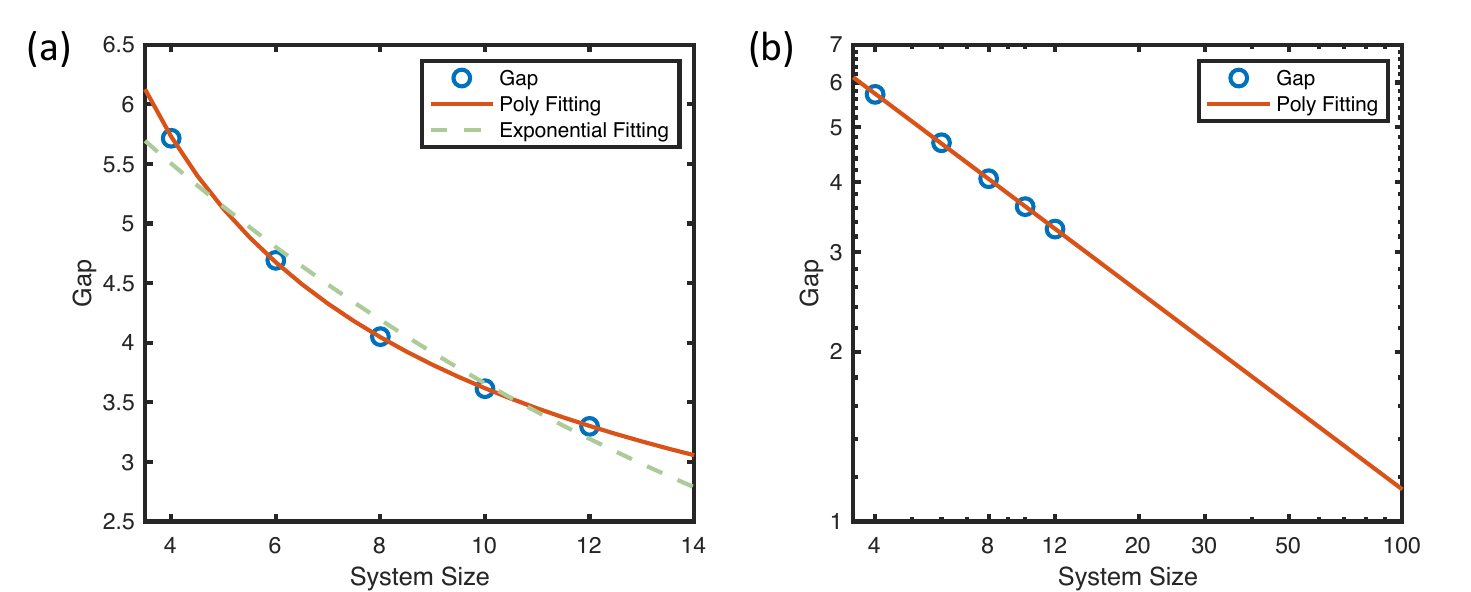}
\caption{ \textbf{Gap dependence for the antiferromagnetic Heisenberg model.} (a) Exact gap dependence with increasing system sizes. The polynomial fitting for the gap works better than the exponential fitting. (b) Gap dependence up to $100$ qubits by fitting.
}
\label{fig:gap_fitting_SI}
\end{figure*}

\begin{figure*}[ht!]
\centering
\includegraphics[width=1\textwidth]{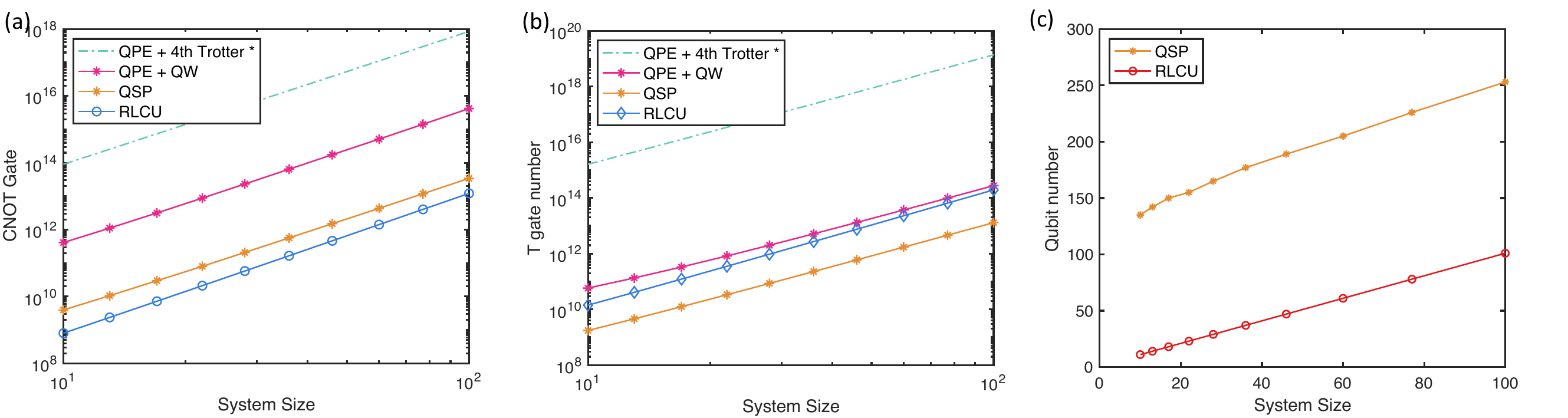}
\caption{ \textbf{
Resource estimation (the number of \textsc{cnot} gates and T gates) for the eigenstate property estimation for the 2-local Hamiltonian.} The Hamiltonian is
$ 
H = \sum_{i,j} X_i X_j + \sum_i Z_i
$, in which case the commutator relation between different Hamiltonian summands is ignored. Here the $ 4$th-order random Trotter formula is used as it performs the best over other orders, which is marked by an asterisk alongside Trotter in the legend.
We compare the gate count involved in different methods. Note that the commutator relation of the Hamiltonian is not taken into account which results in a higher gate count for our method.
}
\label{fig:SysSize_2local}
\end{figure*}

\begin{figure}[t!]
\centering
\includegraphics[width=0.8\linewidth]{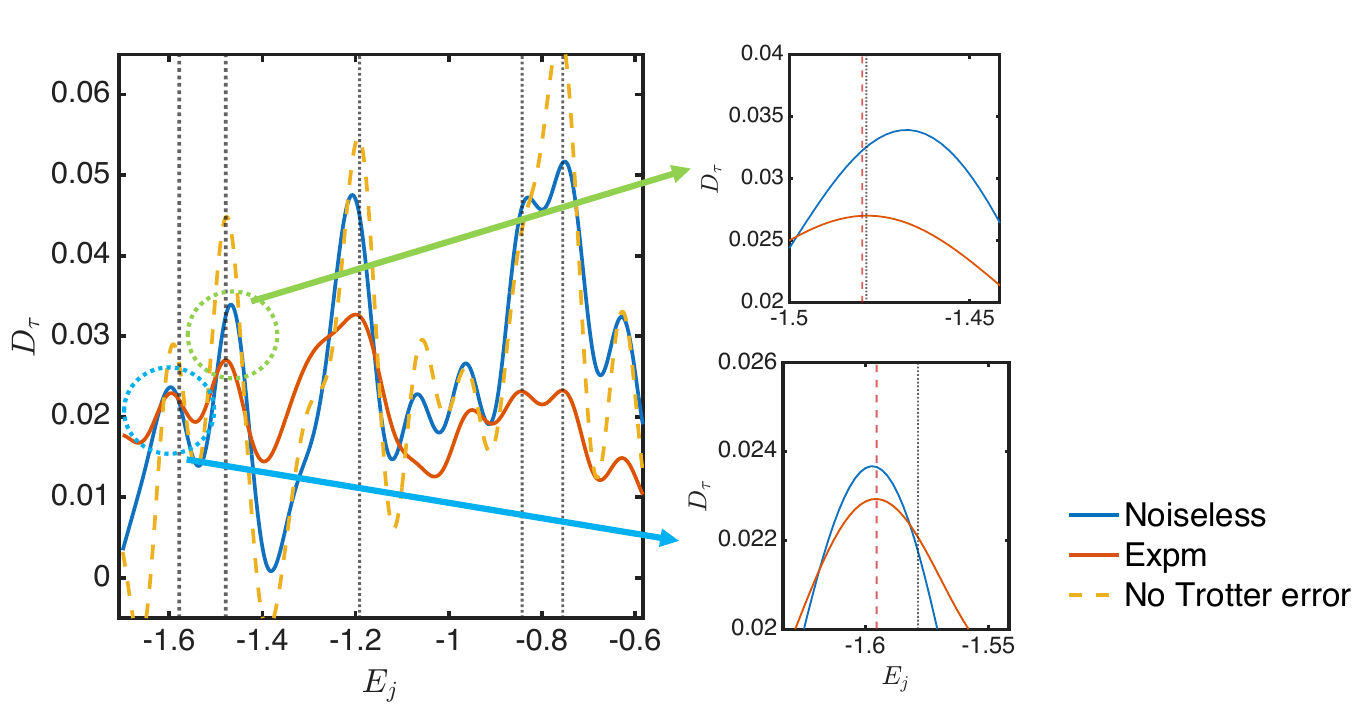}
\caption{\textbf{
Searching ground state and first excited state energies of Heisenberg Hamiltonians on IBM quantum devices.}
(a) We consider a 12-qubit normalised anisotropic Heisenberg Hamiltonian (see \autoref{eq:Hamil_XXZ} in the main text), without any external field. All parameter settings are consistent with those in the main text. This represents another circuit instance, using 1024 measurement shots. The figure on the right provides a zoomed-in view of a narrower energy range, highlighting the estimated ground-state and first excited-state energies shown in the left panel.
}
\label{fig:IBM_SI_2}
\end{figure}

\subsection{Additional resource estimates}

The results for T gates are shown in \autoref{fig:Heis_NC_Eps}(a), (b), and (c), corresponding to the tasks with \textsc{cnot} results in \autoref{fig:SysSize_NC}. \autoref{fig:Heis_NC_Eps}(d) presents the results when energy gap dependence is taken into account.
In the resource estimates, we choose to use the compensate the Trotter error up to $2k$th-order. As a result this has a polynomial scaling with $\mathcal{O}(\varepsilon^{-\frac{1}{4k+1}})$.  This is easy for the sampling process. Even with this conservative estimation, our method outperforms others as shown in \autoref{fig:Heis_NC_Eps}(a). The actual performance may be better in practice.

For general Hamiltonians without considering commutation relations, our method may require more T gates than QSP because our method requires $\Rz$ gates, which have a large overhead when they are synthesised into T gates. 
To verify this point, we consider the two-local Hamiltonian.
The results are shown in \autoref{fig:SysSize_2local}.


\subsection{ Implementation on IBM quantum cloud}

In the main text, we consider normalised anisotropic Hamiltonians with parameters $J_x = 1.05$, $J_y = 1$, $J_z = 0.7$ and $h_z  = 0.2$ in \autoref{fig:IBM}. The settings are $\tau = 2.5$ and $x_{c} = 2$. Another experiment instance with fewer samples ($N_s = 1024$ for each circuit run) is shown in \autoref{fig:IBM_SI_2}. 

\revise{
The experimental results are aligned well with the ideal results.
In the following, we give a few comments on the observed noise resilience in practice. 
First of all, the observed noise resilience can be understood from the sampling structure of our algorithm. Circuits with longer evolution times are deeper and thus more affected by noise, but they are sampled much less frequently because the designed probability distribution. To be concrete, as in Fig. 1(a), both the time $t_i$ and the circuit instance $\vec{i}$ according to their probability distribution in Eq. (1) and Eq. (13). That is, the time is sampled according to non-uniform distribution $\Pr(t_i)$ which quickly decay with time (which ensure the time complexity of our algorithm is small). Therefore, circuit instances with long time $t_i$ (hence more noisy) contribute smaller on the final eigenstate property estimation than the short time ones, i.e, the impact of noise which is more serious in deep circuits is suppressed.}

\revise{
In addition, Hamiltonian simulation may exhibit a certain level of intrinsic noise resilience. As discussed in recent studies, random errors in quantum circuits tend to show concentration behaviour, which brings smaller errors than symmetric errors, meaning that their cumulative effect averages out rather than accumulates coherently. See, for example, the numerical tests in Fig. 1 in \cite{cai2023stochastic}. This finding implies that noise may be suppressed with our circuit design.  In the intermediate-scale simulation regime, it is interesting to explore whether we can have some error concentration effect so that in this type of quantum algorithm the performance is good. }

\end{document}